\documentclass[a4paper,11pt]{article}
\usepackage[utf8]{inputenc}

\usepackage{amsmath}
\usepackage{amssymb}
\usepackage{amsthm} 	%
\usepackage{amsfonts}
\usepackage{mathrsfs}

\usepackage{caption}
\usepackage{lmodern}    %
\usepackage{microtype}  %

\usepackage{enumerate}  %
\usepackage{graphicx}   %

\usepackage{centernot}  %

\usepackage{mathrsfs}
\DeclareMathAlphabet{\mathscrbf}{OMS}{mdugm}{b}{n}

\usepackage[usenames,dvipsnames]{xcolor} %
\definecolor{nicerblue}{hsb}{0.6,1,0.75}
\definecolor{nicerred}{HTML}{cc0000}
\definecolor{nicerbrown}{HTML}{a23000}

\usepackage[bookmarks=true,
bookmarksnumbered=true,
pagebackref=false,      %
colorlinks=true, 
linkcolor=nicerred,     %
citecolor=nicerblue,    %
urlcolor=nicerbrown]{hyperref}

\newcommand{\mystyle}{\sffamily}

\usepackage{sectsty}

\allsectionsfont{\mystyle}

\renewenvironment{abstract}{%
	\small%
	\begin{center}%
		{\bfseries \mystyle \abstractname\vspace{-.5em}}%
	\end{center}%
	\quotation
}

\usepackage{authblk}

\usepackage{etoolbox}

\makeatletter
\patchcmd{\@maketitle}{\LARGE}{\huge\bfseries\textsf}{\typeout{OK 1}}{\typeout{Failed 1}}
\patchcmd{\@maketitle}{\large \lineskip}{\Large \lineskip}{\typeout{OK 2}}{\typeout{Failed 2}}
\makeatother

\newcommand{\orcidid}[1]{\href{https://orcid.org/#1}{\includegraphics[scale=0.09]{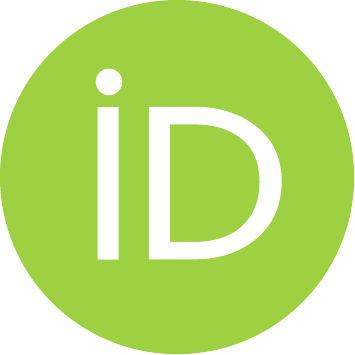}}}

\def\keywordname{\textcolor{darkgray}{{\bfseries \emph{Keywords}}}}%
\def\keywords#1{\medskip\par\addvspace\bigskipamount{\rightskip=0pt plus1cm
		\def\and{\ifhmode\unskip\nobreak\fi\ $\cdot$
		}\noindent\keywordname\enspace\ignorespaces#1\par}\medskip}

\def\prevversionname{\textcolor{darkgray}{{\bfseries \emph{Previous Versions}}}}%
\def\prevversion#1{\medskip\par\addvspace\bigskipamount{\rightskip=0pt plus1cm
		\noindent\prevversionname\enspace\ignorespaces#1\par}\medskip}

\usepackage[]{geometry}
\AtBeginDocument{
	\newgeometry{
		textwidth=6.25in,
		top=1.25in,
		headheight=14pt,
		headsep=25pt,
		footskip=30pt
	}
}    

\newcommand{\headeright}{Jan-Hendrik Lorenz \& Florian Wörz}
\newcommand{\shorttitle}{Towards an Understanding of Long-Tailed Runtimes}

\usepackage{fancyhdr}
\fancyheadoffset{0pt}
\newtheorem{theorem}{Theorem}
\newtheorem*{theorem*}{Theorem}
\newtheorem{lemma}[theorem]{Lemma}
\newtheorem{corollary}[theorem]{Corollary}
\newtheorem{proposition}[theorem]{Proposition}

\newtheorem{conjecture}[theorem]{Conjecture}
\theoremstyle{definition}
\newtheorem{definition}[theorem]{Definition}
\newtheorem{example}[theorem]{Example}

\usepackage[utf8]{inputenc}

\usepackage{tikz}
\usepackage{pgfplots}
\usepackage{pgffor}
\usepackage{pgfplotstable}
\usetikzlibrary{patterns,snakes,arrows,shapes,shapes.arrows}
\pgfplotsset{compat=newest}

\usepackage[boxed]{algorithm2e}

\SetAlCapSty{xAlCapSty}

\SetAlCapNameSty{kleinererText}

\usepackage{subcaption}

\usepackage{siunitx}
\usepackage{booktabs}
\usepackage{float}

\usepackage{xspace}
\usepackage{scalerel}

\usepackage{bbm}
\usepackage{mathtools}
\usepackage{doi}

\usepackage{todonotes}

\usepackage{mathtools}

\newcommand{\languageformat}[1]{\mathrm{#1}}

\newcommand{\ZVcardhash}[1]{\ZV{\#} #1}

\newcommand{\UNSATunder}[2][]{\ensuremath{\languageformat{UNSAT}_{#1}(#2)}}
\newcommand{\ZVUNSATunder}[2][]{\ensuremath{\ZV{\languageformat{UNSAT}}_{#1}(#2)}}

\newcommand{\UNSATin}[3][]{\ensuremath{\languageformat{UNSAT}_{#1}(#2,\in \! #3)}}
\newcommand{\ZVUNSATin}[3][]{\ensuremath{\ZV{\languageformat{UNSAT}}_{#1}(#2,\in \! #3)}}

\newcommand{\UNSATnotin}[3][]{\ensuremath{\languageformat{UNSAT}_{#1}(#2,\notin \! #3)}}
\newcommand{\ZVUNSATnotin}[3][]{\ensuremath{\ZV{\languageformat{UNSAT}}_{#1}(#2,\notin \! #3)}}

\DeclareMathOperator{\Varsop}{Vars}
\newcommand{\Vars}[1]{\Varsop \! \big( #1 \big)}

\newcommand{\twopartdef}[4]
{
	\left\{
	\begin{array}{ll}
		#1 & \mbox{if } #2, \\
		#3 & #4
	\end{array}
	\right.
}

\newcommand{\ListCommas}[1]{
	\foreach \x [count=\ni] in {#1}{
		\ifnum\ni=1%
		\x%
		\else%
		,\,\x%
		\fi%
	}
}

\newcommand{\ZV}[1]{\ensuremath{\boldsymbol{#1}}}

\DeclareMathOperator{\Probop}{\mathbb{P}}%
\DeclareMathOperator{\ProbopZV}{\fakebold{\mathbb{P}}\!}%
\DeclareMathOperator{\Expop}{\mathbb{E}}

\newcommand{\prob}[2][]{\Probop_{#1} [ #2 ]}
\newcommand{\Prob}[2][]{\Probop_{#1} \bigl[ #2 \bigr]}
\newcommand{\PROB}[2][]{\Probop_{#1} \left[ #2 \right]}

\newcommand{\probcondZV}[3][]{\ProbopZV{(\ListCommas{#2} \!\mid\!  \ListCommas{#3})}}
\newcommand{\probcond}[3][]{\prob[#1]{\ListCommas{#2} \!\mid\!  \ListCommas{#3}}}
\newcommand{\Probcond}[3][]{\Prob[#1]{\ListCommas{#2} \stretchleftright{\mid}{\vphantom{\ListCommas{#2}\ListCommas{#3}}}{.}\,  \ListCommas{#3}}}
\newcommand{\probcondResize}[3][]{\PROB[#1]{\ListCommas{#2} \stretchleftright{\mid}{\vphantom{\ListCommas{#2}\ListCommas{#3}}}{.}\,  \ListCommas{#3}}}

\newcommand{\expectation}[2][]{\Expop_{#1} [ #2 ]}

\newcommand{\EXPECTATION}[2][]{\Expop_{#1} \left[ #2 \right]}
\newcommand{\expectationzv}[2][]{\fakebold{\mathbb{E}}_{#1} ( #2 )}
\newcommand{\expectationZV}[2][]{\fakebold{\mathbb{E}}_{#1} \bigl( #2 \bigr)}

\newcommand{\expectationcond}[3][]{\expectation[#1]{#2 \!\mid\! #3}}
\newcommand{\expectationcondzv}[3][]{\expectationzv[#1]{#2 \!\mid\! #3}}
\newcommand{\expectationcondZV}[3][]{\expectationZV[#1]{#2 \!\mid\! #3}}

\newcommand{\samplespace}{\Omega}
\newcommand{\sigmaalgebra}{\mathcal{F}}
\newcommand{\probspace}{(\samplespace, \sigmaalgebra, \Probop)}

\newcommand{\singleevent}[1]{\set{#1}}

\newcommand{\distr}{\sim}

\newcommand{\indistribution}[1]{\underset{#1}{\overset{d}{\longrightarrow}}}
\newcommand{\inprobability}[1]{\underset{#1}{\overset{p}{\longrightarrow}}}

\newcommand{\approxd}{\stackrel{d}{\approx}}

\newcommand{\SBop}{\operatorname{SB}}
\newcommand{\LogNop}{\operatorname{LogN}}

\newcommand{\Bern}[1]{{\operatorname{Bern} ( #1 )}}

\newcommand{\Bin}[2]{{\operatorname{Bin} \! \big( #1, \, #2 \big)}}

\newcommand{\Normal}[2]{{\operatorname{N} \left( #1, #2 \right)}}
\newcommand{\LogN}[2]{{\LogNop \left( #1, #2 \right)}}

\newcommand{\SB}[4]{{\SBop \left( #1, #2, #3, #4 \right)}}
\newcommand{\SBwithequations}[4]{{\SBop \left( #1, \, #2, \, #3, \, #4 \right)}}

\newcommand{\bigohsymbol}{\mathrm{O}}
\newcommand{\bigoh}[1]{\bigohsymbol ( #1 )}
\newcommand{\bigohp}[1]{\bigohsymbol_p ( #1 )}

\newcommand{\clstd}{\ensuremath{K}}

\newcommand{\litstd}{\ensuremath{u}}

\newcommand{\applyassi}[2]{#1#2}

\newcommand{\falsassi}[2]{#1: \, \applyassi{#2}{#1} = 0}

\DeclareMathOperator{\E}{ER}
\DeclareMathOperator{\W}{\textnormal{\textbf{\#}}\textnormal{\textbf{Flips}}} %
\DeclareMathOperator{\F}{Flip}

\newcommand{\assi}{\{0,1\}^n}

\newcommand{\ERL}{\ZV{\E_{\Clauses}}}

\newcommand{\econd}[1]{\ensuremath{\ZV{\E_{\Clauses}}{( \ListCommas{#1} )}}}
\newcommand{\ECond}[1]{\ensuremath{\ZV{\E_{\Clauses}}{\big( \ListCommas{#1} \big)}}}
\newcommand{\EACond}[1]{\ECond{#1, \InitAssi{\alpha}}}

\newcommand{\Pbig}[1]{\ensuremath{\ZV{\ProbopZV_{\Clauses}} \big( \ListCommas{#1} \big)}}

\newcommand{\PCond}[2]{\ensuremath{\ZV{\ProbopZV_{\Clauses}} \big( \ListCommas{#1} \mid \ListCommas{#2} \big)}}
\newcommand{\PCondNoL}[2]{\ensuremath{\ZV{\ProbopZV} \left( \ListCommas{#1} \!\mid\! \ListCommas{#2} \right)}}
\newcommand{\PCondResize}[2]{\ensuremath{\ZV{\ProbopZV_{\Clauses}} \left( \ListCommas{#1} \;\stretchleftright{\mid}{\vphantom{\ListCommas{#1}\ListCommas{#2}}}{.}\; \ListCommas{#2} \right)}}

\makeatletter
\newcommand{\PACond}[2]{
	\def\tmp{#2}%
	\ifx\tmp\empty%
	\PCond{#1}{\InitAssi{\alpha}}%
	\else%
	\PCond{#1}{#2, \InitAssi{\alpha}}%
	\fi%
}
\newcommand{\PGammaCond}[2]{
	\def\tmp{#2}%
	\ifx\tmp\empty%
	\PCond{#1}{\InitAssi{\alpha}}%
	\else%
	\PCond{#1}{#2, \InitAssi{\alpha}}%
	\fi%
}
\newcommand{\PACondResize}[2]{
	\def\tmp{#2}%
	\ifx\tmp\empty%
	\PCondResize{#1}{\InitAssi{\alpha}}%
	\else%
	\PCondResize{#1}{#2, \InitAssi{\alpha}}%
	\fi%
}
\newcommand{\PGammaCondResize}[2]{
	\def\tmp{#2}%
	\ifx\tmp\empty%
	\PCondResize{#1}{\InitAssi{\gamma}}%
	\else%
	\PCondResize{#1}{#2, \InitAssi{\gamma}}%
	\fi%
}
\makeatother

\newcommand{\Wa}[3][\alpha]{\ensuremath{\mathrm{From}_{#1}\mathrm{To}_{#3}\mathrm{In}_{#2}}}

\newcommand{\Flip}[3]{\Flipa{#1}{#2}{#3}}
\newcommand{\Flipa}[3]{\ensuremath{\F_{\alpha}\!\left(#1, #2, #3\right)}}

\newcommand{\NoSel}[1]{\ensuremath{\operatorname{FirstSel}(#1+1)}}
\newcommand{\NeverSel}{\ensuremath{\mathrm{NeverSel}}}%

\newcommand{\InitAssi}[1]{\ensuremath{A(#1)}}

\newcommand{\FlipVarInAssi}[2]{#2[#1]}

\newcommand{\Sel}[1]{\ensuremath{\mathrm{Sel}(#1)}}

\newcommand{\InfCase}{\mathcal{I}}
\newcommand{\FiniteCase}{\mathcal{F}}
\newcommand{\Expression}{\mathcal{E}}
\newcommand{\DExprOne}{\mathcal{D}_1}
\newcommand{\DExprTwo}{\mathcal{D}_2}

\usepackage{mathrsfs}
\usepackage{stackrel}

\usepackage{xurl}

\newcommand{\abstractlink}[3][blue]{\hyperlink{#2}{\color{#1}{#3}}}%

\title{Towards an Understanding of \\ Long-Tailed Runtimes of SLS Algorithms}

\author{Jan-Hendrik Lorenz \orcidid{0000-0002-9554-4347}}
\author{Florian Wörz \orcidid{0000-0003-2463-8167}}
\affil{Universit\"at Ulm, Germany\\\texttt{jan-hendrik.lorenz@alumni.uni-ulm.de}, \texttt{florian.woerz@uni-ulm.de}}

\date{\today}%

\begin{document}

	\newlength\bshft
	\bshft=.18pt\relax
	\def\fakebold#1{\ThisStyle{\ooalign{$\SavedStyle#1$\cr%
				\kern-\bshft$\SavedStyle#1$\cr%
				\kern\bshft$\SavedStyle#1$}}}

	\newcommand{\bftab}{\fontseries{b}\selectfont}

	\newcommand{\ie}{i.\,e.,\ }
	\newcommand{\eg}{e.\,g.,\ }
	\newcommand{\Eg}{E.\,g.,\ }
	\newcommand{\egcite}{e.\,g.,}
	\newcommand{\wrt}{with respect to }%
	\newcommand{\etal}{et al.\ }
	
	\newcommand{\Ie}{I.\,e.,\ }
	\newcommand{\stabbrev}{s.\,t.\ }
	\newcommand{\ia}{i.\,a.\ }

	\newcommand{\pa}{par2}
	\newcommand{\score}{score}
	\newcommand{\bwuni}{bwUniCluster}
	
	\newcommand{\iid}{i.\,i.\,d.\ }

	\renewcommand{\N}{\mathbb{N}}
	\newcommand{\R}{\mathbb{R}}
	\newcommand{\Rpos}{\R^{+}}
	\newcommand{\Rplus}{\Rpos}
	
	\newcommand{\e}{\mathrm{e}}
	\newcommand{\dd}{\mathrm{d}}
	\newcommand{\Indikator}[1]{\mathbbm{1}_{#1}}
	
	\newcommand{\D}{\textnormal{d}}
	
	\newcommand{\diff}[1]{\frac{\D}{\D #1}\,}
	\newcommand{\nat}[1]{[#1]}

	\newcommand{\introduceterm}[1]{{\emph{#1}}}

	\newcommand{\refChap}[1]{Section~\ref{#1}}%
	
	\newcommand{\refsec}[1]{Section~\ref{#1}}
	\newcommand{\refSec}[1]{Section~\ref{#1}}
	\newcommand{\Refsec}[1]{Section~\ref{#1}}
	\newcommand{\refsecP}[1]{Section~\vref{#1}}
	\newcommand{\reftwosecs}[2]{Sections~\ref{#1} and~\ref{#2}}
	\newcommand{\Reftwosecs}[2]{Sections~\ref{#1} and~\ref{#2}}
	
	\newcommand{\refsubsec}[1]{Subsection~\ref{#1}}
	
	\newcommand{\figuretext}{Figure}
	\newcommand{\reffig}[1]{\figuretext~\ref{#1}} %
	\newcommand{\Reffig}[1]{\figuretext~\ref{#1}} %
	\newcommand{\refFig}[1]{\figuretext~\ref{#1}}
	\newcommand{\reffigP}[1]{\figuretext~\vref{#1}}
	
	\newcommand{\reftab}[1]{Table~\ref{#1}}
	\newcommand{\refTab}[1]{Table~\ref{#1}}
	\newcommand{\Reftab}[1]{Table~\ref{#1}}
	
	\newcommand{\refDef}[1]{Definition~\ref{#1}}
	\newcommand{\refLem}[1]{Lemma~\ref{#1}}
	\newcommand{\refCor}[1]{Corollary~\ref{#1}}
	\newcommand{\refTheo}[1]{Theorem~\ref{#1}}
	\newcommand{\refCon}[1]{Conjecture~\ref{#1}}
	\newcommand{\refRecap}[1]{Proposition~\ref{#1}}
	\newcommand{\refRecapThm}[1]{Theorem~\ref{#1}}
	
	\newcommand{\refalg}[1]{Algorithm~\ref{#1}}
	\newcommand{\refAlg}[1]{Algorithm~\ref{#1}}
	
	\newcommand{\refEq}[1]{Equation~\eqref{#1}}
	\newcommand{\refIneq}[1]{Inequality~\eqref{#1}}
	
	\newcommand{\refApp}[1]{Appendix~\ref{#1}}

	\newcommand{\algoformat}[1]{\texttt{#1}}
	
	\newcommand{\OurSolver}{\algoformat{GapSAT}} %
	\newcommand{\sr}{\algoformat{Sparrow2Riss}}
	\newcommand{\glu}{\algoformat{Glucose}}
	
	\newcommand{\dimetheus}{\algoformat{dimetheus}}

	\newcommand{\setsofvarsorlit}[2]%
	{\mathrm{#1}({#2})}
	
	\newcommand{\vars}[1]{\setsofvarsorlit{Vars}{#1}}
	
	\newcommand{\dom}[1]{\setsofvarsorlit{Dom}{#1}}

	\newcommand{\tvastd}{{\ensuremath{\alpha}}}
	
	\newcommand{\restrict}[2]{#1\!\!\upharpoonright_{#2}}

	\newcommand{\set}[1]{\{ #1 \}}
	\newcommand{\Set}[1]{\big\{ #1 \big\}}
	\newcommand{\SET}[1]{\left\{ #1 \right\}}
	
	\newcommand{\setdescr}[3][\mid]{\set{ #2 #1 #3 }}
	\newcommand{\Setdescr}[3][\bigm\vert]{\Set{ #2 #1 #3 }}
	
	\newcommand{\twincommand}[6]%
	{#1#2#3\vphantom{#2#5}\mspace{-2.05mu}#4.#5#6}
	
	\newcommand{\SETDESCR}[3][|]%
	{\twincommand{\left\{}{#2\,}{\left#1}{\right}{\,#3}{\right\}}}

	\newcommand{\setsize}[1]{\lvert#1\rvert}
	\newcommand{\Setsize}[1]{\bigl\lvert#1\bigr\rvert}
	\newcommand{\SETSIZE}[1]{\left\lvert#1\right\rvert}
	
	\newcommand{\cardhash}[1]{\# #1}

	\newcommand{\problemlanguageformat}[1]{\textrm{#1}\xspace}
	\newcommand{\SAT}{\ensuremath{\problemlanguageformat{SAT}}}
	
	\newcommand{\complexityclassformat}[1]{\textrm{\upshape{\textsf{#1}}}\xspace}
	\newcommand{\NP}{\ensuremath{\complexityclassformat{NP}}}

	\newcommand{\integral}[4][x]{\int\limits_{#2}^{#3}#4 \, \D #1}
	\newcommand{\integralLow}[4][x]{\int_{#2}^{#3}#4 \, \D #1}

	\newcommand{\Clauses}{\ensuremath{L}}
	\newcommand{\Alfa}{\algoformat{Alfa}}
	\newcommand{\alfa}{\Alfa}
	\newcommand{\probSAT}{\algoformat{probSAT}}
	\newcommand{\SRWA}{\algoformat{SRWA}}
	\newcommand{\YalSAT}{\algoformat{YalSAT}}
	\newcommand{\YAL}{\YalSAT}
	\newcommand{\GapSAT}{\algoformat{GapSAT}}
	
	\newcommand{\Lim}[1]{\lim_{#1 \rightarrow \infty}}
	\newcommand{\LIM}[1]{\lim\limits_{#1 \rightarrow \infty}}
	
	\newcommand{\quantorsep}{\,\,}

	\newcommand{\inputtikz}[1]{%
		\tikzsetnextfilename{#1}%
		\input{#1.tex}%
	}
	
	\newcommand{\todoin}[1]{\tikzexternaldisable\todo[inline]{#1}\tikzexternalenable}
	\newcommand{\talk}[1]{\tikzexternaldisable\todo[inline, color=green!40]{#1}\tikzexternalenable}

	\newcommand{\new}[1]{#1}%
	\newcommand{\jan}[1]{#1}%
	
	\newcommand{\showjan}[1]{\textcolor{cyan}{#1}}

	\definecolor{glossarygray}{HTML}{4f4f54}
	\newcommand{\gloss}[1]{\textcolor{glossarygray}{#1}}

	{
		\theoremstyle{theorem}

		\newtheorem{recapitulation}[theorem]{Proposition}
		\newtheorem{recapitulationthm}[theorem]{Theorem}
		
		\newtheorem{definitionandtheorem}[theorem]{Definition and Theorem}
		
		\newtheorem*{maintheorem}{Main Theorem}
		\newtheorem*{mainresult}{Main Result}
		
		\newtheorem{implication}[theorem]{Implication}
		\newtheorem{assumption}[theorem]{Assumption}
		\newtheorem*{openquestion}{Open Question}
		
		\theoremstyle{definition}
		\newtheorem{notation}[theorem]{Notation}
		\newtheorem{convention}[theorem]{Convention}
		\newtheorem{remark}[theorem]{Remark}
		
	}

	\maketitle
	
	\begin{abstract}
		The satisfiability problem (SAT) is one of the most famous problems in computer science. Traditionally, its NP-completeness has been used to argue that SAT is intractable. However, there have been tremendous practical advances in recent years that allow modern SAT solvers to solve instances with millions of variables and clauses. A particularly successful paradigm in this context is stochastic local search (SLS).

		In most cases, there are different ways of formulating the underlying SAT problem. While it is known that the precise formulation of the problem has a significant impact on the runtime of solvers, finding a helpful formulation is generally non-trivial. 
		The recently introduced \GapSAT{} solver~[\abstractlink[nicerblue]{cite.LW20OnTheEffectOfLearnedClauses}{Lorenz~and~Wörz~2020}] demonstrated a successful way to improve the performance of an SLS solver on average by learning additional information which logically entails from the original problem. Still, there were also cases in which the performance slightly deteriorated. This justifies in-depth investigations into how learning logical implications affects runtimes for SLS algorithms.

		In this work, we propose a method for generating logically equivalent problem formulations, generalizing the ideas of \GapSAT{}. This method allows a rigorous mathematical study of the effect on the runtime of SLS SAT solvers.
		Initially, we conduct empirical investigations.
		If the modification process is treated as random, Johnson SB distributions provide a perfect characterization of the hardness. 
		Since 
		the observed
		Johnson SB distributions approach lognormal distributions, our analysis also suggests that the hardness is long-tailed.

		As a second contribution, we theoretically prove that restarts are useful for long-tailed distributions. This implies that incorporating additional restarts can further refine \emph{all}~algorithms employing above mentioned modification technique.

		Since the empirical studies compellingly suggest that the runtime distributions follow Johnson~SB distributions, we also investigate this property on a theoretical basis. We succeed in proving that the runtimes for the special case of Schöning's random walk algorithm~[\abstractlink[nicerblue]{cite.Schoening02AProbabilisticAlgorithm}{Schöning~2002}] are approximately Johnson~SB distributed.
	\end{abstract}
	
	
	\prevversion{\textbf{This is the full-length version of the paper in the ACM Journal of Experimental Algorithmics (JEA).} See the last section for a discussion.}

	\newpage
	
	{\small
		\tableofcontents
	}

	\section{Introduction}
	\label{sec:intro}
	
	The \introduceterm{satisfiability problem} (\introduceterm{SAT}) asks to determine if a given propositional formula~$F$ has a satisfying assignment or not.
	Since Cook's $\mathsf{NP}$-completeness proof of the problem~\cite{Cook71ComplexityTheoremProving}, SAT is believed to be computationally intractable in the worst case.
	However, in the field of applied SAT~solving, there have been enormous improvements in the performance of SAT~solvers in the last 20 years. 
	Motivated by these significant improvements, SAT solvers have been applied to an increasing number of areas, including bounded model checking~\cite{DBLP:journals/ac/BiereCCSZ03,DBLP:journals/fmsd/ClarkeBRZ01}, cryptology~\cite{DBLP:conf/sat/EibachPV08}, and even bioinformatics~\cite{DBLP:conf/sat/LynceM06}, to name just a few.

	\introduceterm{Stochastic local search} (\introduceterm{SLS}) is an especially successful algorithmic paradigm that many SAT solvers employ~\cite[Chapter~6]{BHMW09HandbookOfSAT}: There are solvers solely based on the SLS paradigm, \eg the solvers \algoformat{probSAT}~\cite{BS12ChoosingProbabilityDistributionsSLS}, \dimetheus{}~\cite{dimetheus}, and \YalSAT{}~\cite{YalSAT}; SLS has been used in parallel solvers, \eg \algoformat{Plingeling}~\cite{biere2017splatz}; and is nowadays even a standard component of sequential conflict-driven clause learning (CDCL) solvers, for example of \algoformat{ReasonLS}~\cite{ReasonLS}, \algoformat{CaDiCaL}~\cite{BiereFazekasFleuryHeisinger-SAT-Competition-2020-solvers}, the \algoformat{Relaxed$^\ast$} family of solvers \cite{FourRelaxed, DeepCoop}, \algoformat{Kissat}~\cite{BiereFleuryHeisinger-SAT-Competition-2021-solvers}, newer versions of \algoformat{CryptoMiniSat}~\cite{SoosNC09}, and \algoformat{MergeSat}~\cite{MergeSat}.
	In~\cite{cai2022better}, Cai et al.\ tightly integrated SLS with three CDCL solvers, which significantly increased performance. The SLS paradigm is furthermore frequently employed in solving MaxSAT (see \egcite~\cite{maxsat}).

	Broadly speaking, SLS solvers operate on complete assignments for a formula~$F$. These solvers are started with a randomly generated complete initial assignment~$\alpha$. If $\alpha$~satisfies~$F$, a solution is found. Otherwise, the SLS solver explores the neighborhood of the current assignment by repeatedly flipping the value of some variable in the assignment when this variable is chosen according to some underlying heuristic (\eg aiming to minimize the number of unsatisfied clauses by the assignment). That is, these solvers perform a random walk over the set of complete assignments for the underlying formula.\footnote{In contrast to CDCL solvers and resolution, which are complete algorithms that can prove the unsatisfiability of a formula in a finite amount of steps, SLS solvers are incomplete, \ie in general, they cannot output the solution in a finite number of steps.}

	The success of SLS solvers is demonstrated by \probSAT{}~\cite{BS12ChoosingProbabilityDistributionsSLS}, \dimetheus{}~\cite{dimetheus}, and \YalSAT{}~\cite{biere2017splatz}, winning several gold medals in the random track of previous SAT competitions. SLS algorithms are also of interest from a theoretical perspective. For example, Schöning~\cite{Schoening02AProbabilisticAlgorithm} describes an algorithm (called \SRWA{} in the following) with an appealing worst-case guarantee. Furthermore, we firmly believe that a better understanding of SLS will help in the design of future CDCL--SLS hybrids.

	\subsection{Studying Runtime Distributions}
	
	Although SLS algorithms are highly successful in solving SAT instances, as witnessed by their comparatively low mean runtime, they often show a high variation in the runtime required to solve a fixed instance over repeated runs.
	However, measures like the mean or the variance cannot capture the long-tailed behavior of difficult instances.
	Some authors (\eg~\cite{FRV97SummarizingCSPHardness,GS97AlgorithmPortfolioDesign,RF97StatisticalAnalysis}) thus shifted their focus to studying the runtime distributions of search algorithms, which helps to understand these methods better and draw meaningful conclusions for the design of new algorithms.

	A relatively new algorithmic technique is considering modified versions of the input problem. For example, in the mixed integer programming community, it is known that the performance is sensitive to the used modification~\cite{lalla2016improving}.
	A similar approach is also employed in some backtracking SAT solvers (known as CDCL solvers~\cite{MS96Grasp,MMZZM01Engineering}) that learn additional information during their run. However, all \emph{successful} SLS SAT solvers of the last decades work on the original, unmodified instance.

	In~\cite{LW20OnTheEffectOfLearnedClauses}, the authors investigated the effect of modifying the input instance for SLS SAT solvers. More specifically, they changed the input instance by adding new, logically equivalent clauses to the problem. For this, a new solver, called~\GapSAT{}, was introduced. 
	This new solver is based on \probSAT{} and uses the addition of new clauses as a preprocessing step, thus, yielding a terraformed landscape. 
	A comprehensive experimental evaluation found statistical evidence that the performance of \probSAT{} substantially increased with this modification technique.
	However, the authors pointed to the fact that for some instances, the performance slightly deteriorated when \probSAT{} had access to these additional clauses, albeit all of them contained useful information.

	These experiments motivate to study the technique of adding new clauses in more detail. In particular, it seems worthwhile to obtain a better understanding of the phenomenon that adding new clauses improves the mean runtime, but there exist instances where adding clauses can harm the performance of SLS.

	Motivated by that, this work centers around studying the behavior of SLS solvers when these solvers work on formulas that were extended by logical consequences of the initial formulation.

	\subsection{Our Results}
	
	\subsubsection{Hardness Distribution}

	We study the runtime (or, more precisely, hardness) distribution of several SLS algorithms when logical implications are added to an original formula.
	Central to all our investigations is the basic elementary algorithm \Alfa{}, that we introduce in this work. This algorithm is specifically constructed in such a way that it is convenient to construct mathematical arguments after an initial empirical analysis.

	Our empirical evaluations suggest that the hardness distribution is long-tailed (called the \textbf{Weak Conjecture}). In fact, a stronger statement can be deduced: The data indicate that the distribution follows a Johnson SB distribution (called the \textbf{Strong Conjecture}). We also empirically show for our setting that this distribution converges to a lognormal distribution. Since lognormal distributions are long-tailed, it is thus already established that if the Strong Conjecture is true, the algorithm can be improved by restarts~\cite{Lorenz18RuntimeDistributions}. We extend this result to the case in which the Weak Conjecture is true:
	That is, we theoretically prove that restarts are useful for the larger class of algorithms that exhibit a long-tailed distribution.

	\subsubsection{Theoretical Arguments for the Hardness Distribution}
	
	It should be highlighted how good the Johnson SB fit is for the observed data. The distribution describes both typical and exceedingly low or high values exceptionally accurately. Only a marginal absolute and relative error between the fits and the observations can be observed. Moreover, this is true for all considered problem domains. 
	
	It is extraordinary that a simple parameterized distribution accurately describes the runtime behavior of an entire group of algorithms (SLS solvers) on various domains. Since such behavior is unlikely due to chance, we are pursuing theoretical explanations for this phenomenon.
	We succeed in showing that the hardness distribution for the special case of Schöning's random walk algorithm \SRWA{} is indeed approximately Johnson SB distributed, confirming the Strong Conjecture in practice.
	To the best of our knowledge, there are no comparable works deriving the runtime distribution for the full support.

	\subsection{Previous Work on Runtime Distributions}

	Before continuing, we proceed to report on related work regarding the analysis of runtime distributions. We include here related work showing why knowledge of runtime distributions, as we obtain it in this work, is immensely valuable.

	The study~\cite{FRV97SummarizingCSPHardness} presented empirical evidence for the fact that the distribution of the effort (more precisely, the number of consistency checks) required for backtracking algorithms to solve constraint satisfaction problems randomly generated at the 50\,\% satisfiable point can be approximated by the Weibull distribution (in the satisfiable case) and the lognormal distribution (in the unsatisfiable case).
	These results were later extended to a wider region around the 50\,\% satisfiable point~\cite{RF97StatisticalAnalysis}.
	It should be emphasized that this study created all instances using the same generation model. This resulted in the creation of similar yet logically \emph{non}-equivalent formulas.
	We, however, firstly use different models to rule out any influence of the generation model and secondly generate logically equivalent modifications~of~a base instance (see \refAlg{algo:main}).
	This approach lends itself to the analysis of existing~SLS solvers, like \OurSolver{}.
	The significant advantage is that the conducted work is not lost in the case~of a restart: only the logically equivalent instance could be changed while keeping the current assignment.

	The runtime distributions of CDCL solvers was studies in~\cite{ArbeitmitTom}.
	The authors empirically demonstrated that Weibull mixture distributions can accurately describe the multimodal distributions found.
	They concluded that adding new clauses to a base instance has an inherent effect of making runtimes long-tailed.

	In~\cite{GSCK00HeavyTailedPhenomena}, the cost profiles of combinatorial search procedures were studied. 
	It was shown that they are often characterized by Pareto-Lévy distributions and empirically demonstrated how rapid randomized restarts can eliminate tail behavior.
	We, however, theoretically prove the effectiveness of restarts for the larger class of long-tailed distributions.

	The paper~\cite{ATCTC13UsingSequentialRuntimeDistributions} studied the solvers \algoformat{Sparrow} and \algoformat{CCASAT} and found that the lognormal distribution is a good fit for the runtime distributions of randomly generated instances.
	For this, the Kolmogorov--Smirnov statistic $\sup_{t \in \R} | \hat{F}_n(t) - F(t) |$ was used.
	Although the KS-test is very versatile, this comes with the disadvantage that its statistical power is relatively low.
	The KS statistic is also nearly useless in the tails of a distribution: A high relative deviation of the empirical from the theoretical cumulative distribution function in either tail results in a very small absolute deviation.
	It should also be remarked that the paper studies only few formulas in just two domains, ten randomly generated and nine crafted.
	Our work addressed both shortcomings of this paper: The $\chi^2$-test gives equal importance to the goodness-of-fit over the full support, and various instance domain models (both theoretical and applied) are considered in this paper.

	We want to stress the fact that studies on the runtime distribution of algorithms are quite sparse, even though knowledge of the runtime distribution of an algorithm is extremely valuable:
	\begin{itemize}
		\item 
		Intuitively speaking, if the distribution is long-tailed, one knows there is a risk of ending in the tail and experiencing very long runs; simultaneously, the knowledge that the time the algorithm used thus far is in the tail of the distribution can be exploited to restart the procedure (and create a new logically equivalent instance~$F^{(2)}$). We rigorously prove this statement for \emph{all} long-tailed algorithms.
		\item
		Given the distribution of an algorithm's sequential runtime, it was shown in~\cite{ATCTC13UsingSequentialRuntimeDistributions} how to predict and quantify the algorithm's expected speedup due to parallelization. 
		\item
		If the hardness distribution is known, experiments with a small number of instances can lead to parameter estimations of the underlying distribution~\cite{FRV97SummarizingCSPHardness}.	
		\item 
		Knowledge of the distribution can help to compare competing algorithms: \eg one can test if the difference in the means of two algorithm runtimes is significant if the distributions are known~\cite{FRV97SummarizingCSPHardness}.
	\end{itemize}
	
	\subsection{Outline of This Paper}
	
	The rest of this paper is organized as follows.
	We start by presenting the necessary notations and the resolution proof system in \refSec{sec:Preliminaries}. This section also includes a short probability primer and an overview of probability distributions that will be appealed to in this paper.
	In \refSec{sec:evidence}, we then begin our empirical analysis to provide evidence for long-tails in SLS algorithms. We show that the Johnson SB distribution (which converges to the lognormal distribution) provides an exceptional fit to the hardness distribution. We also obtain strong evidence that the distribution is long-tailed. To conclude the section, we prove that restarts are useful for long-tailed distributions.
	\refSec{sec:theory} contains theoretical justifications for the Johnson SB distribution in the \SRWA{} case.
	Finally, in \refSec{sec:Conclusion}, we make some concluding remarks.
	All data and code produced for this paper is made publicly available. Therefore, all experiments are completely reproducible. For this, we refer to the end of the article.

	\section{Preliminaries}
	\label{sec:Preliminaries}

	\subsection{Basic Notation}
	\label{sec:NotationsBooleanLogic}
	
	A \introduceterm{literal} over a Boolean variable~$x$ is either $x$ itself or its negation~$\overline{x}$.
	A \introduceterm{clause} $C = a_1 \lor \dots \lor a_\ell$ is a (possibly empty) disjunction of literals $a_i$ over pairwise disjoint variables.
	A \introduceterm{CNF formula} $F = C_1 \land \dots \land C_m$ is a conjunction of clauses.
	\new{We will sometimes interpret clauses as sets of literals and CNF formulas as sets of clauses.}
	The set of variables of a clause~$C$ is denoted by $\vars{C}$. This notion is extended to formulas by taking unions.
	\new{The \introduceterm{width} of a clause~$C$ is given by~$|\vars{C}|$.}
	A CNF formula is a \introduceterm{$k$-CNF} if all clauses in it have at~most~$k$~variables.
	An \introduceterm{assignment}~$\tvastd$ for a CNF formula~$F$ is a function that maps some subset $\dom{\tvastd} \subseteq \vars{F}$ to $\{0,1\}$. \new{The assignment is called \introduceterm{complete} if $\dom{\tvastd} = \vars{F}$, otherwise it is called \introduceterm{partial}.}
	The application of an assignment~$\alpha$ to a clause~$C$ or a formula $F$ will be denoted with $\applyassi{C}{\alpha}$ or $\applyassi{F}{\alpha}$, respectively.
	An assignment~$\tvastd$ \introduceterm{satisfies} a CNF formula~$F$ if at least one literal in every clause of~$F$ is set to~$1$ by~$\tvastd$.
	A formula~$F$ \introduceterm{logically implies} a clause~$C$ if every complete truth assignment which satisfies~$F$ also satisfies~$C$, for which we write $F \vDash C$. If $L$ is a set of clauses, we write $F \vDash L$ if $F \vDash C$ for all~$C \in L$. If $L$ is such that $F \vDash L$, then we call $F$ and $F \cup L$ \introduceterm{logically equivalent} formulas.
	The act of changing the truth value of precisely one variable of a complete assignment~$\tvastd$ is called a \introduceterm{flip}.
	When changing an assignment $\tvastd$ by flipping a variable~$x$ of this assignment, the new assignment will be denoted with~$\tvastd[x]$, \ie $\tvastd[x](x) \coloneqq 1 - \tvastd(x)$, while $\tvastd[x](y) \coloneqq \tvastd(y)$ for $y \neq x$.
	If $\tvastd(x) = 1$, we also write $\FlipVarInAssi{\overline{x} = 1}{\tvastd}$ for $\tvastd[x]$; otherwise we write $\FlipVarInAssi{x = 1}{\tvastd}$ for $\tvastd[x]$.
	
	\subsection{The Resolution Proof System}
	
	\introduceterm{Resolution} is the proof system with the single derivation rule 
	\[
	\frac{B \lor x \:\:\: C \lor \overline{x}}{B \lor C},
	\]
	where~$B \lor x$ and~$C \lor \overline{x}$ are clauses. 
	Clearly,
	\[
	(B \lor x) \land (C \lor \overline{x}) \vDash (B \lor C).
	\]
	In the paper, we will also use \emph{width-$w$ restricted resolution}, introduced in the following.

	\begin{definition}
		\label{def:width-w-res}
		Let $F$ be a clause set, and $w \in \N$ be a positive integer. We define the operator
		\[
		\operatorname{Res}_w(F) \coloneqq F \cup \Setdescr{R}{R \text{ is a resolvent of two clauses in } F \text{ and } |R| \leq w}.
		\]
		Moreover, we inductively define $\operatorname{Res}_w^{0}(F) \coloneqq F$ and
		\[
		\operatorname{Res}_w^{n+1}(F) \coloneqq \operatorname{Res}_w \! \big( \operatorname{Res}_w^{n}(F) \big), \text{ for } n \geq 0.
		\]
		Finally, we set	
		\[
		\operatorname{Res}_w^{\ast}(F) \coloneqq \bigcup_{n \geq 0} \operatorname{Res}_w^{n}(F).
		\]
	\end{definition}

	\subsection{A Short Probability Primer}

	We assume knowledge of conditional probabilities.
	In \refSec{sec:theory}, however, these elementary notions do not suffice. Thus, in this section, we introduce the necessary concepts involving random variables in the expectations.

	All random variables in this paper will be denoted with bold lettering.
	We wish to especially highlight that some random variables describe probabilities for which we will use the notation~$\ProbopZV\,(\cdot)$. On the other hand, some probabilities are constants not depending on a random set, say $\ZV{\Clauses}$, and are thus denoted by~$\prob{\cdot}$. 
	We will adhere to this convention already in the preliminaries to help the readers familiarize themselves with this notation. While this differentiation might initially feel strange and unnecessary, we believe it immensely helps to parse the equations in \refSec{sec:theory} and find the random variables that ``hide under the cloak'' of appearing as a standard probability at first sight.

	\begin{definition}
		\label{def:KomplizierteWahrscheinlichkeit}
		Let $\ZV{X}$ be a discrete random variable and $A$ an event.
		The \introduceterm{conditional probability of $A$ given $\ZV{X}$} is defined as the random variable, written ${\PCondNoL{A}{\ZV{X}}}(\omega)$, that takes on the value $\probcond{A}{\ZV{X}=x}$ whenever $\ZV{X}=x$. More formally,
		\[{\PCondNoL{A}{\ZV{X}}}(\omega) \coloneqq \probcond{A}{\ZV{X}=\ZV{X}(\omega)}.\]
	\end{definition}

	Thus, the conditional probability $\PCondNoL{A}{\ZV{X}}$ is a function of~$\ZV{X}$ and, therefore, itself a random variable (thus denoted with the $\ProbopZV\,$-symbol and round brackets for better discriminability). In particular, it is not a real value in the interval~$[0,1]$.

	A similar concept like in \refDef{def:KomplizierteWahrscheinlichkeit} can be defined for expectations.

	\begin{definition}[\cite{mitzenmacher2017probability}]
		\label{def:KomplizierterErwartungswert}
		Let $\ZV{X}$ %
		be a random variable on a sample space~$\samplespace$. %
		Further, let $\ZV{Y}$ 
		be a discrete random variable defined on the same sample space. 
		Then, the \introduceterm{conditional expectation} of $\ZV{X}$ with respect to $\ZV{Y}$ is the random variable $\expectationcondZV{\ZV{X}}{\ZV{Y}}(\cdot) \colon \samplespace \to \R$
		defined by
		\[\expectationcondZV{\ZV{X}}{\ZV{Y}} (\omega) \coloneqq \expectationcond{\ZV{X}}{\ZV{Y} = \ZV{Y}(\omega)}.\]
	\end{definition}

	Notice that $\expectationcondZV{\ZV{X}}{\ZV{Y}}$ itself is again a random variable -- it is not a real value. Its value depends on the random variable~$\ZV{Y}$. We make this clear with two examples. The second example will take a central stage in Section~\ref{sec:theory}.

	\begin{example}[\cite{mitzenmacher2017probability}]
		Suppose that two standard dice are rolled independently. Let $\ZV{X_1}$ be the result of the first dice, $\ZV{X_2}$ the result of the second dice, and $\ZV{S}$ the sum of both results. 
		For all $x \in \set{1, \dots, 6}$ it holds 
		\[
		\expectationcond{\ZV{S}}{\ZV{X_1}=x} = \frac{1}{6} \sum_{s = x+1}^{x+6} s = x +\frac{7}{2}.
		\]
		Hence, $\expectationcondzv{\ZV{S}}{\ZV{X_1}} = \ZV{X_1} + \frac{7}{2}$ is a random variable whose value depends on $\ZV{X_1}$. If event~$\omega$ occurs, then $\ZV{X_1}$ has value~$\ZV{X_1}(\omega)$, and therefore $\expectationcondzv{\ZV{S}}{\ZV{X_1}}$ takes on the value 
		\[
		\expectationcondzv{\ZV{S}}{\ZV{X_1}}(\omega) = \expectationcond{\ZV{S}}{\ZV{X_1} = \ZV{X_1}(\omega)} = \ZV{X_1}(\omega) + \frac{7}{2} \in \R.
		\] 
	\end{example}

	\begin{example}
		\label{ex:EFlipsFundLGegebenL}
		If we let $\W(G)$ denote the random variable specifying the number of flips a fixed SLS algorithm takes to find a satisfying assignment for instance~$G$, then 
		\[
		\ERL \coloneqq \expectationcond{\W(F\cup \ZV{\Clauses})}{\ZV{\Clauses}}
		\]
		denotes the expected runtime of this SLS algorithm on the extended instance~$F \cup \ZV{\Clauses}$, which is dependent on the concrete realization~$L$ of the random extension set~$\ZV{\Clauses}$. In particular, 
		\[
		\expectationcondZV{\W(F\cup \ZV{\Clauses})}{\ZV{\Clauses}}(L) = \expectationcond{\W(F\cup \ZV{\Clauses})}{\ZV{\Clauses} = L}.
		\]
	\end{example}

	\begin{theorem}[Law of total probability, LTP]
		Let $(B_i)_i$ be a finite or countable partition of the sample space $\samplespace$ such that $\prob{B_i} > 0$ for all $i$ and $A \subseteq \samplespace$. Then
		\[
		\prob{A} = \sum_{i} \probcond{A}{B_i} \cdot \prob{B_i}.
		\]
		Similarly, if\/ $\probcond{A}{C}$ is defined for $C \subseteq \samplespace$ and the partition is such that $\prob{C \cap B_i} > 0$ for each~$i$, then
		\[
		\probcond{A}{C} =
		\sum_i \probcond{A}{C \cap B_i} \cdot
		\probcond{B_i}{C}.
		\]
	\end{theorem}

	To simplify arguments using the LTP, it is common practice to omit all terms for which $\prob{B_i} = 0$, because $\probcond{A}{B_i}$ is finite (if $\prob{B_i}=0$, then according to the simple definition above, $\probcond{A}{B}$ is undefined; however, it is possible to define a conditional probability with respect to a $\sigma$-algebra of such events). The same holds for the LTE below.

	\begin{theorem}[Law of total expectation, LTE]%
		Let $\ZV{X}$ be a discrete random variable on a probability space $\probspace$ such that $\expectation{\ZV{X}}$ is defined. Further, let $(A_i)_i$ be a finite or countable partition of $\samplespace$ such that $\prob{A_i} > 0$ for each $i$.
		Then
		\[
		\expectation{\ZV{X}} = \sum_i \expectationcond{\ZV{X}}{A_i} \cdot \prob{A_i}.
		\]
		Similarly, if\/ $\expectationcond{\ZV{X}}{B}$ is defined and the partition is such that $\prob{B \cap A_i} > 0$ for each $i$, then
		\[
		\expectationcond{\ZV{X}}{B} =
		\sum_i \expectationcond{X}{B \cap A_i} \cdot
		\probcond{A_i}{B}.
		\]
	\end{theorem}

	\begin{theorem}[Chain rule for probabilities]
		\label{thm:ChainRule}
		If $A_1, \dots, A_n$ are random events, then
		\[
		\PROB{ \bigcap_{k=1}^n A_k }
		=
		\prod_{k=1}^n \probcondResize{A_k}{ \bigcap_{j=1}^{k-1} A_j }.
		\]
	\end{theorem}

	\subsection{Probability Distributions}
	
	Throughout the paper, we will make use of various probability distributions. We need the concepts of cumulative distribution functions and probability density functions to introduce these distributions.

	\begin{definition}[\cite{norman1994continuous}]
		\label{def:cdf_quantilemain}
		Let $\ZV{X}$ be a real-valued random variable.
		\begin{enumerate}[(i)]
			\item
			Its \introduceterm{cumulative distribution function} (cdf) is the function $F \colon \R \to [0,1]$ with \[F(t) \coloneqq \prob{\ZV{X} \leq t}.\]
			\item
			If $\ZV{X}$ is continuous, its \introduceterm{quantile function} $Q \colon (0,1) \to \R$ is given by \[Q(p) \coloneqq \inf \setdescr{t \in \R}{F(t) \geq p}.\]
			\item
			Its \introduceterm{survival function}~$S$ is given by \[S(x) \coloneqq 1-F(x).\] 
			\item
			If a non-negative, integrable function~$f$ with the property 
			\[
			F(t) = \integralLow[u]{-\infty}{t}{f(u)} \text{ for all } t \in \R
			\]
			exists, it is called \introduceterm{probability density function} (pdf) of~$\ZV{X}$.
		\end{enumerate}
	\end{definition}

	If the underlying cdf of a sample is unknown, we use the empirical distribution function.
	
	\begin{definition}
		\label{def:ecdf}
		Let $\ZV{X}_1, \dots, \ZV{X}_n$ be independent, identically distributed, real-valued random variables
		with realizations $x_i$ of $\ZV{X}_i$.
		Then, the \introduceterm{empirical cumulative distribution function} (\introduceterm{ecdf}) of the sample $(x_1, \dots, x_n)$ is defined as
		\[
		\hat{F}_n(t) \coloneqq \frac{1}{n} \sum_{i=1}^n \Indikator{\set{x_i \leq t}}, \quad t \in \R,
		\]	
		where $\Indikator{A}$ is the indicator of event $A$.
	\end{definition}

	\subsubsection{The Johnson SB Distribution}
	
	Central to all distributions considered in this paper and necessary to introduce the Johnson SB distribution is the concept of the well-known Gaussian normal distribution.
	
	\begin{definition}
		An absolutely continuous random variable $\ZV{X}$ is \introduceterm{normally distributed} with expectation~$\mu \in \R$ and variance~$\sigma^2 > 0$, denoted by $\ZV{X} \distr \Normal{\mu}{\sigma^2}$, if the pdf of $\ZV{X}$ is given by	\[f_{\,\text{N}}(x \mid \mu, \sigma) = \frac{1}{\sigma \sqrt{2 \pi}} \exp \left( - \frac{(x-\mu)^2}{2 \sigma^2} \right).\]
	\end{definition}
	
	Using normal distributions, we introduce the Johnson SB distribution, which takes the central stage in this work.
	
	\begin{definition}[\cite{Johnson49System,Johnson49Bivariate,norman1994continuous}]
		An absolutely continuous random variable $\ZV{X}$ is \introduceterm{Johnson SB distributed} with parameters $\gamma \in \R$, $\delta > 0$, $\lambda > 0$, and $\xi \in \R$, denoted as $\ZV{X} \distr \SB{\gamma}{\delta}{\lambda}{\xi}$, if $\xi < \ZV{X} < \xi + \lambda$ and
		\begin{align*}
			\ZV{Z} 
			\coloneqq 
			\gamma + \delta \cdot \log \left( \frac{\ZV{X} - \xi}{\xi + \lambda - \ZV{X}} \right) \distr \Normal{0}{1}.
		\end{align*}
	\end{definition}

	The Johnson SB distribution is highly flexible and can model distributions with finite support. \refFig{fig:some_sb_pdfs} illustrates several Johnson SB distributions, including the effect of the parameters on the form of the pdf.

	\begin{figure}[ht]
		\centering
		\begin{subfigure}[t]{.465\linewidth}
			\centering
			\includegraphics{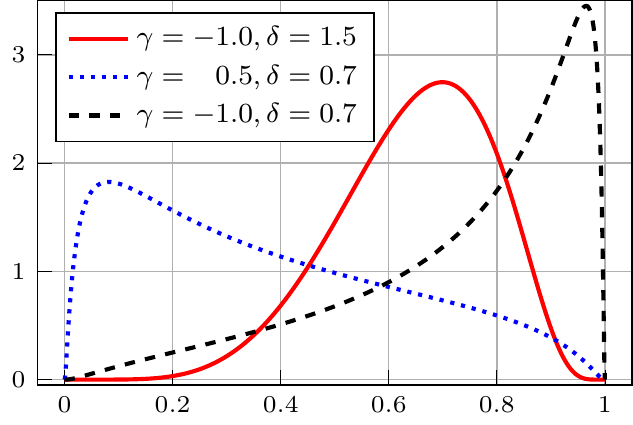}
			\subcaption{$\xi=0, \lambda=1$}
		\end{subfigure}%
		\hfill
		\begin{subfigure}[t]{.465\linewidth}
			\centering
			\includegraphics{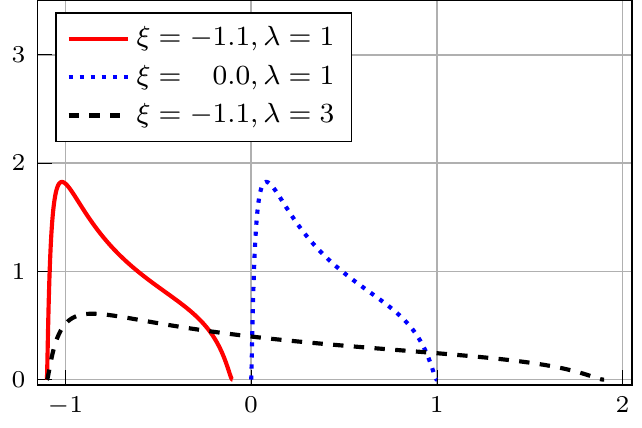}
			\subcaption{$\gamma=0.5, \delta=0.7$}
		\end{subfigure}%
		\caption{
			This figure shows the effect of the parameters on the pdf of the Johnson SB distribution. All Johnson SB distributions in the left plot use $\xi=0, \lambda=1$ as parameters. In the right plot, all distributions use $\gamma=0.5, \delta=0.7$ as parameters. In both cases, there are two varying parameters that are given in the respective legend.
		}
		\label{fig:some_sb_pdfs}
	\end{figure}

	\begin{remark}[\cite{cheng2017non}]
		\label{rem:SBParam}
		A Johnson SB distributed random variable has positive density support on $(\xi, \xi + \lambda)$. The parameters $\gamma$ and $\delta$ are \introduceterm{shape} parameters (governing the \introduceterm{asymmetry} and \introduceterm{kurtosis}, respectively), $\lambda$ is the \introduceterm{scale} parameter, and $\xi$ is the \introduceterm{location} parameter.
		Letting $a \coloneqq \xi$ and $b \coloneqq \xi + \lambda$,
		the pdf~$f_{\,\text{SB}}$ is given by
		\begin{align*}
			f_{\,\text{SB}}(x \mid a, b, \gamma, \delta) = \frac{1}{\sqrt{2 \pi}} \frac{(b-a)\delta}{(x-a)(b-x)}
			\exp{\biggl\{ -\frac{1}{2} \biggl[ \gamma + \delta \ln{\left( \frac{x-a}{b-x} \right)} \biggr]^2 \biggr\}},\quad x \in (a,b).
		\end{align*}
	\end{remark}

		Furthermore, the Johnson SB distribution has the following scaling property.
	
	\begin{lemma}
		\label{lem:MultiplicationOfSB}
		Let $\ZV{X}$ be a Johnson SB distributed random variable having parameters $a, b, \gamma, \delta \in \R$ with $a<b$ and $\delta > 0$. 
		Then, the random variable $\ZV{Y} = g \cdot \ZV{X}$, where $g \in \R_+$ is an arbitrary, positive real number, is Johnson SB distributed with parameters $ga, gb, \gamma, \delta \in \R$.
	\end{lemma}

	\begin{proof}
		Let $f_{\ZV{X}}$ and $F_{\ZV{X}}$ denote the pdf and the cdf of $\ZV{X}$. Likewise, $f_{\ZV{Y}}$ and $F_{\ZV{Y}}$ are the pdf and the cdf of $\ZV{Y}$.
		Since
		\begin{align*}
			\prob{\ZV{Y} \leq x} 
			= \prob{g \cdot \ZV{X} \leq x} 
			= \PROB{\ZV{X} \leq \frac{x}{g}}
			= F_{\ZV{X}} \left( \frac{x}{g} \right),		
		\end{align*}
		it follows $f_{\ZV{Y}}(x) = \frac{1}{g} f_{\ZV{X}} \left( \frac{x}{g} \right)$.
		Thus,
		\begin{align*}
			f_{\ZV{Y}}(x)&=
			\frac{1}{g}
			\frac{1}{\sqrt{2 \pi}} \frac{(b-a)\delta}{(\frac{x}{g}-a)(b-\frac{x}{g})}
			\exp{\left\{ -\frac{1}{2} \left[ \gamma + \delta \ln{\left( \frac{\frac{x}{g}-a}{b-\frac{x}{g}} \right)} \right]^2 \right\}}\\
			&=
			\frac{1}{\sqrt{2 \pi}} \frac{(gb-ga)\delta}{(x-ga)(gb-x)}
			\exp{\biggl\{ -\frac{1}{2} \biggl[ \gamma + \delta \ln{\left( \frac{x-ga}{gb-x} \right)} \biggr]^2 \biggr\}}.
		\end{align*}
		This is the pdf of a Johnson SB distribution with parameters $ga, gb, \gamma, \delta$.
	\end{proof}

	\subsubsection{Lognormal Distributions as Embedded Model of Johnson SB Distributions}
	Experimentally, we show that the fits of the SLS hardness distributions exhibit parameter combinations that suggest the involvement of an embedding process: Informally speaking, the Johnson SB distribution can be thought of as converging to a lognormal distribution (see \refFig{fig:sb_converge} for an illustration). Hence, the Johnson SB distribution is sometimes referred to as four-parameter lognormal distribution~\cite{AitchisonBrown}.
	The lognormal distribution is given as follows.

	\begin{definition}[\cite{Wicksell17OnLogarithmicCorrelation}]
		\label{def:threeparameterLogN}
		An absolutely continuous, positive random variable~$\ZV{X}$ is \introduceterm{(three-parameter) lognormally distributed} with parameters $\sigma^2 > 0$, $\xi > 0$, and $\mu \in \R$, if $\log(\ZV{X} - \xi)$, where $\ZV{X} > \xi$, is normally distributed with mean $\mu$ and variance $\sigma^2$. In the following, we refer to $\sigma$ as the \introduceterm{shape}, $\mu$ as the \introduceterm{scale}, and $\xi$ as the \introduceterm{location parameter}.
		
		If the location parameter $\xi$ is zero, we call $\ZV{X}$ two-parameter lognormal distributed and commonly omit $\xi$.
	\end{definition}

	\begin{remark}[\cite{statisticlognormal}]
		The pdf $f_{\,3\text{LogN}}$ of the three-parameter lognormal distribution is given by
		\[
		f_{\,3\text{LogN}}(x \mid \mu, \sigma, \xi) = \frac{1}{(x-\xi) \sigma \sqrt{2 \pi}} \exp \left\{ - \frac{\big[ \ln(x-\xi) - \mu \big]^2}{2 \sigma^2} \right\} 
		, \quad x > \xi.
		\]
	\end{remark}

	The next definitions make the embedding of the lognormal distribution in the Johnson SB distribution more precise.

	\begin{definition}[\cite{cheng2017non}]
		Consider a function $f$ having parameters $\Theta$. Furthermore, let $\Theta = (\Theta^L, \Theta^R)$ be a partition of the parameters with $\Theta^L = (\Theta_1, \dots, \Theta_\ell)$ and $\Theta^R = (\Theta_{\ell + 1}, \dots, \Theta_{\ell +r})$. Lastly, assume $\widehat{\Theta} = g(\Theta)$ for some function $g$. If
		\begin{align*}
			\lim_{\Theta^L \rightarrow 0} f(x \mid \Theta) = f_0(x \mid \widehat{\Theta}),
		\end{align*}
		for some well-defined function $f_0$, then $f_0$ is called an \introduceterm{embedded model} of $f$.
	\end{definition}
	
	\begin{lemma}[\cite{cheng2017non,norman1994continuous}]
		\label{lem:LogNEmbeddedModelOfSB}
		The lognormal distribution is an embedded model of the SB distribution.
	\end{lemma}

	\begin{proof}
		We have provided a proof in Appendix~\ref{sec:MoreOnEmbedding}.
	\end{proof}
	
	In particular, $f_{\,\text{SB}} \to f_{\,3\text{LogN}}$ for the reparametrization $\gamma = \delta \big( \ln(b-a) - \mu \big)$ in the Johnson SB distribution and $b \to \infty$.
	For all intents and purposes, it suffices for the right endpoint~$b$ of the density support $(a,b)$ to increase (or $\lambda$, respectively), while $\gamma$ can grow logarithmically slow~\cite{cheng2017non}. We refer to \refFig{fig:sb_converge} for an illustration.

	\begin{figure}[ht]
		\centering
		\begin{subfigure}[t]{.3\linewidth}
			\centering
			\includegraphics[height=7.84em]{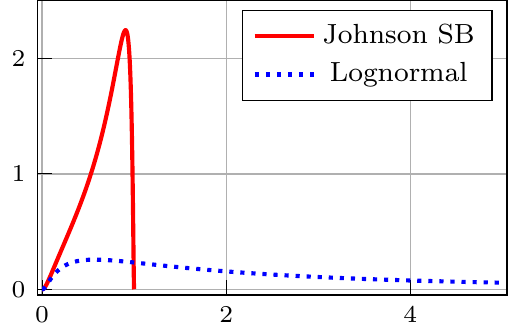}
			\subcaption{$\lambda=1.0, \gamma=-0.8$}
		\end{subfigure}%
		\hfill
		\begin{subfigure}[t]{.3\linewidth}
			\centering
			\includegraphics[height=8.25em]{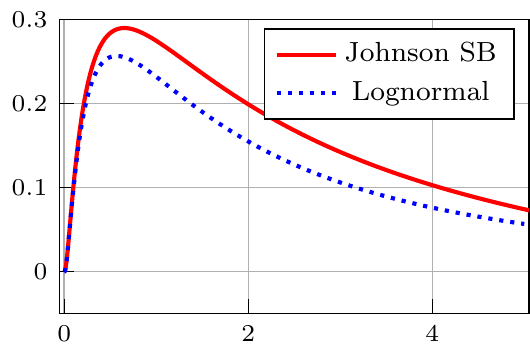}
			\subcaption{$\lambda=10.0, \gamma\approx1.042$}
		\end{subfigure}%
		\hfill	
		\begin{subfigure}[t]{.3\linewidth}
			\centering
			\includegraphics[height=8.25em]{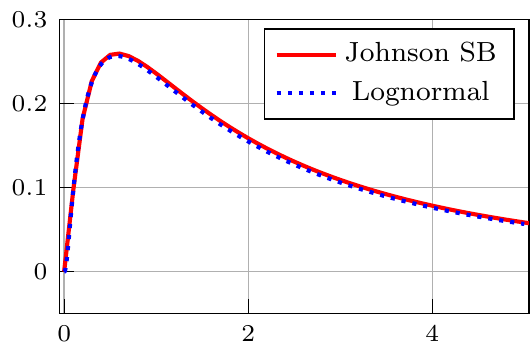}
			\subcaption{$\lambda=100.0, \gamma\approx2.884$}
		\end{subfigure}%
		\caption{
			This figure shows that the Johnson SB distribution converges towards a fixed lognormal distribution if $\lambda$ and $\gamma$ approach infinity. Each subfigure shows a Johnson SB distribution together with its limiting lognormal distribution. The Johnson SB distributions all have $\xi=0.0, \delta=0.8$ as parameters. The remaining two parameters $\lambda$ and $\gamma$ are given in the corresponding subcaptions. 
			One should note that $\gamma$ can take relatively small values for the convergence to still take place. More precisely, it suffices for $\gamma$ to grow logarithmically slow. This is theoretically explained in~\cite{cheng2017non}.
			The limiting lognormal distribution has $\sigma=1.25, \mu=1.0$ as parameters.
		}
		\label{fig:sb_converge}
	\end{figure}

	\section{Evidence for Long-Tails in SLS Algorithms}
	\label{sec:evidence}

	The authors of~\cite{LW20OnTheEffectOfLearnedClauses} introduced the hybrid solver \algoformat{GapSAT} by augmenting an SLS solver with a clause-learning feature. After receiving a set of additional clauses~$\Clauses$ (implied by the original formula~$F$), the solver can be understood as solving a modified instance~$F^\prime$. This paper showed that adding new clauses is beneficial to the mean runtime (in flips) of the SLS solver \probSAT{} underlying the hybrid model. However, it was also demonstrated that although adding new clauses can improve the mean runtime, there exist instances where adding clauses can harm the performance of SLS.
	As announced, this behavior is worth studying further to help eliminate the risk of increasing the runtime of such procedures.

	For this reason, in this section, we study the runtime (or, more precisely, \emph{hardness}) distribution of the procedure~\Alfa{} that we introduce below. This procedure models the addition of a random set of logically equivalent clauses~$\Clauses$ to a formula~$F$ and the subsequent solving of this amended formula~$F^{(1)} \coloneqq F \cup \Clauses$ by an SLS solver.
	Our empirical evaluations show that this distribution is long-tailed.
	This fact enables us to prove that restarts are useful for~\Alfa{}.

	\subsection{Design of the Adjusted Logical Formula Algorithm Alfa}
	\label{sec:DesignOfAlfa}

	Our SLS solver \Alfa{} (Adjusted logical formula algorithm) receives a satisfiable formula~$F$ as input.
	The algorithm then proceeds by adding to~$F$ a random set~$\Clauses$ of logically generated clauses.
	It finally calls an SLS solver to solve the clause set~$F \cup \ZV{\Clauses}$.
	The pseudocode for \Alfa{} is given in \refAlg{algo:main}.

	\begin{algorithm}[htb]
		\textbf{Input:} Boolean formula $F$, \textbf{Promise:} $F \in \SAT$\\
		\BlankLine
		Generate \textbf{randomly} a set $\Clauses$ of clauses such that $F \vDash \Clauses$ (\eg with \refAlg{algo:res})\\
		Call \algoformat{SLS}$(F \cup \Clauses)$ for some SLS solver~\algoformat{SLS}
		\caption{\Alfa{} acts as a base algorithm that can use different SLS algorithms.}
		\label{algo:main}
	\end{algorithm}

	We use width-$w$ restricted resolution (recall \refDef{def:width-w-res}) in \refAlg{algo:res} as a natural way to sample a set~$\Clauses$ of logically equivalent clauses with respect to a base instance~$F$.
	This allows us the formulation of \refAlg{algo:res} that is used to generate random sets $L$ with resolution.

	\begin{algorithm}[htb]
		
		\SetKw{KwWithProb}{with probability}
		\SetKw{KwDo}{do}	
		
		\textbf{Input:} Boolean formula $F$, integer $w$, probability $p \in (0,1]$\\
		\BlankLine
		
		$L \coloneqq \emptyset$\\
		\ForEach{$R \in \operatorname{Res}_w^{\ast}(F) \setminus F$}{
			\KwWithProb $p$ %
			\KwDo
			$L \coloneqq L \cup \set{R}$
		}
		
		\textbf{return} $L$
		
		\caption{Generation of the random set $\Clauses$ with width-$w$ restricted resolution.}
		\label{algo:res}
	\end{algorithm}

	\subsection{Empirical Evaluation of the Hardness Distribution}

	\subsubsection{Experimental Setup, Instance Types, and Solvers Used to Obtain Hardness Distribution Data}
	\label{sec:ExpSetupInstTypesSolversUsed}

	Hoos and Stützle~\cite{HS98EvaluatingLasVegas} introduced the concept of \emph{runtime distribition} to characterize the cdf of Las Vegas algorithms, where the runtime can vary from one execution to another, even with the same input.
	To obtain enough data for a fitting of such a distribution, for each base instance~$F$, we created 5000 modified instances $F^{(1)}, \dots, F^{(5000)}$ by generating resolvent sets $\Clauses^{(1)}, \dots, \Clauses^{(5000)}$ using \refAlg{algo:res} with $w=4$ and a value of $p$ such that the expected number of resolvents being added was~$\frac{1}{10} |F|$.
	We also conducted experiments to rule out the influence of~$p$ on our results.
	Each of the modified instances was solved 100 times, each time using a different seed.
	For $i = 1, \dots, 5000$ and $j = 1, \dots, 100$, we obtained the values $\mathsf{flips}_{S}(F^{(i)}, s_j)$ indicating how many flips solver~$S$ used to solve the modified instance~$F^{(i)}$ when using seed~$s_j$.
	Next, we calculated 
	\[
	\mathsf{mean}_S(F^{(i)}) \coloneqq \frac{1}{100} \sum_{j=1}^{100} \mathsf{flips}_{S}(F^{(i)}, s_j),
	\]
	the mean number of flips
	required to solve $F^{(i)}$ with solver~$S$ whose hardness distribution we are going to analyze.

	All experiments were performed on bwUniCluster 2.0 and three local servers, using
	Sputnik~\cite{VLSKK15Sputnik} to distribute the computation and parallelize the trials.
	Due to the heterogeneity of the computer setup, measured runtimes are not directly comparable to each other.
	Consequently, we instead measured the number of variable flips performed by the SLS solver.
	This is a hardware-independent performance measure with the benefit that it can also be analyzed theoretically.

	Next, we describe the generation and satisfiability sampling of the instances.
	For the experiments, the following instance types were used:
	
	\begin{enumerate}[(1)]
		\item \textit{Hidden Solution:}		
		We used our implementation~\cite{LW21SourceCodeOfConcealSATgen} of the CDC algorithm~\cite{BC18UsingAlgorithmConfigurationTools, BHLRTWZ02Hiding} to generate instances with a hidden solution. 
		SLS solvers typically struggle to solve such instances~\cite{LW20OnTheEffectOfLearnedClauses}. Thus, experiments like these might be beneficial in finding theoretical reasons for this behavior.	
		
		\item \textit{Hidden Solution With Different Chances:}		
		We also created formulas with different chance values, \ie the probability of adding a clause in \refAlg{algo:res}. The purpose is to rule out the influence of the chance value.

		\item \textit{Uniform Random:}	
		Using Gableske's \algoformat{kcnfgen}~\cite{kcnfgen}, we generated formulas with $n \in \set{50,60,70,80,90}$ variables and	a clause-to-variable ratio~$r$ close to the \emph{satisfiability threshold}~\cite{MMZ06ThresholdValues} of $r \approx 4.267$.
		Then, we checked each instance with \algoformat{Glucose3}~\cite{AS09Glucose,ES03MiniSat}
		for satisfiability until we had five formulas of each size.

		\item \textit{Factoring:}
		We encoded the factoring problem in the interval $\set{128, \dots, 256}$ with~\cite{Diemer21GenFactorSat}.

		\item \textit{Coloring:}
		These formulas assert that a graph is colorable with three colors.
		We generated these formulas, using~\cite{LENV17CNFgen}, over random graphs with $n$ vertices and $m = 2.254n$ edges in expectation,	which is slightly below the \introduceterm{non-colorability threshold}~\cite{KaporisKS00}.
		We obtained 32~satisfiable instances in 150 variables.
	\end{enumerate}

	Our experiments investigated leading SLS solvers whose dominating component is based on the random walk procedure proposed in~\cite{Schoening02AProbabilisticAlgorithm}.
	In this paper, Schöning's Random Walk Algorithm \SRWA{} (see \refAlg{algo:schoening} on page~\pageref{algo:schoening}) was introduced.
	The \probSAT{} solver family~\cite{BalintImplementationOfProbSAT}, including \YalSAT{}~\cite{YalSAT}, is based on this approach.
	The excellent performances and similarities were reasons for choosing \SRWA{}, \probSAT{}, and \YalSAT{} as main solvers
	(\probSAT{} won the random track of the SAT competition 2013, and \YalSAT{} won in 2017).
	Only recently, in 2018, other types of solvers significantly exceeded \probSAT{}-based algorithms. This lasting performance is why this solver family is chosen in this study.

	For \SRWA{}, we conducted most of our experiments: All instance types were tested, including different chance values in \refAlg{algo:res}.
	For \probSAT{}, 55 hidden solution instances with $n \in \set{50,100,150,200,300,800}$ were used.
	Since \YAL{} can be regarded as a \probSAT{} derivate, we tested \YAL{} with ten hidden solution instances with 300 variables each.

	\subsubsection{Experimental Results and Statistical Evaluation of the Hardness Distribution}
	
	This section aims to explore how an instance's hardness changes
	when logically equivalent clauses are added in the manner described above. To characterize this effect as accurately as possible, studying the ecdf is the most suitable method (recall Definition~\ref{def:ecdf}).

	In the following, we demonstrate that the Johnson SB distribution, in particular, provides an exceptionally accurate description of the runtime behavior, and this is true for all considered problem domains and solvers. The results are so compelling that we ultimately conjecture that the runtimes of \Alfa{}-type algorithms all follow a Johnson SB distribution, regardless of the problem domain.

	To illustrate the accurate description of the runtime behavior mentioned above, we first demonstrate our approach using two base instances. The first one is a factorization instance that SRWA solved. The second instance has a hidden solution and was solved by \probSAT{}. 
	We refer to the first instance as~$A$ and to the second instance as~$B$.
	We estimate the Johnson SB distribution's four parameters using the 5000 data points obtained by applying the maximum likelihood method (see~\cite{AllData}). After that, one can visually evaluate the suitability of the fitted Johnson SB distribution for describing the data by plotting the ecdf and the fitted cdf on the same graph.

	Such a comparison is illustrated in \refFig{fig:cdf-ecdf} for the two instances~$A$ and~$B$. In both cases, no difference between the empirical data of the ecdf and the fitted distribution can be detected visually (the absolute error between the predicted probabilities from the fitted cdf versus the empirical probabilities from the ecdf is minuscule). These two example instances are representative of the behavior of the investigated algorithms. Hardly any deviation could be observed in this plot type for all instances and all algorithms (all data is published under~\cite{AllData}).
	
	\begin{figure}[htb]
		\centering
		\begin{subfigure}[b]{.465\linewidth}
			\centering
			\includegraphics{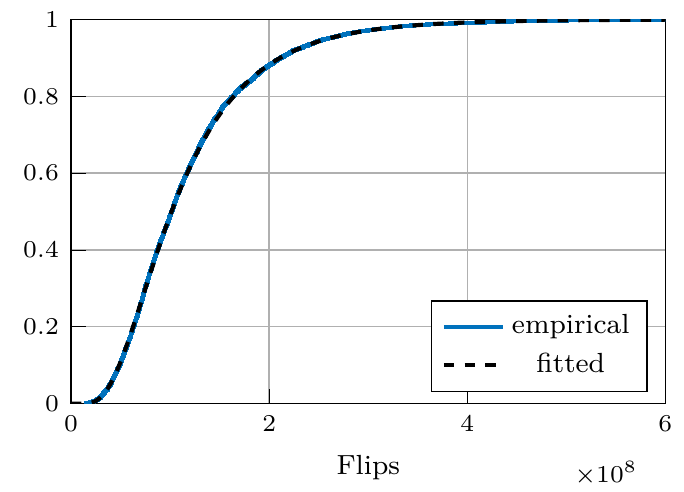}
		\end{subfigure}%
		\hfill%
		\begin{subfigure}[b]{.465\linewidth}
			\centering
			\includegraphics{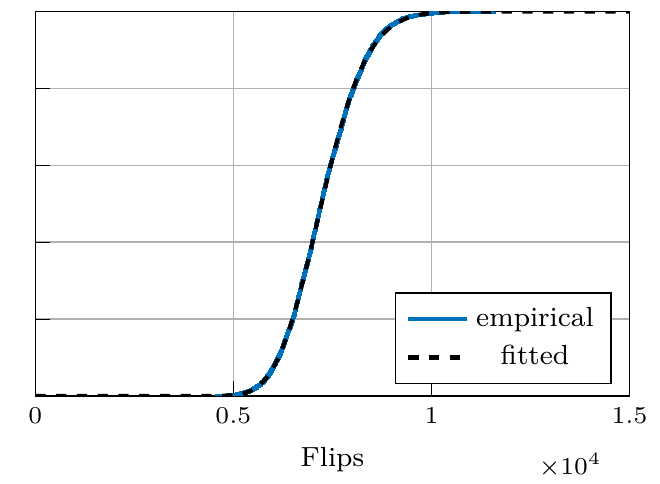}
		\end{subfigure}%
		\caption{%
			The ecdf and fitted cdf of the hardness distribution of instance~$A$~(left) and~$B$~(right).
		}
		\label{fig:cdf-ecdf}
	\end{figure}

	\paragraph{Study of the Left Tail} For the analysis, however, one should not confine oneself to this plot type. Although absolute errors can be observed easily, relative errors are more difficult to detect. Such a relative error may have a significant impact when used for decisions such as restarts. To illustrate this point, suppose that the actual probability of a run of length $\ell$ is $0.0001$. In contrast, the probability estimated based on a fit is $0.001$. As can be seen, the absolute error of $0.0009$ is small, whereas the relative error of $10$ is large. If one were to perform restarts after $\ell$ steps, the actual expected runtime would be ten times greater than the estimated expected runtime. Thus, the erroneous estimate of that probability would have translated into an unfavorable runtime. This example should illustrate the importance of checking the tails of a distribution for errors as well.

	The left tail, \ie the probabilities for very small values, can be checked visually by plotting 
	both axes logarithmically scaled.
	Thereby, probabilities for extreme events (in this case, especially easy instances) can be measured accurately. 
	The two instances~$A$ and~$B$ are being examined in this manner in \refFig{fig:log_cdf-log_ecdf}. As can be observed, the Johnson SB fit accurately predicts the probabilities associated with very short runs. For the other instances, Johnson SB distributions were mostly also able to accurately describe the probabilities for short runs. However, the behavior of the ecdf and the fitted Johnson SB distribution differed very slightly in a few instances.

	\begin{figure}[htb]
		\centering
		\begin{subfigure}[t]{.465\linewidth}
			\centering
			\includegraphics{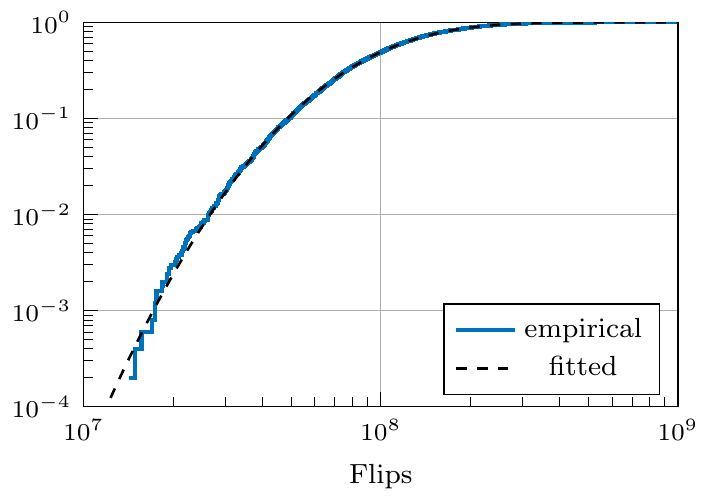}
		\end{subfigure}%
		\hfill%
		\begin{subfigure}[t]{.465\linewidth}
			\centering
			\includegraphics{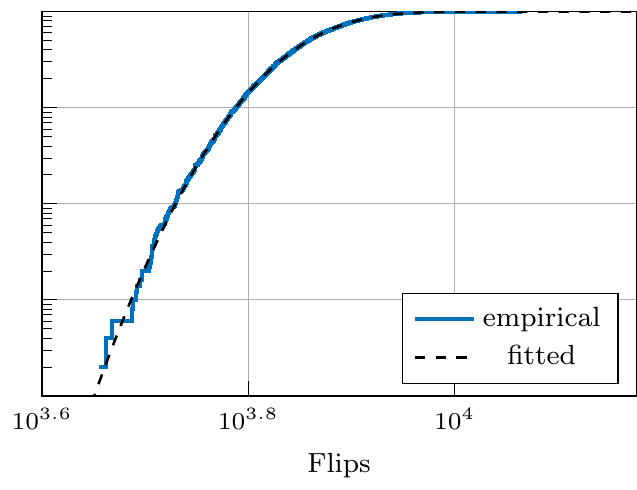}
		\end{subfigure}%
		\caption{%
			Logarithmically scaled ecdf and fitted cdf of instances $A$ (left) and $B$ (right).
		}
		\label{fig:log_cdf-log_ecdf}
	\end{figure}
	
	\paragraph{Study of the Right Tail} The probabilities for particularly hard instances should also be checked. 
	We can easily detect errors in the right tail if we plot the empirical survival function, \ie $\hat{S}_n(x) \coloneqq 1-\hat{F}_n(x)$, and the fitted survival function together on a graph with logarithmically scaled axes.
	\refFig{fig:sf-esf} illustrates this type of plot for the instances~$A$ and~$B$. Here, there is a discernible deviation between $A$ and~$B$. While for~$A$, the Johnson SB fit provides an accurate description of the probabilities for long runs, in the case of~$B$, the empirical survival function seems to approach~$0$ somewhat slower than the Johnson SB estimate. In the vast majority of cases, these extreme value probabilities are accurately reflected by the Johnson SB fit. In most other cases, the empirical survival function approaches~$0$ more slowly than the Johnson SB fit. Thus the likelihood of encountering an exceptionally hard instance is underestimated in these cases.

	\begin{figure}[htb]
		\centering
		\begin{subfigure}[t]{.465\linewidth}
			\centering
			\includegraphics{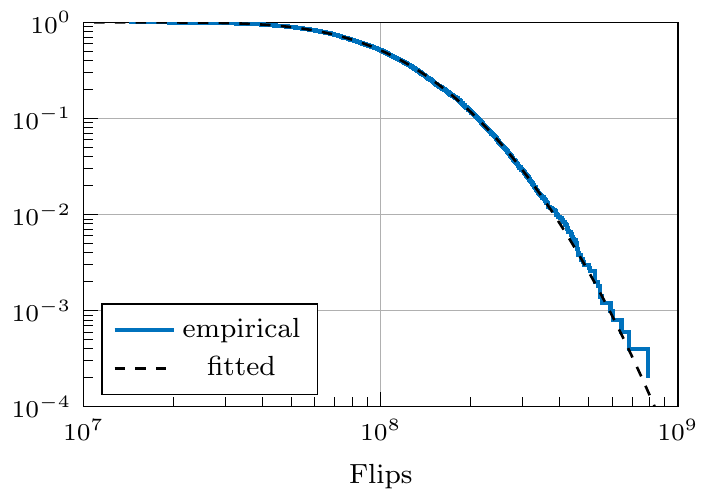}
		\end{subfigure}%
		\hfill%
		\begin{subfigure}[t]{.465\linewidth}
			\centering
			\includegraphics{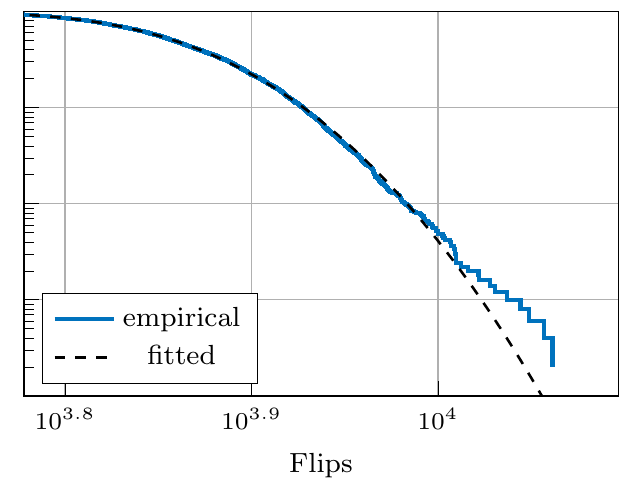}
		\end{subfigure}%
		\caption{%
			Logarithmically scaled empirical survival function and fitted survival function of instances~$A$~(left) and $B$ (right).
		}
		\label{fig:sf-esf}
	\end{figure}

	\paragraph{Goodness-Of-Fit Tests}
	So far, we have discussed the behavior of Johnson SB fits based on this visual inspection. We concluded that Johnson SB distributions seem well suited for describing the data. Next, we concretize this through the $\chi^2$-goodness-of-test that is executed for each instance.
	Subsequently, the probability~$p$ that such a value of the test statistic occurs under the assumption that the data follow a Johnson SB distribution (null hypothesis) is determined. 
	If the fit is poor, then a small $p$-value will occur. 
	We use a sufficiently high $p$-value as a heuristic whether the distribution assumption is reasonable.

	However, there is an obstacle that complicates statistical analysis by this method. As described, each of the $5000$~data points is obtained by first sampling $100$~runtimes of the corresponding instance and then calculating the mean. 
	This means that we do not work with the actual expected values, but only estimates, meaning our data is noisy. 
	If one were to apply the $\chi^2$-test to this noisy data, some cases would be incorrectly rejected, especially if the variance is large. To overcome this limitation, we use a bootstrap test, which is based on a test described by Cheng~\cite{cheng2017non}. This test is presented in \refAlg{algo:bootstrap}. We reject the distribution hypothesis for an instance if it fails the bootstrap test (\mbox{$p<0.05$}).

	\begin{algorithm}[htb]
		\textbf{Input:} (noisy) random sample $\vec{y}=(y_1, y_2, \dots, y_n)$, integer $N$, significance $\alpha\in (0,1)$\\
		\BlankLine
		
		$\hat{\theta}\leftarrow \algoformat{MLE}(\vec{y}, F)$, Johnson SB maximum likelihood estimation, $F$ is the Johnson SB cdf\\
		$\chi^2 \leftarrow \algoformat{ChiSquare}(\vec{y}, \hat{\theta})$, Chi-squared goodness-of-fit test statistic\\
		
		\SetKw{KwWithProb}{with probability}
		\SetKw{KwDo}{do}
		
		\SetKwFunction{Shuffle}{Shuffle}
		
		\For{$j=1$ to $N$}{
			$\vec{z} \leftarrow (z_1, \dots, z_n)$, where all $z_i$ are i.\,i.\,d.\ samples from the fitted Johnson SB
			distribution with parameters $\hat{\theta}$\\
			$\vec{z} \leftarrow \vec{z} + \vec{\xi}$, where $\vec{\xi}$ is sampled from an $n$-dimensional normal distribution\\
			$\hat{\theta}'\leftarrow \algoformat{MLE}(\vec{z}, F)$\\
			$\chi^2_j \leftarrow \algoformat{ChiSquare}(\vec{z}, \hat{\theta}')$\\
		}
		Let $\chi^2_{(1)}\leq \chi^2_{(2)}\leq \dots \leq \chi^2_{(N)}$ be the sorted test statistics.\\
		\lIf{$\chi^2_{( \lfloor (1-\alpha)\cdot N \rfloor )}<\chi^2$}{reject \textbf{else} accept}

		\caption{Bootstrap test for noisy data. The null hypothesis is that each datapoint~$x_i$ comes from an underlying Johnson SB distribution. Furthermore, each datapoint is obscured by additive noise~$\xi_i$. Thus, the only available observations are $\vec{y} = \vec{x} + \vec{\xi}$, \ie noisy data.}
		\label{algo:bootstrap}
	\end{algorithm}

	Briefly summarized, this test simulates how our data points were generated, assuming the null hypothesis. 
	Due to the central limit theorem, it is reasonable to assume that the initial data's sample mean originates from a normal distribution around the true expected value.  We use this assumption in the bootstrap test using a noise signal drawn from a normal distribution with expected value~$0$. Since each data point is the average of $100$~runs, the variance of this normal distribution is determined from the initial data and divided by $100$ (cf.\ central limit theorem).

	We now consider Johnson SB distributions' adequacy for describing \SRWA{} runtimes. The results of the statistical analysis are reported in \refTab{tab:stat} and can be found in~\cite{AllData}.
	The total of $4$~rejected instances may be attributed to so-called \mbox{type~1 errors}. 
	
	\begin{table}[htb]
		\centering
		\caption{Goodness-of-fit results for \Alfa{}+\SRWA{} over various problem domains. The \emph{rejected} row contains the number of instances where the Johnson SB distribution is not a good fit according to the bootstrap test at a significance level of $0.05$.
			To put these results into perspective, the second row contains each domain's total number of instances.
			Out of a total of 230 instances, 4 got rejected.}
		\begin{tabular}{l| c c c c c || c}
			& hidden & different chances & uniform & factoring & coloring & total\\ \hline
			rejected & 1 & 2 & 0 & 0 & 1 & 4\\
			$\#$ of instances & 20 & 120 & 25 & 33 & 32 & 230
		\end{tabular}
		\label{tab:stat}
	\end{table}

	For \probSAT{}, the situation appears to be slightly different. The results are summarized in \refTab{tab:stat_probSAT} and~\cite{AllData}.
	As can be seen, the distribution hypothesis was rejected for $7$ of the $55$~instances.
	This number can no longer be accounted for by type 1 errors at a significance level of~$0.05$.

	\begin{table}[htb]
		\centering
		\caption{Goodness-of-fit results for \Alfa{}+\probSAT{} over various hidden solution instance sizes.	
			The columns refer to the number of variables in the corresponding SAT instances.}
		\begin{tabular}{l| c c c c c c || c}
			number of variables & 50 & 100 & 150 & 200 & 300 & 800 & total \\\hline
			rejected & 0 & 0 & 2 & 1 & 3 & 1 & 7\\
			$\#$ of instances & 10 & 10 & 10 & 10 & 10 & 5 & 55
		\end{tabular}
		\label{tab:stat_probSAT}
	\end{table}

	Lastly, %
	for \YalSAT{} according to the bootstrap test, none of the total $10$~instances got rejected. %

	\paragraph{Distribution Conjectures}
	In summary, the presumption that Johnson SB distributions are the appropriate choice for describing runtimes has been reinforced for \Alfa{}+\SRWA{}. 
	Likewise, the choice of Johnson SB distributions also seems very reasonable for \Alfa{}+\YAL{}.
	This appears to be still plausible for \Alfa{}+\probSAT{}. 
	This leads us to:

	\begin{conjecture}[Strong Conjecture]
		\label{conj:strong}
		The runtime of $\Alfa$ with $\algoformat{SLS} \in \set{\texttt{SRWA}, \texttt{probSAT}, \texttt{YalSAT}}$ follows a Johnson SB distribution.
	\end{conjecture}

	If this 
	is true, then it would be intriguing that one can infer how modifying the base instance affects the hardness of instances.
	Simultaneously, 
	the Johnson SB distribution parameters
	also provide insight into how the hardness of the instance changes. For example, the location parameter~$\xi$ implies an inherent problem hardness that cannot be decreased regardless of the choice of the added clauses. 
	At the same time, $\xi$ also serves as a numerical description for the value of this intrinsic hardness.
	Using Bayesian statistics, it is possible to infer the parameters while the solver is running.
	These 
	estimations
	can, \eg be used to schedule restarts.
	This 
	leads to a scenario similar to that in~\cite{RHK02RestartPolicies}.

	\refCon{conj:strong} is a strong statement.
	However, even a slight deviation of the probabilities, for example, at the left tail, would render the strong conjecture invalid from a strictly mathematical point of view. Notably, the visual analyses revealed that the left tail's behavior, \ie for extremely short runs, is occasionally not accurately reflected by Johnson SB distributions. Conversely, the right tail, \ie the probabilities for particularly long runs, are usually either correctly represented by Johnson SB distributions or, occasionally, the corresponding probability approaches $0$~even more slowly. We, therefore, rephrase our conjecture in a weakened form. Our observations fit a class of distributions known as long-tail distributions defined purely in terms of their behavior at the right tail.

	\begin{definition}[\cite{foss2011introduction}]
		\label{def:long_tail}
		A positive, real-valued random variable $\ZV{X}$ is \introduceterm{long-tailed}, if and only if
		\begin{align*}
			\forall x\in \Rpos:\, \prob{\ZV{X}>x} > 0
			\quad \quad \textnormal{and} \quad \quad
			\forall y\in \Rpos: \lim\limits_{x\rightarrow \infty} \frac{\prob{\ZV{X} > x+y}}{\prob{\ZV{X}>x}} = 1.
		\end{align*}
	\end{definition}

	\begin{conjecture}[Weak Conjecture]
		\label{conj:weak}
		The runtime of \Alfa{} with $\algoformat{SLS} \in \set{\texttt{SRWA}, \texttt{probSAT}, \texttt{YalSAT}}$ is long-tailed.
	\end{conjecture}

	The fact of observing long-tailed ecdfs points towards the presence of a limiting process that is involved. Recall from \refLem{lem:LogNEmbeddedModelOfSB} that the Johnson SB distribution converges towards a lognormal distribution (for $\lambda \to \infty$, while it is sufficient for $\gamma$ to increase at a logarithmic rate \wrt~$\lambda$). This property is called embedded distribution.
	In our experiments, we observed that the parameter~$\gamma$ of the resulting Johnson SB fit is sufficiently high for convergence. 
	The Johnson SB distribution has bounded support, \ie all of its probability mass is concentrated on a finite interval. The endpoints of the support can be derived directly from the parameters.	
	Increasing the number of variables of the formula under consideration will additionally ensure that the density support $(a,b) = (\xi, \xi+\lambda)$ of the fitted distribution will increase since formulas with a higher number of variables will naturally be harder to solve. Hence, $\lambda$ must increase.
	Therefore, the Johnson SB distribution fits approach lognormal distributions. An illustration of this convergence is shown in \refFig{fig:sb_converge}.

	As a case in point for the actual involvement of such a convergence phenomenon, we repeated all tests above for the lognormal distribution. The visual inspection reveals that the lognormal distribution can also fit the data exceptionally well. For \Alfa{}+\SRWA{}, 5 out of 230 instances got rejected by the goodness-of-fit test; for \Alfa{}+\probSAT{}, 7 out of 55 instances got rejected; and for \Alfa{}+\YalSAT{}, 2 out of 10 instances got rejected. 
	Hence, it also seems very reasonable to use a lognormal distribution to describe the hardness.

	It should be noted that lognormal distributions have the long-tail property~\cite{foss2011introduction,nair2020fundamentals}. That is if the Strong Conjecture holds, the Weak Conjecture is implied (at least, after convergence). The reverse is, however, not true. In the next section, we show an important consequence in case the Weak Conjecture holds, \ie when the distribution is long-tailed.

	\subsection{Restarts Are Useful For Long-Tailed Distributions}

	If the runtimes are already lognormally distributed, then restarts are useful~\cite{Lorenz18RuntimeDistributions} in the following sense.

	\begin{definition}
		\label{def:RestartsUseful}
		Let $\ZV{X}$ be a random variable for the runtime of an SLS algorithm~$\mathcal{A}$ on some input.
		For $t > 0$, the algorithm~$\mathcal{A}_t$ is obtained by restarting~$\mathcal{A}$ after time~$t$ if no solution is found. %
		Letting~$\ZV{X_t}$ model the runtime of~$\mathcal{A}_t$, we say that restarts are \introduceterm{useful} if there is a $t > 0$ such that
		$
		\expectation{\ZV{X_t}} < \expectation{\ZV{X}}.
		$
	\end{definition}

	This section extends this result and mathematically proves that restarts are useful even if only the Weak Conjecture holds.
	This will be achieved by showing that restarts are useful for long-tailed distributions.
	For this section, we always implicitly use the
	natural assumption that the cdf~$F$ is continuous and strictly monotonically increasing. In this case, the quantile function~$Q$ is the inverse of~$F$.

	A condition for the usefulness of restarts, as defined in \refDef{def:RestartsUseful}, was proven in~\cite{Lorenz18RuntimeDistributions}. For the following, recall the concept of quantile functions (\refDef{def:cdf_quantilemain}).
	We show the result using the following theorem.

	\begin{definitionandtheorem}[\cite{Lorenz18RuntimeDistributions}]
		\label{theo:sufficient}
		Let $\ZV{X}$ be a positive, real-valued random variable having quantile function~$Q$. Let
		\[
		R(p,\ZV{X}) \coloneqq
		(1-p)\cdot \frac{Q(p)}{\expectation{\ZV{X}}}+\frac{\integralLow[u]{0}{p}{Q(u)}}{\expectation{\ZV{X}}}.
		\]
		Then restarts are useful if and only if there is a quantile $p \in (0,1)$ such that
		\begin{equation*}
			R(p,\ZV{X})<p.
		\end{equation*}
	\end{definitionandtheorem}

	Even if the quantile function~$Q$ and the expected value are unknown, $R(p,\ZV{X})$ can be characterized for large values of~$p$.
	
	\begin{lemma}
		\label{lem:exp_infinite_restart}
		Consider a positive, real-valued random variable~$\ZV{X}$ with pdf~$f$ and quantile function~$Q$ such that~\mbox{$\expectation{\ZV{X}}<\infty$}.
		Also, assume that the limit $\lim_{t\rightarrow\infty} t^2 \cdot f(t)$ exists.
		Then,
		\begin{align*}
			\lim\limits_{p\rightarrow 1} R(p,\ZV{X})
			= \lim\limits_{p\rightarrow 1} \, \left((1-p)\cdot \frac{Q(p)}{\expectation{\ZV{X}}}+\frac{\integralLow[u]{0}{p}{Q(u)}}{\expectation{\ZV{X}}}\right)
			= 1.
		\end{align*}
	\end{lemma}

	For the proof of \refLem{lem:exp_infinite_restart}, we will need the inverse function theorem. This theorem roughly states that a continuously differentiable function $\varphi$ is invertible in a neighborhood of any point~$a$ at which $\varphi^\prime(x)$ does not vanish~\cite{rudin1964principles}.

	\begin{theorem}[Inverse function theorem \cite{bartle2000introduction}]
		\label{thm:IfThm}
		Let $I \subseteq \R$ be an interval and let $\varphi \colon I \to \R$ be continuous and strictly monotone on $I$. 
		Then, there is an inverse function $\psi \coloneqq \varphi^{-1}$ defined on $J \coloneqq \varphi(I)$ that is continuous and strictly monotone.
		If $\varphi$ is differentiable at $a \in I$ and $\varphi^\prime(a) \neq 0$, then $\psi$ is differentiable at $b \coloneqq \varphi(a)$ and 
		\[
		\big(\varphi^{-1}\big)^\prime(b) 
		= \psi^\prime(b)
		= \frac{1}{\varphi^\prime(a)} 
		= \frac{1}{\varphi^\prime\big(\psi(b)\big)}.
		\]
		
	\end{theorem}

	\begin{proof}[Proof of \refLem{lem:exp_infinite_restart}]
		In the following, let~$F$ and~$f$ be the cdf and pdf of~$\ZV{X}$, respectively (see \refDef{def:cdf_quantilemain}).
		We start by specifying the derivative of~$Q$ with respect to~$p$ as a preliminary consideration.
		From $F = Q^{-1}$ and the application of the inverse function theorem~\ref{thm:IfThm} (letting $\psi = Q$ and $\varphi = F$), it follows that
		\begin{align}
			\label{eq:quantil_deriv}
			Q'(p) \coloneqq \frac{\D}{\D p} Q(p)=\frac{1}{f\bigl(Q(p)\bigr)}. 
		\end{align}
		
		As the first step in our proof, we consider the limit
		of the second summand of~$R(p,\ZV{X})$, \ie of the term~$\integralLow[u]{0}{p}{Q(u)} \big/ \expectation{\ZV{X}}$.
		This value can be determined using integration by substitution with~\mbox{$x=Q(u)$} followed by applying the change of variable method with~\mbox{$p=F(t)$}:
		\begin{align*}
			\lim\limits_{p\rightarrow 1}\frac{\integralLow[u]{0}{p}{Q(u)}}{\expectation{\ZV{X}}}
			=\lim\limits_{p\rightarrow 1}\frac{\integralLow{0}{Q(p)}{x\cdot f(x)}}{\expectation{\ZV{X}}}
			=\lim\limits_{t\rightarrow \infty} \frac{\integralLow{0}{t}{x\cdot f(x)}}{\expectation{\ZV{X}}} =1.
		\end{align*}
		The last equality holds because the	numerator matches the definition of the expected value.

		Next, we examine the limit of $(1-p)Q(p)/\expectation{\ZV{X}}$.
		When $\ZV{X}$ has finite support, \ie when there exists an $x \in \R$ with $F(x) = 1$. Then, \mbox{$\lim_{p \to 1} (1-p) Q(p) = 0$} follows from the definition of the quantile function.

		More care needs to be taken in the case when $F(x) < 1$ holds for all $x \in \R$. In this case, we have $\lim_{p \to 1} Q(p) = \infty$. Hence, to examine $\lim_{p \to 1} (1-p)Q(p)/\expectation{\ZV{X}}$, we apply
		L'Hospital's rule twice and use the change of variable method with \mbox{$p=F(t)$} to obtain
		\begin{align*}
			\lim\limits_{p\rightarrow 1} (1-p)\cdot Q(p)
			= \lim\limits_{p\rightarrow 1} Q(p)^2\cdot f\bigl(Q(p)\bigr)
			= \lim\limits_{t\rightarrow \infty} t^2\cdot f(t).
		\end{align*}
		It is well-known that if \mbox{$\liminf_{t\rightarrow \infty} t^2 \cdot f(t) >0$} were to hold, then the expected value $\expectation{\ZV{X}}$ would be infinite (this statement is, for example, implicitly given in~\cite{foss2011introduction}).
		This would contradict the premise of the lemma; therefore, \mbox{$\liminf_{t\rightarrow \infty} t^2 \cdot f(t) = 0$}.
		Moreover, since, by assumption, $\lim_{t\rightarrow\infty} t^2 \cdot f(t)$ exists, we may conclude that 
		\[
		\lim\limits_{t\rightarrow\infty} t^2 \cdot f(t) = \limsup\limits_{t\rightarrow\infty} t^2 \cdot f(t) = \liminf\limits_{t\rightarrow\infty} t^2 \cdot f(t) = 0. \qedhere
		\]
	\end{proof}
	
	A frequently used tool for describing distributions is the hazard rate function.
	
	\begin{definition}[\cite{rausand2020system}]
		\label{def:hazard_rate}
		Let $\ZV{X}$ be a positive, real-valued random variable having cdf~$F$ and pdf~$f$.
		The \introduceterm{hazard rate function} $r \colon \Rpos \to \Rpos$ of $\ZV{X}$ is given by %
		\begin{align*}
			r(t) \coloneqq \frac{f(t)}{1-F(t)}.
		\end{align*}
	\end{definition}
	
	There is an interesting relationship between the long-tail property and the hazard rate function's behavior.

	\begin{lemma}[\cite{nair2020fundamentals}]
		\label{lem:appendixlemma}
		Let $\ZV{X}$ be a positive, real-valued random variable with hazard rate function~$r$ such that the limit \mbox{$\lim_{t\rightarrow\infty}r(t)$} exists.
		Then, the following statements are equivalent:
		\begin{enumerate}[(a)]
			\item $\ZV{X}$ is long-tailed. \label{proof:long_tail_1}
			\item $\LIM{t} r(t)=0$.
			\label{proof:long_tail_3}
		\end{enumerate}
	\end{lemma}

	With the help of these preliminary considerations, we are now ready to show that restarts are useful for long-tailed distributions.
	It should be noted that the conditions of this following theorem are not restrictive since all naturally occurring long-tail distributions satisfy these conditions (see also~\cite{nair2020fundamentals}). To be more precise, to the best of our knowledge, all \emph{named} continuous long-tailed distribution do fulfill the requirements of the following theorem (there are only pathological examples that can be constructed that do not fulfill the requirements).

	\begin{theorem}
		\label{theo:long_tail_restarts}
		Consider a positive, long-tailed random variable~$\ZV{X}$ with continuous pdf~$f$ and hazard rate function~$r$. 
		Also, assume that either \mbox{$\expectation{\ZV{X}}=\infty$} holds or the limits \mbox{$\lim_{t\rightarrow\infty}r(t)$} and \mbox{$\lim_{t\rightarrow\infty} t^2 \cdot f(t)$} both exist.
		In both cases, restarts are useful for $\ZV{X}$.
	\end{theorem}
	
	\begin{proof}
		Let $F$ be the cdf and $Q$ the quantile function of $\ZV{X}$.
		According to \refTheo{theo:sufficient}, restarts are useful if and only if
		\begin{align}
			\label{eq:RecallRestartUsefulIneq}
			(1-p)\cdot \frac{Q(p)}{\expectation{\ZV{X}}} + 	\frac{1}{\expectation{\ZV{X}}} \cdot \integralLow[u]{0}{p}{Q(u)} < p
		\end{align}
		for some $p \in (0,1)$. Let us consider two cases.
		
		First, consider the case where the expected value~$\expectation{\ZV{X}}$ is infinite. 
		Let $p \in (0,1)$ be such that $Q(p) < \infty$. Since $\expectation{\ZV{X}} = \infty$, it immediately follows that $Q(p) / \expectation{\ZV{X}} = 0$. Moreover, we also have $\integralLow[u]{0}{p}{Q(u)} \leq p \cdot Q(p) < \infty$. Hence, the left side of \refIneq{eq:RecallRestartUsefulIneq} is zero, and the inequality is obviously satisfied.
		Thus, the statement follows.
		
		For the second case, we assume that $\expectation{\ZV{X}} < \infty$ and that
		both \mbox{$\Lim{t}r(t)$} and \mbox{$\Lim{t} t^2 \cdot f(t)$} exist. \refEq{eq:quantil_deriv} can now be used to calculate the following derivative:
		\begin{align*}
			\frac{\D}{\D p} \Big( R(p,\ZV{X}) - p \Big)
			= \frac{\D}{\D p} \left((1-p)\cdot \frac{Q(p)}{\expectation{\ZV{X}}}+\frac{\integralLow[u]{0}{p}{Q(u)}}{\expectation{\ZV{X}}}-p\right)
			= \frac{1-p}{\expectation{\ZV{X}}\cdot f\bigl(Q(p)\bigr)}-1.
		\end{align*}
		Consider the limit of this expression for $p\rightarrow 1$. Once again, the change of variable method is applied with $p=F(t)$, resulting in
		\begin{align*}
			\lim\limits_{p\rightarrow 1} \frac{1-p}{\expectation{\ZV{X}}\cdot f\bigl(Q(p)\bigr)}-1 
			= \lim\limits_{t\rightarrow \infty} \frac{1-F(t)}{\expectation{\ZV{X}}\cdot f(t)}-1 
			= \lim\limits_{t\rightarrow \infty} \frac{1}{\expectation{\ZV{X}}\cdot r(t)}-1.
		\end{align*}
		By assumption, $\ZV{X}$ has a long-tail distribution, and the limit of \mbox{$\Lim{t}r(t)$} exists. For this reason, \mbox{$\lim_{t\rightarrow \infty}r(t)=0$} follows as a result of \refLem{lem:appendixlemma}. Furthermore, since $\expectation{\ZV{X}} < \infty$ holds, we may conclude that  	
		\begin{align}
			\label{eq:long_tail_exp_infinity}
			\lim\limits_{p\rightarrow 1} \frac{1-p}{\expectation{\ZV{X}}\cdot f\bigl(Q(p)\bigr)}-1 
			= \lim\limits_{t\rightarrow \infty} \frac{1}{\expectation{\ZV{X}}\cdot r(t)}-1 = \infty.
		\end{align}
		The condition from \refTheo{theo:sufficient} can be rephrased in such a way that restarts are useful if and only if 
		\[
		R(p,\ZV{X}) - p < 0.
		\]
		According to \refLem{lem:exp_infinite_restart}, the left-hand side of this inequality approaches $0$ for $p\rightarrow 1$.
		However, as shown in \refEq{eq:long_tail_exp_infinity}, the derivative of $R(p,\ZV{X})-p$ approaches infinity for $p\rightarrow 1$. These two observations imply that there is a $p\in (0,1)$ satisfying $R(p,\ZV{X}) - p < 0$.
		Consequently, restarts are useful for~$\ZV{X}$.
	\end{proof}

	With the help of this theorem, we obtain the following corollary of the Weak Conjecture.

	\begin{conjecture}
		Restarts are useful for \Alfa{} with $\algoformat{SLS} \in \set{\texttt{SRWA}, \texttt{probSAT}, \texttt{YalSAT}}$.
	\end{conjecture}
	
	If \refCon{conj:weak} is true, then this statement follows immediately by \refTheo{theo:long_tail_restarts}.

	\section{Theoretical Justifications for the Johnson SB Conjecture}
	\label{sec:theory}
	
	\renewcommand{\D}{\partial}
	\renewcommand{\Clauses}{\ensuremath{\ZV{L}}}

	Up to this point, we have established that Johnson SB distributions accurately describe the runtime behavior of the \Alfa{} algorithm as demonstrated for the solvers \SRWA{}, \probSAT{}, and \YalSAT{}. This observation was derived from extensive empirical investigations. This section provides a theoretical justification for why the runtime distributions are Johnson SB distributed in the special case of \SRWA{} as the SLS component of \Alfa{}. We focus on \SRWA{} because this algorithm is best suitable for purely theoretical analyses (as witnessed by the worst-case analysis conducted by Schöning in \cite{Schoening02AProbabilisticAlgorithm}). Furthermore, it is a special case of \probSAT{}.
	For convenience, \SRWA{} is presented in \refAlg{algo:schoening}.
	
	\SetAlgoVlined
	\begin{algorithm}[htb]
		\DontPrintSemicolon
		\SetKwComment{tcp}{$\triangleleft$\,}{}%
		
		\let\oldnl\nl %
		\newcommand{\nonl}{\renewcommand{\nl}{\let\nl\oldnl}} %

		\SetKwInOut{Input}{Input}
		\Input{Boolean formula $G$, \textbf{Promise:} $G \in \SAT$}

		\BlankLine
		\everypar={\nl}
		
		\While{True}{
			
			Choose a random assignment $\alpha$ \label{alg:srwa:assignment} \;
			
			\For(\tcp*[f]{restart mechanism}){$j=1$ \KwTo $t_{\textnormal{restart}}$}{
				\label{alg:srwa:restarttime}
				\everypar={\nl}
				\lIf{$G \alpha = 1$}{\KwRet $\alpha$}
				
				Uniformly at random choose a clause $\clstd \in G$ with $\applyassi{\clstd}{\alpha} = 0$ \label{alg:srwa:clause}\; %
				
				Uniformly at random choose a literal $\litstd \in \clstd$ \label{alg:srwa:literal}\; %
				
				Perform the flip $\alpha \coloneqq \FlipVarInAssi{\litstd = 1}{\alpha}$
				\label{alg:srwa:flip}\; %
				
			}
			
		}
		\caption{Schöning's original random walk algorithm \SRWA{}~\cite{Schoening02AProbabilisticAlgorithm}. Line~\ref{alg:srwa:restarttime} takes care of restarting the search after $t_{\textnormal{restart}}$ flips. Then, a new random assignment is chosen.}
		\label{algo:schoening}
	\end{algorithm}

	\subsection{Proof Overview}
	
	Let us begin by providing an overview of the organization of \refSec{sec:theory}. This overview will also function as a rough proof sketch. The overall idea of the proof is to study which random variables make up the expected runtime (called $\ZV{P}, \ZV{Q}$, and $\ZV{R}$ in the following) and then, subsequently, analyze these random variables. We succeed in showing that these three random variables are indeed approximately Johnson SB distributed. We have provided more details of the proof in the following overview:
	\begin{description}
		\item[\refSec{sec:glossaryofnotation}] To increase readability, we provided an overview of all used notation in a glossary. We also discuss notational convention in this section. %
		\item[\refSec{sec:analysis_rtd}] We start the proof by showing that the expected runtime (as measured in the number of flips), $\ERL$, on the extended instance $F \cup \Clauses$ can be analyzed by separating the expected value into two components. The first component, $\FiniteCase$, takes care of the case where at some point during the run of the algorithm, a clause of $\Clauses$ will be selected by Schöning's algorithm (see line~\ref{alg:srwa:clause} in \refAlg{algo:schoening}). The second component, $\InfCase$, takes care of the case when the formula is solved solely on the initial formula~$F$. We analyze each component in a separate subsection.
		\begin{description}
			\item{\emph{\refSec{sec:TheInfiniteCase}}} We analyze the term $\InfCase$ that gives the expected number of flips in case no clause of $\Clauses$ will ever get selected. We show that this term consists of one random variable~$\ZV{Q}$.
			\item{\emph{\refSec{sec:TheFiniteCase}}} We show that the term $\FiniteCase$ that gives the expected number of flips in the case where $\Clauses$ is involved in the solving process contains three random variables that we call $\ZV{P}, \ZV{Q}$ and $\ZV{R}$ (plus the expected value $\ERL$ after one flip has taken place).
			\item{\emph{\refSec{sec:CombiningBothCases}}} We combine the two cases in one single equation.
		\end{description}
		\item[\refSec{sec: AnalysisOfTheRVs}] We analyze the random variables $\ZV{P}, \ZV{Q}$ and $\ZV{R}$ that we have obtained in the last section and find that each is asymptotically Johnson SB distributed.
	\end{description}

	\subsection{Glossary of Notations and Notational Conventions}
	\label{sec:glossaryofnotation}

	For the following sections, the reader might want to refer to the following glossary to look up terminology that is introduced in the following subsections. %
	Let us also emphasize that all our notations abide by the following convention.

	\begin{convention}
		\label{convention:probs}
		In the next section, we will consider a random set~$\ZV{\Clauses}$ of clauses. For the sake of clarity, we print random variables depending on~$\ZV{\Clauses}$ in bold font (\eg $\ZV{P}, \ZV{Q}, \ZV{R}$), whereas constants are not printed in bold font (\eg $C_1$, $C_2$, or $\probcond{\text{the number of flips to solve } F \text{ is } i}{\SRWA \text{ is started in } \alpha}$, etc.).
		
		We wish to especially highlight that some random variables describe probabilities since they depend on the random set~$\ZV{\Clauses}$. We use a subscript $\ZV{\Clauses}$ and boldface and denote this by~$\ZV{\ProbopZV_{\Clauses}} \left( \cdot \right)$. For example, we will write 
		\[
		\ZV{\ProbopZV}_{\Clauses}(\text{the first time a clause from } \ZV{\Clauses} \text{ is chosen is in iteration } c+1 \mid \SRWA \text{ was started in } \alpha, \: \ZV{\Clauses} ).
		\]
		To correctly interpret this notation, we refer to \refDef{def:KomplizierteWahrscheinlichkeit}. On the other hand, some probabilities are constants not depending on~$\Clauses$, denoted by the notation~$\prob{\cdot}$.
		
		We use the same principle for the expectation operator.
	\end{convention}

	Upon first reading the paper, the reader might skip this glossary, as all definitions are introduced and explained in the main body of the following section. In brackets, we indicate where the full definition can be found.

	\begin{description}
		\item[$F$] \gloss{the original formula \SRWA{} is trying to solve}
		
		\item[$F^\prime$] \gloss{the modified formula, given by $F^\prime \coloneqq F \cup \Clauses$}
		
		\item[$\Clauses$] \gloss{random set of logically equivalent clauses that gets added to $F$}
		
		\item[$\mathbb{L}$] \gloss{set of some logically equivalent clauses \wrt $F$}
		
		\item[$\W(G)$] \gloss{random variable for the runtime in flips of \SRWA{} on instance~$G$ (\refDef{def:NumberOfFlips})}
		
		\item[$\ERL$] \gloss{expected runtime $\expectationcondZV{\W(F\cup \Clauses)}{\Clauses}$ of \SRWA{} (in flips) on the extended instance $F \cup \Clauses$ (\refDef{def:ERLEingefuehrt})}
		
		\item[$A(\alpha)$] \gloss{event that the initial assignment of \SRWA{} is~$\alpha$ (page~\refDef{def:Aalpha})}
		
		\item[$\econd{\cdot}$] \gloss{expected runtime of \SRWA{} (in flips) on $F \cup \Clauses$ subject to the conditions given in brackets (page~\pageref{page:ERLMitBrackets})}
		
		\item[$\ZV{\ProbopZV_{\Clauses}} \left( \cdot \right)$] 
		\gloss{$\coloneqq \probcondZV{\cdot}{\Clauses}$ if $B$ is some event (page~\pageref{page:PLBrackets})}
		
		\item[$\ZV{\ProbopZV}_{\Clauses}(\cdot \! \mid \! B)$]
		\gloss{$\coloneqq \probcondZV{\cdot}{B, \Clauses}$ if $B$ is some event (page~\pageref{page:PLBrackets})}
		
		\item[$\NoSel{c}$] \gloss{event that 
			the first time \SRWA{} selects a clause of~$\ZV{\Clauses}$ in line~\ref{alg:srwa:clause} is in the $(c+1)$-st iteration (\refDef{def:FirstSelAndNeverSel})}
		
		\item[$\NeverSel$] \gloss{event that \SRWA{} never chooses a clause of $\Clauses$ and solves the formula only using clauses from $F$ (\refDef{def:FirstSelAndNeverSel})}
		
		\item[$\Sel{c}$] \gloss{indicator variable being $1$ if and only if a clause in $\Clauses$ gets selected in the $c$-th iteration of \SRWA{} (\refDef{def:ScIndVar})}
		
		\item[$\Wa{c}{\beta}$] 
		\gloss{event that \new{\algoformat{SRWA}} (started from $\alpha$) ends up in assignment~$\beta$ after performing $c$~flips (\refDef{def:FromAlphaToBetaInC})} 
		
		\item[${\UNSATunder[G]{\beta}}$] \gloss{$\coloneqq \setdescr{D \in G}{ \applyassi{D}{\beta} = 0}$ (\refDef{def:UNSAT})}
		
		\item[${\UNSATin[G]{\beta}{x}}$] \gloss{$\coloneqq \setdescr{D \in G}{ \applyassi{D}{\beta} = 0 \text{ and } x \in \Vars{D}}$ (\refDef{def:UNSAT})}
		
		\item[${\UNSATnotin[G]{\beta}{x}}$] \gloss{$\coloneqq \setdescr{D \in G}{ \applyassi{D}{\beta} = 0 \text{ and } x \notin \Vars{D}}$ (\refDef{def:UNSAT})}
		
		\item[$\ZV{V}$] \gloss{$\coloneqq \Vars{\UNSATunder[\Clauses]{\beta}}$, the set of clauses in~$\ZV{\Clauses}$ that are falsified by $\beta$ (\refDef{def:VZV})}
		
		\item[$\Flipa{c}{\beta}{x}$] 
		\gloss{event that in the execution of \SRWA{} (started with assignment~$\alpha$), at the \new{beginning of the} $c$-th iteration, the current assignment is $\beta$. Furthermore, the next flip will flip variable~$x$. (\refDef{def:FlipAlphacbetax})}
		
		\item[$C_1(\alpha,i)$]
		\gloss{$\coloneqq \Probcond{ \W(F) = i }{ \InitAssi{\alpha} }$, \ie some constant independent of $\Clauses$, namely the probability that $F$ gets solved with $i$ flips when \SRWA{} is started from~$\alpha$ (\refLem{lem:LemmaWithSimpleEquation})}
		
		\item[$C_2(\alpha,k,\gamma)$]
		\gloss{$\coloneqq \probcondResize{ \Wa{k-1}{\gamma} }{\bigcap_{j=1}^{k-1} \singleevent{ \Sel{j} = 0 }, \InitAssi{\alpha}}$, \ie some constant independent of $\Clauses$, namely the probability that \SRWA{} takes the random walk from the initial assignment~$\alpha$ to $\gamma$ in $k-1$ flips, under the condition that no clause of $\ZV{\Clauses}$ was touched (\refLem{lem:LemmaWithCZwoUndQ})} 
		
		\item[$C_3(\alpha,c,\beta)$]
		\gloss{$\coloneqq \Probcond{ \Wa{c}{\beta} }{\NoSel{c}, \InitAssi{\alpha}}$, \ie some constant independent of $\Clauses$, namely the probability that \SRWA{} takes the random walk from the initial assignment~$\alpha$ to $\gamma$ in $c$ flips, under the condition that the first clause of~$\ZV{\Clauses}$ gets selected in iteration~$c+1$ (Proposition~\ref{recap:Finite})} 
		
		\item[$C_4(\alpha,c,\gamma)$]
		\gloss{$\coloneqq C_2(\alpha,k+1,\gamma)$, \ie some constant independent of $\Clauses$ (Proposition~\ref{recap:Finite})} 
		
		\item[$\ZV{P}$] \gloss{$\coloneqq \PACond{ \Flipa{c}{\beta}{x} }{ \Wa{c}{\beta}, \NoSel{c} }$, 
			Johnson SB distributed random variable
			(page~\pageref{page:IntroVonZVP})}
		
		\item[$\ZV{Q}$] \gloss{$\coloneqq \PACondResize{ \Sel{k} = 0 }{ \Wa{k-1}{\gamma} }$, Johnson SB distributed random variable (\refLem{lem:LemmaWithCZwoUndQ})} 
		
		\item[$\ZV{R}$] \gloss{$\coloneqq \PACondResize{ \Sel{c+1} = 1 }{ \Wa{c}{\gamma} }$, Johnson SB distributed random variable (\refSec{sec:CaseSk0})}
	\end{description}

	\subsection{Analysis of the Runtime Distribution of the Algorithm}
	\label{sec:analysis_rtd}

	Before beginning the analysis, let us quickly recall the setting of \refSec{sec:evidence}:
	The input of the algorithm \Alfa{} is a Boolean CNF formula~$F$, and the promise that this formula is indeed satisfiable.
	The algorithm then randomly generates some 
	set~$\ZV{\Clauses}$ of clauses that are logical implications of the original formula.
	Then, some SLS solver (in this section, Schöning's random walk algorithm, again abbreviated with \SRWA{}) is called on~$\ZV{F^\prime} \coloneqq F \cup \Clauses$.

	Note that the parameter~$t_{\textnormal{restart}}$ in \refAlg{algo:schoening} controls the restart mechanism of \SRWA{}. In other words, \SRWA{} chooses a new random assignment after every $t_{\textnormal{restart}}$ flips. Initially, we consider the case $t_{\textnormal{restart}} = \infty$ implying that no restarts are performed. This choice simplifies the analysis slightly. However, the observations and results can be extended to the case in which restarts are performed at the cost of sacrificing clarity of exposition. We briefly explain how to adapt our arguments to the case with restarts in \refSec{app:AllowingRestarts} of the appendix.

	The aim at the beginning of this section is to establish that the expected runtime on the extended instance $F \cup \Clauses$ can be analyzed by considering two components: The first component~$\FiniteCase$, taking care of the case where the algorithm will at some point select a clause of~$\Clauses$; and the other component~$\InfCase$ of the case where such a clause is never selected.

	We need the following two definitions to make the notion of expected runtime more precise.

	\begin{definition}
		\label{def:NumberOfFlips}
		We let $\W(G)$ denote the number of flips of \SRWA{} on an instance~$G$ until a satisfying assignment is found. Since \SRWA{} is a Las Vegas algorithm, $\W(G)$ is a random variable.
	\end{definition}
	
	We frequently refer to the number of flips required to find a satisfying assignment as the runtime of \SRWA{}. We aim to find the asymptotic distribution of the expected runtimes of \algoformat{SRWA} when the algorithm is provided with a random set~$\Clauses$ from \Alfa{}. We capture this random choice of additional clauses with the following definition. To understand the term $\expectationcondZV{\W(F\cup \Clauses)}{\Clauses}$ in the definition, the reader might refer back to \refDef{def:KomplizierterErwartungswert} and Example~\ref{ex:EFlipsFundLGegebenL}.
	
	\begin{definition}
		\label{def:ERLEingefuehrt}
		Let $F$ be some SAT instance and let $\mathbb{L}$ be a set of clauses such that for all $R \subseteq \mathbb{L}$, the formula $F \cup R$ is logically equivalent to $F$. Furthermore, let $\Clauses$ be the random subset of $\mathbb{L}$. In the following, the random variable~$\ERL$ denotes the expected runtime $\expectationcondZV{\W(F\cup \Clauses)}{\Clauses}$ of \SRWA{} (in flips) on $F \cup \Clauses$.
		That is,
		\[
		\ERL(R) = \expectationcond{\W(F \cup \ZV{L})}{\ZV{L} = R} = \expectation{\W(F \cup R)} \in \N_{0}, \text{ where } R \subseteq \mathbb{L}.
		\]
	\end{definition}
	
	In this section, $\mathbb{L}$ may be arbitrarily chosen as long as it only contains clauses implied by $F$. In the first part of this section, any stochastic process can create the random set~$\Clauses$. Later on, in \refsec{sec: AnalysisOfTheRVs}, we fix a generating model for~$\Clauses$. 
	As~$\ZV{\Clauses}$ is being randomly selected in \Alfa{}, $\ZV{\Clauses}$ is a random set (denoted in bold). Thus, the expected runtime~$\ERL \colon \mathcal{P}(\mathbb{L}) \to \N_0$ is also a random variable.
	
	\label{page:ERLMitBrackets}
	Furthermore, we frequently work with further restrictions on $\ERL$, such as a condition for the initial assignment. These restrictions result in a conditional expectation. 
	
	\begin{notation}
		\label{not:ERLmitZusatz}
		We denote additional conditions in round brackets, \ie $\ERL(\cdot)$. 
	\end{notation}

	\begin{example}
		Let $\InitAssi{\alpha}$ be the event that the initial assignment chosen by \SRWA{} is some fixed assignment $\alpha$ (cf. line~\ref{alg:srwa:assignment} in \refAlg{algo:schoening}). Then, $\ECond{\InitAssi{\alpha}}$ denotes the conditional expectation $\expectationcondZV{\W(F \cup \Clauses)}{\InitAssi{\alpha}, \Clauses}$.
	\end{example}

	We begin our analysis of the runtime distribution of \Alfa{} by applying the law of total expectation (LTE) to~$\ERL$ when conditioning on the randomly chosen initial assignment.
	
	\begin{definition}
		\label{def:Aalpha}
		Let $A(\alpha)$ denote the event that the initial starting assignment chosen in line~\ref{alg:srwa:assignment} of \SRWA{} is $\alpha$.
	\end{definition}

	Using this definition, the LTE yields
	\begin{align}
		\label{eq:FirstApplicationOfLTE}
		\ERL = \sum_{\alpha \in \assi} \Prob{\InitAssi{\alpha}} \cdot \ECond{\InitAssi{\alpha}},
	\end{align}
	where $\ECond{\InitAssi{\alpha}}$ denotes the expected runtime of \SRWA{} (in flips) on a given extended problem instance~\mbox{$F \cup \Clauses$} under the condition that the algorithm picks $\alpha$ as the initial assignment.
	Letting $n \coloneqq |\!\operatorname{Vars}(\ZV{F^\prime})| = |\!\operatorname{Vars}(F)|$, one can
	notice that $\prob{\InitAssi{\alpha}} = \frac{1}{2^n}$ for all $\alpha \in \assi$, \ie this probability is independent of the random set~$\ZV{\Clauses}$.

	In the following, we concentrate on analyzing the expression~$\ECond{\InitAssi{\alpha}}$ appearing in \refEq{eq:FirstApplicationOfLTE}. The following definition is used to distinguish between the cases of whether the algorithm will use a clause of $\Clauses$ in the solution or not.
	
	\begin{definition}
		\label{def:FirstSelAndNeverSel}
		Let $\NoSel{c}$ denote the event that the first time \SRWA{} selects a clause of~$\ZV{\Clauses}$ in line~\ref{alg:srwa:clause} is in iteration~$c+1$ (\ie in iteration $1, \dots, c$ a clause of $F$ is chosen).
		Similarly, we let $\NeverSel$ denote the event that \SRWA{} never chooses a clause of~$\ZV{\Clauses}$, \ie the algorithm solves the instance~$\ZV{F^\prime}$ using only clauses from~$F$. 
	\end{definition}

	\label{page:PLBrackets}
	Similar to $\ERL$, denoting an expected runtime depending on~$\ZV{\Clauses}$, we also deal with probabilities depending on $\Clauses$. Following Convention~\ref{convention:probs}, we write $\Clauses$ as subscript and use bold font: 
	
	\begin{notation}
		\label{not:PLmitZusatz}
		The notation $\ZV{\ProbopZV}_{\Clauses}(\cdot)$ 
		will be used as a
		shorthand for $\probcondZV{\cdot}{\Clauses}$, and $\ZV{\ProbopZV}_{\Clauses}(\cdot \! \mid \! B) \coloneqq \probcondZV{\cdot}{B, \Clauses}$ if $B$ is some event.
	\end{notation}
	
	\begin{example}
		Again, since $\Clauses$ is a random set, $\ZV{\ProbopZV}_{\Clauses}(\cdot)$ is a random variable.
		For example, we will write 
		\[
		\PACond{\NoSel{c}}{}
		\]
		as a shorthand for the conditional probability
		\[
		\probcondZV{\NoSel{c}}{A(\alpha), \Clauses},
		\]
		\ie for the (random variable) probability that \SRWA{} picks a clause from~$\ZV{\Clauses}$ in iteration~$c+1$ for the first time, given that it was started from~$\alpha$, and dependent on the concrete realization of the random set~$\ZV{\Clauses}$. 
		Similarly, as introduced above, the notation 
		\[
		\EACond{\NoSel{c}}
		\]
		should be understood as 
		\[
		\expectationcondZV{\W(F\cup \Clauses)}{\NoSel{c}, A(\alpha), \Clauses}.
		\]
		The definitions of $\ERL(\cdot)$ and $\ZV{\ProbopZV}_{\Clauses}(\cdot)$ in Notation~\ref{not:ERLmitZusatz} and~\ref{not:PLmitZusatz} are flexible enough to add multiple events, separated by commas.
	\end{example}

	Now, applying the LTE again to the respective terms $\ECond{\InitAssi{\alpha}}$ in sum~\eqref{eq:FirstApplicationOfLTE} 
	and conditioning on the iteration~$c$ in which a clause of~$\ZV{\Clauses}$ gets selected for the first time yields
	\begin{align}
		\label{eq:SplittingERAlphaWithLTE}
		\ECond{\InitAssi{\alpha}} 
		= \sum_{c\in \N_0} &\PACond{\NoSel{c}}{} \cdot \EACond{\NoSel{c}}\\
		\label{eq:SplittingERAlphaWithLTEsecondline}
		+\:
		&\PACond{\NeverSel}{} \cdot \EACond{\NeverSel}.
	\end{align}

	For clarity of exposition, we analyze each line of this sum in a different case.
	We start with line~\eqref{eq:SplittingERAlphaWithLTEsecondline}, \ie the case in which \SRWA{} never selects a clause from~$\ZV{\Clauses}$ (called the \introduceterm{infinite case}~$\InfCase$). Its treatment can be found in \refSec{sec:TheInfiniteCase}.
	The case of line~\eqref{eq:SplittingERAlphaWithLTE}, the \introduceterm{finite case}~$\FiniteCase$, uses a similar but more involved argument. For this reason, we only mention the result of the analysis in \refSec{sec:TheFiniteCase} and defer the analysis to \refApp{app:FiniteCaseHandled}. We then proceed to present the combined result in \refSec{sec:CombiningBothCases}, showing which random variables (later called $\ZV{P}$, $\ZV{Q}$, and $\ZV{R}$) make up the above expression. These random variables will be in such an elementary form that it is easy for us to analyze their distribution.

	\subsubsection{The Infinite Case $\NeverSel$}%
	\label{sec:TheInfiniteCase}

	As announced, this section will treat the analysis of the term
	\[
	\InfCase
	\coloneqq \PACond{\NeverSel}{} \cdot \EACond{\NeverSel},
	\]
	\ie the term in \refEq{eq:SplittingERAlphaWithLTEsecondline}.
	We will state the goal of this section in an informal form before we begin with the detailed analysis (a detailed version can be found in Proposition~\ref{recap:InfiniteCaseDone}).

	\begin{recapitulation}[The infinite case $\NeverSel$, informal]
		\label{prop:InfiniteCaseInformal}
		It holds that
		\begin{align*}
			\InfCase
			=
			\PACond{\NeverSel}{} \cdot \EACond{\NeverSel}
			=
			\sum_{i \in \N_0} 
			\left\{
			C_1 %
			\cdot
			\prod_{k=1}^{i}
			\left( \sum_{ \falsassi{\gamma}{F} }	
			C_2 %
			\cdot \ZV{Q} \right)
			\cdot i
			\right\}, %
		\end{align*}
		with constants $C_1, C_2$ not depending on $\ZV{\Clauses}$,
		and the random variable $\ZV{Q} \coloneqq \ZV{Q}(\alpha,k,\gamma)$ that, roughly speaking, tells us how likely it is that we select a clause of the set $\Clauses$ in the $k$-th iteration of \SRWA{}, given the knowledge of the previous random walk path (from $\alpha$ to $\gamma$) over the last $k-1$ iterations. The random variable~$\ZV{Q}$ is ``elementary enough'' such that we can analyze its distribution in \refSec{sec:analysis_Q}.
	\end{recapitulation}

	To make the proof of Proposition~\ref{prop:InfiniteCaseInformal} more digestible, we have split it into several steps, each containing one or multiple lemmas of what will be achieved in this step.

	\paragraph{Step 1: Reduction From Probability and Expectation to Probabilities Only}
	In the first step, we rewrite~$\InfCase$
	in a form that only contains probabilities. 
	These
	probabilities will then, in turn, be analyzed in a later step.
	For the formulation of the lemma in this step, we would like to remind the reader that $\ZV{F^\prime}$ refers to the modified formula, \ie $\ZV{F^\prime} \coloneqq F \cup \ZV{\Clauses}$.

	Additionally, we want to emphasize that the definitions of $\ERL(\cdot)$ and $\ZV{\ProbopZV}_{\Clauses}(\cdot)$ in Notation~\ref{not:ERLmitZusatz} and~\ref{not:PLmitZusatz} are flexible enough to add multiple events, separated by commas.
	For example, $\PACond{\W(\ZV{F^\prime})=i}{\NeverSel}$ denotes the probability that \SRWA{} takes $i$~flips to solve~$\ZV{F^\prime}$ under the condition that it was started with~$\alpha$ and never selects a clause from~$\ZV{\Clauses}$.
	Similarly, $\EACond{\NeverSel, \W(\ZV{F^\prime})=i}$ denotes the expected runtime of \SRWA{} subject to the conditions listed in brackets, \ie under the assumptions that no clause of~$\ZV{\Clauses}$ gets selected, $\ZV{F^\prime}$ is solved in $i$~flips, and the random walk started in assignment~$\alpha$.

	\begin{lemma}%
		\label{lem:StepOneLemma}
		It holds 
		\[
		\InfCase
		= \sum_{i \in \N_0} \PACond{\NeverSel}{} \cdot \PACond{\W(\ZV{F^\prime})=i}{\NeverSel} %
		\cdot
		i.
		\]
	\end{lemma}
	
	\begin{proof}
		By applying the LTE to the factor $\EACond{\NeverSel}$ when conditioning on the event $\W(\ZV{F^\prime}) = i$, one obtains
		\begin{align*}
			\InfCase 
			&\stackrel{\phantom{\text{LTE}}}{=} \PACond{\NeverSel}{} \cdot \EACond{\NeverSel} \\
			&\stackrel{\text{LTE}}{=} \sum_{i \in \N_0}
			\Big[
			\PACond{\NeverSel}{} \cdot \PACondResize{\W(\ZV{F^\prime})=i}{\NeverSel} \\
			&\phantom{\stackrel{\text{LTE}}{=} \sum_{i \in \N_0}
				\Big[}
			\cdot
			\EACond{\NeverSel, \W(\ZV{F^\prime})=i}
			\Big].
		\end{align*}
		The last factor of the above sum can be expressed in a simpler form:
		\[
		\EACond{\NeverSel, \W(\ZV{F^\prime})=i} = i.
		\]
		This equation holds because we have already conditioned on the event $\singleevent{\W(\ZV{F^\prime})=i}$, \ie the runtime on~$\ZV{F^\prime}$ being~$i$, and the condition that the algorithm never selects a clause from the random set~$\ZV{\Clauses}$. Hence, the lemma follows.
	\end{proof}

	\paragraph{Step 2: Analysis of the Remaining Two Probabilities With the Help of a Selector Variable and the Chain Rule}
	As our next step, we will analyze the product of the first two factors in Lemma~\ref{lem:StepOneLemma}, \ie the expression 
	\[
	\Expression \coloneqq \PACond{\NeverSel}{} \cdot \PACond{\W(\ZV{F^\prime})=i}{\NeverSel}.
	\]
	For this, we need the following definition telling us if in the $c$-th iteration of~\SRWA{} a clause of $\Clauses$ gets selected or not.
	
	\begin{definition}
		\label{def:ScIndVar}
		Let $\Sel{c}$ be the indicator variable defined as follows: 
		\[
		\Sel{c} \coloneqq \twopartdef{1}{\text{a clause in } $\Clauses$ \text{ gets selected in the } $c$\text{-th iteration of \SRWA{}}}{0}{\text{otherwise}.}
		\]
	\end{definition}
	
	With this definition in place, we can present the lemma for this step.
	
	\begin{lemma}%
		We have
		\begin{align}
			\notag
			\Expression 
			&= %
			\PACond{\NeverSel}{} \cdot \PACondResize{\W(\ZV{F^\prime})=i}{\NeverSel} \\
			\label{eq:ResultOfExpressionCalElone}
			&=
			\PACondResize{\W(\ZV{F^\prime}) = i}{ \bigcap_{j=1}^{i} \singleevent{ \Sel{j} = 0 } } \\
			\label{eq:ResultOfExpressionCalEltwo}
			&\phantom{=} \: \cdot \prod_{k=1}^{i} \PACondResize{\Sel{k} = 0}{ \bigcap_{j=1}^{k-1} \singleevent{ \Sel{j} = 0 } }.
		\end{align}
	\end{lemma}

	\begin{proof}

		By the definition of the conditional probability and by reducing the resulting fraction, we obtain
		\begin{align}
			\notag
			\Expression 
			&= \PACond{\NeverSel}{} \cdot \PACondResize{\W(\ZV{F^\prime})=i}{\NeverSel} \\
			\notag
			&= \frac{\Pbig{\NeverSel, \InitAssi{\alpha}}}{\Pbig{\InitAssi{\alpha}}} \cdot \frac{\Pbig{\W(\ZV{F^\prime})=i, \NeverSel, \InitAssi{\alpha}}}{\Pbig{\NeverSel, \InitAssi{\alpha}}} \\
			\label{eq:MussWeiterEinfacherWerden}
			&= \PACond{\W(\ZV{F^\prime})=i, \NeverSel, \InitAssi{\alpha}}{}.
		\end{align}
		Now notice that
		\[
		\probcond{A}{B} 
		= \frac{\prob{A \cap B}}{\prob{B}}
		= \frac{\prob{A \cap B \cap B}}{\prob{B}}
		= \probcond{A \cap B}{B},
		\]
		hence we can simplify line~\eqref{eq:MussWeiterEinfacherWerden} even further
		\[
		\Expression = \PACond{\W(\ZV{F^\prime})=i, \NeverSel, \InitAssi{\alpha}}{} = \PACond{\W(\ZV{F^\prime})=i, \NeverSel}{}.
		\]

		We continue with analyzing $\PACond{\W(\ZV{F^\prime})=i, \NeverSel}{}$.
		Since we have the information that $\ZV{F^\prime}$ is being solved with $i$~flips, we can express $\NeverSel$ more precisely as the event that in none of the iterations $1, \dots, i$ a clause from $\ZV{\Clauses}$ gets selected, \ie as the intersection $\bigcap_{j=1}^{i} \singleevent{\Sel{j} = 0}$. Thus,
		\[
		\Expression = \PACondResize{ \singleevent{\W(\ZV{F^\prime}) = i} \cap \bigcap_{j=1}^{i} \singleevent{\Sel{j} = 0} }{}.
		\]

		Applying the chain rule for probabilities (cf.\ \refTheo{thm:ChainRule}) we obtain
		\[
		\Expression
		=
		\PACondResize{\W(\ZV{F^\prime}) = i}{ \bigcap_{j=1}^{i} \singleevent{ \Sel{j} = 0 } }
		\cdot \prod_{k=1}^{i} \PACondResize{\Sel{k} = 0}{ \bigcap_{j=1}^{k-1} \singleevent{ \Sel{j} = 0 } },
		\]
		which is what we wanted.
	\end{proof}

	\paragraph{Step 3: Analyzing the Product Rule Factors}
	
	Having achieved this, we proceed to analyze the factors in lines~\eqref{eq:ResultOfExpressionCalElone} and~\eqref{eq:ResultOfExpressionCalEltwo}.
	Let us begin with the first factor in line~\eqref{eq:ResultOfExpressionCalElone}.
	
	\begin{lemma}
		\label{lem:LemmaWithSimpleEquation}
		The following expression is not a random variable anymore:
		\begin{align}
			\label{eq:simple_equation}
			\PACondResize{\W(\ZV{F^\prime}) = i}{ \bigcap_{j=1}^{i} 	\singleevent{ \Sel{j} = 0 } }
			\eqqcolon C_1.	
		\end{align}
	\end{lemma}

	\begin{proof}
		When \mbox{$\bigcap_{j=1}^{i} \singleevent{ \Sel{j} = 0 }$} holds and we have the information that $\W(\ZV{F^\prime}) = i$, Schöning's algorithm has never selected a clause from~$\ZV{\Clauses}$ in the $i$ iterations it required to solve formula~$\ZV{F^\prime}$. Thus, the algorithm performs its random walk only on clauses of the original formula~$F$; hence, the random set~$\Clauses$ does not have any influence. We can therefore write
		\begin{align*}
			\PACondResize{\W(\ZV{F^\prime}) = i}{ \bigcap_{j=1}^{i} \singleevent{ \Sel{j} = 0 } }
			= \Probcond{ \W(F) = i }{ \InitAssi{\alpha} } = C_1.	%
		\end{align*}
		Notice that $\Probcond{ \W(F) = i }{ \InitAssi{\alpha} }$ is not a random variable.
	\end{proof}

	Because of the simplification provided in \refEq{eq:simple_equation}, we can concentrate on the factors in the big product of line~\eqref{eq:ResultOfExpressionCalEltwo}, \ie the part
	$
	\prod_{k=1}^{i} \PACondResize{\Sel{k} = 0}{ \bigcap_{j=1}^{k-1} \singleevent{ \Sel{j} = 0 } }.
	$
	To proceed, however, we require additional notation.
	
	\begin{definition}
		\label{def:FromAlphaToBetaInC}
		We will denote the event that \algoformat{SRWA} (started from the initial assignment $\alpha$) ends up in assignment~$\beta$ after performing $c$~flips with $\Wa{c}{\beta}$.
	\end{definition}
	
	Let us look at a few easy examples to get an intuition for this notation.
	
	\begin{example}
		\label{ex:WalkNotation}
		\begin{enumerate}[(i)]
			\item It holds that $\probcond{\Wa{0}{\alpha}}{\InitAssi{\alpha}} = 1$, since $\alpha$ is the initial assignment selected by \SRWA{}.
			\item If $\alpha \coloneqq \set{x = y = 0, z = 1}$ gets chosen in line \ref{alg:srwa:assignment} of \refAlg{algo:schoening} and we also know that the clause $K \coloneqq (x \lor y \lor \overline{z})$ gets selected in line~\ref{alg:srwa:clause}, then setting 
			\[
			\beta_1 = \set{x = 1, y = 0, z = 1}, \:
			\beta_2 = \set{x = 0, y = 1, z = 1}, \:
			\beta_3 = \set{x = y = z = 0},
			\]
			we have $\probcond{\Wa{1}{\beta_i}}{\InitAssi{\alpha}, K \textnormal{ is selected}} = \frac{1}{3}$ for all $i \in \set{1,2,3}$.
			\item \label{ex:SelectFalsifyBeob} If in the $(c+1)$-th iteration of \SRWA{} a clause of $\Clauses$ gets selected, then for all $\beta \in \assi$ that do not falsify a clause in $\Clauses$ we have
			\[
			\probcond{ \Wa{c}{\beta} }{\InitAssi{\alpha}, \Sel{c+1} = 1} = 0.
			\]
		\end{enumerate}
	\end{example}

	Now, we are ready to analyze
	the factors in the big product of line~\eqref{eq:ResultOfExpressionCalEltwo}. 
	In analyzing these factors, we end up with the random variable~$\ZV{Q}$, of which we show in the \refSec{sec: AnalysisOfTheRVs} that it is Johnson SB distributed.

	\begin{lemma}
		\label{lem:LemmaWithCZwoUndQ}
		For $k \in \set{1, \dots, i}$
		it holds
		\begin{align*}
			\PACondResize{\Sel{k} = 0}{ \bigcap_{j=1}^{k-1} \singleevent{ \Sel{j} = 0 } } 
			=	
			\sum_{ \falsassi{\gamma}{F} } C_2 \cdot \ZV{Q},
		\end{align*}
		where
		\begin{align*}
			C_2 \coloneqq C_2(\alpha,k,\gamma) \coloneqq
			\probcondResize{ \Wa{k-1}{\gamma} }{\bigcap_{j=1}^{k-1} \singleevent{ \Sel{j} = 0 }, \InitAssi{\alpha}} \in [0,1] 
		\end{align*}
		is no random variable, and
		\begin{align*}
			\ZV{Q} 
			\coloneqq
			\ZV{Q}(\alpha,k,\gamma)
			\coloneqq
			\PACond{ \Sel{k} = 0 }{ \Wa{k-1}{\gamma} }.
		\end{align*}
	\end{lemma}

	\begin{proof}
		Let $k \in \set{1, \dots, i}$. One can notice 
		by applying the LTP when conditioning on the event $\Wa{k-1}{\gamma}$, that
		\begin{align*}
			&\PACondResize{\Sel{k} = 0}{ \bigcap_{j=1}^{k-1} \singleevent{ \Sel{j} = 0 } } \\ %
			=& \sum_{\falsassi{\gamma}{F}}
			\left\{
			\PACondResize{\Sel{k} = 0}{\Wa{k-1}{\gamma}, \bigcap_{j=1}^{k-1} \singleevent{\Sel{j}=0}} \right. \\
			&\phantom{=\sum_{\falsassi{\gamma}{F}}}
			\cdot \left. \probcondResize{ \Wa{k-1}{\gamma} }{ \bigcap_{j=1}^{k-1} \singleevent{\Sel{j} = 0}, \InitAssi{\alpha}} \right\}. %
		\end{align*}
		Notice that we sum only over those assignments that falsify a clause in $F$ since the algorithm would already have finished in case of a satisfying assignment. Two additional remarks are due:
		
		First, one can observe that $\PACondResize{\Sel{k} = 0}{\Wa{k-1}{\gamma}, \bigcap_{j=1}^{k-1} \singleevent{\Sel{j}=0}}$ can be rewritten in the form $\PACond{ \Sel{k} = 0 }{ \Wa{k-1}{\gamma} }$ since the condition $\Wa{k-1}{\gamma}$ (\ie the algorithm being in $\gamma$ after $k-1$ flips) makes the condition $\bigcap_{j=1}^{k-1} \singleevent{\Sel{j}=0}$ (\ie the information which clauses were selected along the way) obsolete for determining the expected number of flips the algorithm performs. 
		
		Secondly,
		notice 
		that $\probcondResize{ \Wa{k-1}{\gamma} }{ \bigcap_{j=1}^{k-1} \singleevent{\Sel{j} = 0}, \InitAssi{\alpha}}$ is a probability that does not depend on $\ZV{\Clauses}$; in other words, this term is no random variable (also recall Convention~\ref{convention:probs}). This is true because of the condition $\bigcap_{j=1}^{k-1} \singleevent{\Sel{j} = 0}$, the complete $(k-1)$ flip random walk from $\alpha$ to $\gamma$ is made without considering clauses from~$\ZV{\Clauses}$.
	\end{proof}
	
	Putting together Lemmata~\ref{lem:LemmaWithSimpleEquation} and \ref{lem:LemmaWithCZwoUndQ} yields:
	
	\begin{corollary}
		\label{cor:ExpressionEAnaDone}
		We have 
		\[
		\Expression = C_1 \cdot \prod_{k=1}^{i} \left( \sum_{ \falsassi{\gamma}{F} } C_2 \cdot \ZV{Q} \right).
		\]
	\end{corollary}

	\paragraph{Final Step: Putting Everything Together}
	Having done the steps above, we summarize the results obtained in Subsection~\ref{sec:TheInfiniteCase}, dealing with the infinite case, in the following recapitulation.

	\begin{recapitulation}[The infinite case $\NeverSel$]
		\label{recap:InfiniteCaseDone}
		It holds that
		\begin{align*}
			\InfCase
			=
			\PACond{\NeverSel}{} \cdot \EACond{\NeverSel}
			=
			\sum_{i \in \N_0} 
			\left\{
			C_1 %
			\cdot
			\prod_{k=1}^{i}
			\left( \sum_{ \falsassi{\gamma}{F} }	
			C_2 %
			\cdot \ZV{Q} \right)
			\cdot i
			\right\}, %
		\end{align*}
		with the constants
		\begin{align*}
			C_1 &\coloneqq C_1(\alpha, i) \coloneqq
			\Probcond{ \W(F) = i }{ \InitAssi{\alpha} } \in [0,1], \text{ and} \\
			C_2 &\coloneqq C_2(\alpha,k,\gamma) \coloneqq
			\probcondResize{ \Wa{k-1}{\gamma} }{\bigcap_{j=1}^{k-1} \singleevent{ \Sel{j} = 0 }, \InitAssi{\alpha}} \in [0,1]
		\end{align*}
		and the random variable
		\begin{align*}
			\ZV{Q} 
			\coloneqq
			\ZV{Q}(\alpha,k,\gamma)
			\coloneqq
			\PACond{ \Sel{k} = 0 }{ \Wa{k-1}{\gamma} }.
		\end{align*}
	\end{recapitulation}
	
	\begin{proof}
		We want to emphasize that we will, from now on, drop the dependencies of random variables like $\ZV{Q}$ (from $\alpha,k,\gamma$, etc.) in the notation in order not to overload notation. The reader should, however, keep this in mind.
		We have
		\begin{align*}
			\InfCase
			&\stackrel[\phantom{\text{Cor}~\ref{cor:ExpressionEAnaDone}}]{\phantom{\text{Cor}~\ref{cor:ExpressionEAnaDone}}}{=} 
			\PACond{\NeverSel}{} \cdot \EACond{\NeverSel} \\	
			&\stackrel[\phantom{\text{Cor}~\ref{cor:ExpressionEAnaDone}}]{\text{Lem}~\ref{lem:StepOneLemma}}{=}
			\sum_{i \in \N_0} \Expression \cdot	i \\
			&\stackrel{\text{Cor}~\ref{cor:ExpressionEAnaDone}}{=} 
			\sum_{i \in \N_0} \left\{ C_1 \cdot	\prod_{k=1}^{i}	\left( \sum_{ \falsassi{\gamma}{F} } C_2 \cdot \ZV{Q} \right) \cdot i \right\}. \qedhere
		\end{align*}	
	\end{proof}

	\subsubsection{The Finite Case $c < \infty$}
	\label{sec:TheFiniteCase}

	As seen in the previous section, we sometimes end up with terms that are not random variables (above $C_1$ and $C_2$). Since we will almost exclusively care about random variables in \refSec{sec:CombiningBothCases}, we introduce the following notation that will allow us to drop non-random variables. This has the advantage of keeping the calculations cleaner and more readable.
	This is theoretically founded in Lemma~\ref{lem:MultiplicationOfSB}.
	
	\begin{notation}
		We will use the symbol $\cong$ to indicate equality up to constants, \eg $\frac{1}{2} \ZV{X} + \frac{1}{4} \cong \ZV{X}$.
	\end{notation}
	
	Furthermore, we need the following definition.
	
	\begin{definition}
		\label{def:FlipAlphacbetax}
		We let $\Flipa{c}{\beta}{x}$ denote that in the execution of \SRWA{} (started with assignment~$\alpha$), at the beginning of the $c$-th iteration, the current assignment is $\beta$. Furthermore, the next flip will flip variable~$x$ (\ie at the end of the algorithm's $c$-th iteration in line~\ref{alg:srwa:flip}, we have~$\beta[x]$ as the current assignment).
	\end{definition}
	
	By applying similar arguments as in \refSec{sec:TheInfiniteCase}, where we analyzed $\InfCase$, \ie applying the LTE, the chain rule for probabilities, and using the LTP, one can obtain the following recursion.
	
	\begin{recapitulation}[The finite case $c < \infty$]
		\label{recap:FiniteCaseDone}
		Let $\ZV{V} \coloneqq \Vars{\UNSATunder[\Clauses]{\beta}}$, where $\UNSATunder[\Clauses]{\beta}$ is the set of clauses in~$\Clauses$ that are falsified by $\beta$.
		It then holds that
		\begin{align*}
			\FiniteCase
			&= \sum_{c\in \N_0} \PACond{\NoSel{c}}{} \cdot \EACond{\NoSel{c}} \\
			&\cong \sum_{c\in \N_0} \sum_{ \falsassi{\beta}{\Clauses} }
			\left\{
			\left( \sum_{ \falsassi{\gamma}{F} }
			\ZV{R} \right)
			\cdot
			\prod_{k=1}^{c}
			\left( \sum_{ \falsassi{\gamma}{F} }
			\ZV{Q} \right)
			\right.			
			\left.
			\cdot \sum_{x \in \ZV{V}}  
			\Big(          
			\ZV{P}
			\cdot 
			\ECond{ \InitAssi{ \FlipVarInAssi{x}{\beta} } }
			\Big)
			\right\}
		\end{align*}
		with the random variables
		\begin{align*}
			\ZV{P} &\coloneqq \PACond{ \Flipa{c}{\beta}{x} }{ \Wa{c}{\beta}, \NoSel{c} }, \\	
			\ZV{Q} &\coloneqq \PACond{ \Sel{k} = 0 }{ \Wa{k-1}{\gamma} },\\
			\ZV{R} &\coloneqq \PACond{ \Sel{c+1} = 1 }{ \Wa{c}{\gamma} }.
		\end{align*}
	\end{recapitulation}
	
	\begin{proof}    
		For a detailed presentation of the arguments involved, we refer to \refApp{app:FiniteCaseHandled}.
	\end{proof}

	\subsubsection{Combining Both Cases}
	\label{sec:CombiningBothCases}

	Putting together the cases treated in Sections~\ref{sec:TheInfiniteCase} and~\ref{sec:TheFiniteCase}, we obtain the following:

	\begin{recapitulationthm}
		\label{recap:expectation_distribution}
		Let $\ZV{V}$ be as in Proposition~\ref{recap:FiniteCaseDone}.
		The distribution of the expected value~$\ECond{\InitAssi{\alpha}}$ is given by
		\begin{align*}
			\ECond{\InitAssi{\alpha}} 
			\cong
			&\sum_{c\in \N_0} \sum_{ \falsassi{\beta}{\Clauses} }
			\left\{
			\left( \sum_{ \falsassi{\gamma}{F} }
			\ZV{R} \right)
			\cdot
			\prod_{k=1}^{c}
			\left( \sum_{ \falsassi{\gamma}{F} }
			\ZV{Q} \right)
			\right.			
			\left.
			\cdot \sum_{x \in \ZV{V}}            
			\Big( \ZV{P}
			\cdot \ECond{ \InitAssi{ \FlipVarInAssi{x}{\beta} } } \Big) 
			\right\} \\
			&+  \sum_{i \in \N_0} \prod_{k=1}^{i} \sum_{\falsassi{\gamma}{F}} \ZV{Q} ,
		\end{align*}
		with $\ZV{P}$, $\ZV{Q}$, and $\ZV{R}$ again as in Proposition~\ref{recap:FiniteCaseDone}.
	\end{recapitulationthm}

	We proceed in the next section to analyze the distribution of the random variables~$\ZV{P}$, $\ZV{Q}$, and $\ZV{R}$.
	Finally, in \refSec{sec:ArgumentsLogNDistro}, we present arguments, why~$\ECond{\InitAssi{\alpha}}$ follows one distribution type.

	\subsection{Analysis of the Random Variables}
	\label{sec: AnalysisOfTheRVs}

	Recall that in the setting of this section, a base instance~$F$ is modified by adding a set of logically equivalent clauses~$\Clauses$. The precise manner in which the set of clauses~$\Clauses$ is randomly generated does not affect any derivations and results presented so far in this section. In other words, no matter how the clauses from~$\Clauses$ are sampled, the expression from \refRecapThm{recap:expectation_distribution} will always describe the runtime distribution.
	
	In this section, we proceed by fixing the generating model for the random set~$\ZV{\Clauses}$. Then, we analyze the random variables~$\ZV{P}$, $\ZV{Q}$, and $\ZV{R}$ in more detail based on this assumption. In this context, we prove that each of them (asymptotically) follows the Johnson SB distribution. We begin by outlining the model used to modify formula~$F$. 
	
	\paragraph{Fixing the Generating Model}
	
	The generating model for the set~$\ZV{\Clauses}$ is specified in \refAlg{algo:generating_model}.
	Note that there are pronounced similarities to the model employed in \refsec{sec:DesignOfAlfa}. However, there are two notable differences. First, the model from \refAlg{algo:generating_model} does not necessarily use clauses generated by resolution. Instead, other methods (for example, CDCL, like in the \GapSAT{} solver) can also be used to produce logically equivalent clauses. In this sense, the model from \refAlg{algo:generating_model} generalizes the model discussed in \refsec{sec:DesignOfAlfa}. On the other hand, in \refAlg{algo:generating_model}, we restrict ourselves to clauses with a certain fixed length~$\ell$. In that sense, the chosen model is weaker than the one in \refsec{sec:DesignOfAlfa}. The reason for this restriction is that otherwise, another random variable, namely the length of the added clause, would have to be considered, which would only further complicate the analysis.

	\begin{algorithm}[htb]
		
		\SetKw{KwWithProb}{with probability}
		\SetKw{KwDo}{do}	
		\SetKwFunction{Shuffle}{Shuffle}
		
		\textbf{Input:} Boolean formula $F$, probability $p \in (0,1]$, set $\mathbb{L}$ containing logically equivalent clauses \wrt $F$
		such that $\exists \ell \in \N$: $\forall K \in \mathbb{L}$: $|K| = \ell$\\
		\BlankLine
		
		$\ZV{L} \coloneqq \emptyset$\\ 
		\ForEach{$K \in \mathbb{L}$}{
			\KwWithProb $p$ %
			\KwDo
			$\ZV{L} \coloneqq \ZV{L} \cup \set{K}$
		}
		
		\caption{Generation of the random set $\Clauses$.}
		\label{algo:generating_model}
	\end{algorithm}

	\paragraph{Some Preliminary Comments for the Analysis of the Random Variables~$\ZV{P}$, $\ZV{Q}$, and $\ZV{R}$}

	As indicated above, the Johnson SB distribution is applied several times in the following.
	For our purposes, the exact values of the distribution's parameters are irrelevant for the most part. Therefore, we often omit the brackets containing the exact values of the parameters. 
	In addition, we often deal with asymptotic distributions in this section. A sequence of random variables \mbox{$\ZV{X_1}, \ZV{X_2}, \dots$} with associated cdfs \mbox{$F_1, F_2, \dots$} has an asymptotic distribution~$F$ if \mbox{$\lim_{n\rightarrow \infty} F_n(x) = F(x)$} for all~$x$. A prominent example is that a (suitably scaled) random variable \mbox{$X \distr \Bin{n}{p}$} is asymptotically normally distributed for \mbox{$n\rightarrow \infty$}.
	
	Finally, in the following, it is necessary to argue frequently about the number of unsatisfied clauses given a specific assignment. For this purpose, we introduce a suitable notation.
	
	\begin{definition}
		\label{def:UNSAT}
		Let $G$ be a CNF formula, $\beta$ an assignment of $G$ and $x \in \vars{G}$. We set
		\begin{alignat*}{3}
			&\UNSATunder[G]{\beta} &~~\coloneqq~~&\Setdescr{D \in G}{ \applyassi{D}{\beta} = 0}, \\
			&\UNSATin[G]{\beta}{x} &~~\coloneqq~~ &\Setdescr{D \in G}{ \applyassi{D}{\beta} = 0 \text{ and } x \in \vars{D}}, \\
			&\UNSATnotin[G]{\beta}{x} &~~\coloneqq~~ &\Setdescr{D \in G}{ \applyassi{D}{\beta} = 0 \text{ and } x \notin \vars{D}}.
		\end{alignat*}
		Using the symbol $\ZVcardhash$ before any of these sets denotes the cardinality of the set.
	\end{definition}

	\subsubsection{Distribution Analysis of the Random Variable $\boldsymbol{P}$}
	\label{sec:TheFirstRV}

	We begin our distribution analysis with the random variable~$\ZV{P}$. The goal of this subsection is to prove that $\ZV{P}$ is asymptotically Johnson SB distributed (Proposition~\ref{prop:PIsAsympSB}).

	The rough idea of the proof is to rewrite~$\ZV{P}$ (up to a constant) as the fraction~$\frac{1}{1+\ZV{X}}$, see \refEq{eq:DarstellungVonPAlsBruch}, where $\ZV{X}$ is asymptotically lognormal distributed. Thus, in order to analyze the random variable~$\ZV{P}$, it is helpful to establish a relationship between the lognormal and the Johnson SB distribution.
	
	For that, one technical detail should be clarified beforehand. In the previous sections, we mainly considered the three-parameter lognormal distribution, \ie the lognormal distribution, which is shifted by an additional location parameter. However, for most of this section, we use the two-parameter lognormal distribution, which can be interpreted as a three-parameter lognormal distribution with location parameter zero.

	\begin{lemma}[\cite{MOSBanswer}]
		\label{lem:1/1+Y_for_Y_LogN_ist_SB_verteilt}
		Let $\ZV{X} \distr \LogN{\mu}{\sigma^2}$, then for all $c \in \Rplus$ we have
		\[
		\frac{1}{c+\ZV{X}} \distr \SBwithequations{\gamma = \frac{\mu-\ln c}{\sigma}}{\delta = \frac{1}{\sigma}}{\lambda = \frac{1}{c}}{\xi = 0}.
		\] 
	\end{lemma}

\begin{proof}
	For a proof, refer to Appendix~\ref{sec:FirstAppProof}.
\end{proof}

	Furthermore, a connection between binomial and lognormal distributions can also be shown. We need this connection to show that $\ZV{X}$ is asymptotically lognormal distributed. For this purpose, the following lemma is necessary.

	\begin{lemma}[\cite{Lachin2011Biostatistical}]
		\label{lem:log_binomial_normal}
		If $\ZV{Y} \distr \Bin{n}{p}$, then $\log \ZV{Y}$ is asymptotically normally distributed.
	\end{lemma}

	\begin{proof}
		The proof is given in Appendix~\ref{sec:SecondAppendixProof}.
	\end{proof}

	Together with \refDef{def:threeparameterLogN} this lemma immediately yields the following corollary.
	
	\begin{corollary}
		\label{cor:binomial_lognormal}
		Let $\ZV{Y}$, $\ZV{Z}$ be two independent binomial distributions. Then, $\ZV{Y}$, $\ZV{Z}$, as well as $\ZV{Y}/\ZV{Z}$, are asymptotically lognormal distributed.
	\end{corollary}
	
	\begin{proof}
		By applying \refDef{def:threeparameterLogN} of lognormal distributions,
		we see that a random variable is lognormal distributed if and only if its logarithm is normally distributed. According to \refLem{lem:log_binomial_normal}, both $\log \ZV{Y}$ and $\log \ZV{Z}$ are asymptotically normally distributed. Therefore, it follows that $\ZV{Y}$ and $\ZV{Z}$ are both asymptotically lognormal distributed.
		
		It holds that
		\[
		\log ( \ZV{Y} / \ZV{Z} ) = \log \ZV{Y} - \log \ZV{Z}.   
		\]
		By \refLem{lem:log_binomial_normal}, we already know that $\log \ZV{Y}$ and $\log \ZV{Z}$ are each asymptotically normal distributed. An extremely valuable property is that the difference of independent normal distributions is again normally distributed. Therefore, we note that $\log \ZV{Y} - \log \ZV{Z}$ is asymptotically normal distributed and, as a result, $\ZV{Y}/\ZV{Z}$ is asymptotically lognormal distributed.
	\end{proof}
	
	\begin{remark}
		Note that this corollary does not contradict the famous theorem by de Moivre and Laplace, which states that $\ZV{Y}$ is approximately normally distributed. Roughly speaking, this is because both approximations converge towards the same limiting distribution.
		For further elaboration, refer to~\cite{cheng2017non}.
	\end{remark}
	
	We are now ready to analyze the asymptotic distribution of the random variable~$\ZV{P}$.

	\begin{proposition}
		\label{prop:PIsAsympSB}
		The random variable
		\[
		\ZV{P} \coloneqq \PACond{\Flipa{c}{\beta}{x} }{ \Wa{c}{\beta}, \NoSel{c} }
		\]
		is asymptotically Johnson $\SBop$ distributed.
	\end{proposition}

	\begin{proof}
		As a first step, it is worthwhile to clarify what exactly is being examined. Therefore, we will initially focus on the conditions in \mbox{$\PACond{\Flipa{c}{\beta}{x}}{ \Wa{c}{\beta}, \NoSel{c}}$}. 
		The conditions~$A(\alpha)$ and $\Wa{c}{\beta}$ tell us that \SRWA{} was initialized in $\alpha$ and is in $\beta$ after exactly $c$~flips. Due to the additional condition $\NoSel{c}$, we know that a clause from $\Clauses$ is chosen for the $(c+1)$-st flip. In other words, a clause from the newly added clauses is picked.

		Subject to these conditions, we are interested in how likely the algorithm will flip~$x$ next. This first requires selecting a clause containing~$x$ and then, in the next step, selecting~$x$ as the variable to be flipped (cf. \refAlg{algo:schoening}). Since these two events are independent, the overall probability can be expressed as the product of the two individual probabilities.

		The probability of selecting a clause containing $x$ is the ratio of the unsatisfied clauses from $\ZV{\Clauses}$ containing $x$ to all clauses from $\Clauses$, as one already knows that a clause from $\ZV{\Clauses}$ is chosen. If a clause $K$ containing $x$ has been selected, then $x$ will be flipped with probability $1/\vert K \vert$. In \refAlg{algo:generating_model}, we have restricted ourselves to clauses of equal length~$\ell$, so $1/\vert K\vert$ is constant and is therefore not a random variable. With these preliminary considerations, \mbox{$\ZV{P}=\PCond{ \Flipa{c}{\beta}{x} }{ \Wa{c}{\beta}, \NoSel{c} }$} can now be given more precisely:
		\begin{align}
			\label{eq:DarstellungVonPAlsBruch}
			\begin{split}
				\ZV{P} &= \frac{1}{|\clstd|}
				\cdot \frac{\ZVcardhash{\ZVUNSATin[\Clauses]{\beta}{x}}}{\ZVcardhash{\ZVUNSATunder[\Clauses]{\beta}}} \\
				&= \frac{1}{\ell}
				\cdot \frac{\ZVcardhash{\ZVUNSATin[\Clauses]{\beta}{x}}}{\ZVcardhash{\ZVUNSATin[\Clauses]{\beta}{x}} + \ZVcardhash{\ZVUNSATnotin[\Clauses]{\beta}{x}}} \\
				&= \frac{1}{\ell} \cdot
				\dfrac{1}{1 + \dfrac{\ZVcardhash{\ZVUNSATnotin[\Clauses]{\beta}{x}}}{\ZVcardhash{\ZVUNSATin[\Clauses]{\beta}{x}}}}.
			\end{split}
		\end{align}

		We can now utilize knowledge about the model used to generate~$\ZV{\Clauses}$ (cf.~\refAlg{algo:generating_model}). Since each clause from $\mathbb{L}$ is independently from each other included in $\ZV{\Clauses}$ with probability~$p$, both $\ZVcardhash{\ZVUNSATin[\Clauses]{\beta}{x}}$ and $\ZVcardhash{\ZVUNSATnotin[\Clauses]{\beta}{x}}$ are binomially distributed. Moreover, since no clause can both contain and not contain $x$, the two random variables are independent of each other. As a consequence,
		\[
		\ZV{X} \coloneqq
		\frac{\ZVcardhash{\ZVUNSATnotin[\Clauses]{\beta}{x}}}{\ZVcardhash{\ZVUNSATin[\Clauses]{\beta}{x}}}
		\]
		is asymptotically lognormal distributed according to \refCor{cor:binomial_lognormal}. Then, \refLem{lem:1/1+Y_for_Y_LogN_ist_SB_verteilt} yields the result.
	\end{proof}

	\subsubsection{Distribution Analysis of the Random Variable $\boldsymbol{Q}$}
	\label{sec:analysis_Q}
	
	The analysis of the random variable~$\ZV{Q}$ can be achieved using similar reasoning as before. Again, we show that $\ZV{Q}$ is asymptotically Johnson SB distributed.

	\begin{proposition}
		\label{prop:QisAsyJSBdist}
		The random variable
		\[
		\ZV{Q} \coloneqq \PACondResize{ \Sel{k} = 0 }{ \Wa{k-1}{\gamma}, \bigcap_{j=1}^{k-1} \singleevent{ \Sel{j}=0 } }
		\]
		is asymptotically Johnson SB distributed.
	\end{proposition}

	\begin{proof}
		Let us again focus on the conditions in the above probability expression. We know that on its random walk from assignment~$\alpha$ to~$\gamma$, \SRWA{} has never selected a clause of the set~$\ZV{\Clauses}$ in the $k-1$ flips this random walk took.
		By design, \SRWA{} is independent of its past. Thus, the information that it started from $\alpha$ and that no clause from $\ZV{\Clauses}$ was selected in the last $k-1$~flips does not affect the probability of choosing a clause from $\ZV{\Clauses}$  for the current flip. In contrast, the information that the algorithm is in assignment $\gamma$ is of relevance. Therefore, the probability that no clause from $\ZV{\Clauses}$ is chosen is given by the ratio of the original unsatisfied clauses under~$\gamma$ to all unsatisfied clauses in the extended instance $F \cup \ZV{\Clauses}$. Based on this reasoning, $\ZV{Q}$ can be expressed as follows:
		\begin{align*}
			\ZV{Q} 
			&= \PACond{ \Sel{k} = 0 }{\Wa{k-1}{\gamma}}  \\
			&= \frac{ \cardhash{\UNSATunder[F]{\gamma}} }{ \cardhash{\UNSATunder[F]{\gamma}} + \ZVcardhash{\ZVUNSATunder[\Clauses]{\gamma}} } \\
			&= \frac{ 1 }{ 1 + \frac{\ZVcardhash{\ZVUNSATunder[\Clauses]{\gamma}}}{\cardhash{\UNSATunder[F]{\gamma}}} }.
		\end{align*}

		Due to the model for generating $\ZV{\Clauses}$, we conclude that $\ZVcardhash{\ZVUNSATunder[\Clauses]{\gamma}}$ is binomially distributed. Moreover, $\ZVcardhash{\ZVUNSATunder[\Clauses]{\gamma}}$ is asymptotically lognormal distributed because of \refCor{cor:binomial_lognormal}. Since lognormal distributions are closed under multiplication of constants, $\frac{\ZVcardhash{\ZVUNSATunder[\Clauses]{\gamma}}}{\cardhash{\UNSATunder[F]{\gamma}}}$ is also asymptotically lognormal distributed. According to \refLem{lem:1/1+Y_for_Y_LogN_ist_SB_verteilt} it follows that $\ZV{Q}$ is asymptotically Johnson SB distributed.
	\end{proof}

	\subsubsection{Distribution Analysis of the Random Variable $\boldsymbol{R}$}
	\label{sec:AnalysisOfR}

	For the analysis of the random variable $\ZV{R}$, we use another helpful property of lognormal distributions: They are closed under reciprocity.
	
	\begin{lemma}[\cite{statisticlognormal}]
		\label{lem:InverseOfLogNormal}
		If $\ZV{X} \distr \LogN{\mu}{\sigma^2}$, then $\frac{1}{\ZV{X}} \distr \LogN{-\mu}{\sigma^2}$.
	\end{lemma}

	With this knowledge, one can now turn to the analysis of $\ZV{R}$ itself.
	
	\begin{proposition}
		The random variable
		$\ZV{R} \coloneqq \PACondResize{ \Sel{c+1} = 1 }{ \Wa{c}{\gamma} }$ is asymptotically Johnson SB distributed.
	\end{proposition}
	
	\begin{proof}
		The reasoning is similar to that in \refsec{sec:analysis_Q}. The probability that a clause from $\ZV{\Clauses}$ is selected is given by the ratio of unsatisfied clauses from $\ZV{\Clauses}$ to the total number of unsatisfied clauses:
		\[
		\ZV{R} 
		= \frac{ \ZVcardhash{\ZVUNSATunder[\Clauses]{\gamma}} }{ \cardhash{\UNSATunder[F]{\gamma}} + \ZVcardhash{\ZVUNSATunder[\Clauses]{\gamma}} }
		= \dfrac{1}{ 1 + \dfrac{\cardhash{\UNSATunder[F]{\gamma}}}{\ZVcardhash{\ZVUNSATunder[\Clauses]{\gamma}}} }.		
		\]

		In our model for adding the new clauses, $\ZVcardhash{\ZVUNSATunder[\Clauses]{\gamma}}$ has a binomial distribution. Thus $\ZVcardhash{\ZVUNSATunder[\Clauses]{\gamma}}$ is asymptotically lognormal distributed (cf. \refCor{cor:binomial_lognormal}). By applying \refLem{lem:InverseOfLogNormal}, we see that $\frac{\cardhash{\UNSATunder[F]{\gamma}}}{\ZVcardhash{\ZVUNSATunder[\Clauses]{\gamma}}}$ is also asymptotically lognormal distributed. \refLem{lem:1/1+Y_for_Y_LogN_ist_SB_verteilt} then implies that $\ZV{R}$ is asymptotically Johnson SB distributed.
	\end{proof}

	\subsubsection{Concluding Remarks Regarding the Analysis of the Random Variables}

	In the results of this section, we frequently mentioned that the random variables are asymptotically Johnson SB distributed. Nevertheless, it is crucial to clarify \emph{under which conditions} the random variables converge to a Johnson SB distribution.

	As a starting point, we commonly deal with binomial distributions of the form $\ZVcardhash{\ZVUNSATunder[\Clauses]{\gamma}} \distr \Bin{n_\gamma}{p}$, where $n_\gamma$ is the number of clauses in $\mathbb{L}$ that are unsatisfied under $\gamma$. Our results take effect precisely when $n_\gamma$ approaches infinity.

	As is common practice, these results can also be applied under weaker conditions. For example, if one has ``only'' a large number of unsatisfied clauses in~$\mathbb{L}$, then the random variables are not exactly Johnson SB distributed, but Johnson SB distributions then represent an excellent approximation.

	\subsection{Putting Everything Together}
	\label{sec:ArgumentsLogNDistro}

	In the last two sections, we observed the distribution of the expected runtime of \Alfa{}. Then, in \refsec{sec:analysis_rtd}, we deduced an expression for this distribution consisting of the random variables $\ZV{P}$, $\ZV{Q}$, and $\ZV{R}$.
	Finally, in \refsec{sec: AnalysisOfTheRVs}, we considered 
	these random variables
	in detail. We found that each asymptotically follows a Johnson SB distribution, respectively. 
	In other words, the distribution consists of asymptotically Johnson SB distributed random variables.

	\begin{mainresult}
		The expected runtime~$\ERL$ consists of approximately Johnson SB distributed random variables.
	\end{mainresult}
	
	While this is not enough to prove Strong Conjecture (\refCon{conj:strong}), the distribution of~$\ERL$ can be treated as Johnson SB for all intents and purposes.
	Note that this result is slightly weaker than the Strong Conjecture.
	Therefore, it does not make our conjectures obsolete.

	\section{Conclusion}
	\label{sec:Conclusion}
	
	It has been shown in~\cite{LW20OnTheEffectOfLearnedClauses} that adding new, logically equivalent clauses to a formula improves the performance of SLS SAT solvers on average.
	Building upon this observation, we have shown the following main results in this paper:
	\begin{enumerate}
		\item Treating this process as a random process, the hardness distribution follows a Johnson SB distribution. These distributions can converge to long-tailed distributions. 
		\item We have proven that restarts are useful to avoid long-tails. Thus, the algorithms can be further improved by implementing a restart strategy.
	\end{enumerate}
	
	There are several possible starting points for future research.
	In this work, we studied the runtime distributions of SAT solvers on modified instances.  A major influencing factor for our work was that different problem formulations are often employed in mixed integer programming. Therefore, it would be interesting to see if Johnson SB distributions are also the appropriate descriptive model in that context or if a different family of distributions has to be used. This would also provide more insight into how exactly modified instances affect an algorithm.

	As can be inferred from the theoretical analysis of the runtime behavior, Johnson SB distributions significantly impact the overall runtime. This suggests that for future studies of algorithms, one should also consider using the Johnson SB distribution as a suitable model, alongside established distribution types such as the normal, lognormal, and Weibull distributions. This is especially true if one knows that there is an upper bound for the runtimes since the Johnson SB, in contrast to the other distributions mentioned, is able to model finite support. Such a scenario is given, for example, in the case of a systematic search. As soon as the search tree has been completely traversed, the algorithm terminates. Thus, the runtime of the algorithm can be bounded from above by the size of the search tree.

	Furthermore, it could be a worthwhile pursuit to theoretically investigate the distribution of other SLS solving paradigms, such as configuration checking solvers, \eg \algoformat{CCAnr}~\cite{CCANRCaiLS15}, or \algoformat{WalkSAT}~\cite{DBLP:conf/aaai/SelmanKC94}. For the case of CDCL solvers we have reported results in~\cite{ArbeitmitTom}.

	\section*{Code and Data Availability}
	
	The source code realizing the methodologies described in this paper and the corresponding results are freely available.  
	Therefore, all experiments are fully reproducible.		
	The corresponding code of \refsec{sec:evidence} (visual and statistical evaluations) is available under 
	\begin{center}
		\textsf{\small{\url{https://github.com/FlorianWoerz/Towards-an-Understanding-of-Long-Tailed-Runtimes}}}.
	\end{center} 
	All evaluations take place in the files that can be found under 
	\begin{center}
		\texttt{\small{./evaluation/jupyter\_SB/evaluate\_*.ipynb}}.
	\end{center} 
	A permanent version of this repository has been preserved under 
	\begin{center}
		\textsf{\small{\doi{10.5281/zenodo.6945926}}}. 
	\end{center}
	The CNF data (all base instances, resolvents, and modifications) can be found under 
	\begin{center}
		\textsf{\small{\doi{10.5281/zenodo.4715893}}}. 
	\end{center}
	The source code of \algoformat{concealSATgen} that was used in \refsec{sec:ExpSetupInstTypesSolversUsed} to generate hidden solution instances is published under~\cite{LW21SourceCodeOfConcealSATgen}.

	\section*{Acknowledgments}
	
	The authors acknowledge support by the state of Baden-Württemberg through bwHPC. This research was funded by the Deutsche Forschungsgemeinschaft (DFG, German Research Foundation).	
	
	The authors would like to express their gratitude to Uwe Schöning for several helpful discussions regarding the theoretical analysis.
	
	Furthermore, the authors are greatly indebted to the anonymous referees of the European Symposium on Algorithms (ESA ’21) and the ACM Journal of Experimental Algorithmics (ACM JEA) for several comments that immensely helped to improve the presentation of this paper.

	\section*{Previous Versions of the Paper}
	
	This is the full-length version of the paper in the ACM Journal of Experimental Algorithmics~(JEA). It is an extended and enhanced version of the paper ``Evidence for Long-Tails in SLS Algorithms,''~\cite{WL21EvidenceForLongTails} presented at ESA~2021. \refSec{sec:evidence} describes the results from~\cite{WL21EvidenceForLongTails}, adapted to the Johnson SB distribution which converges to lognormal distributions.
	\refSec{sec:theory} and the appendix is entirely new material providing mathematical justification for the fact that Johnson SB distributions underlie the modification model.

	\DeclareUrlCommand{\Doi}{\urlstyle{sf}}
	\renewcommand{\path}[1]{\footnotesize\Doi{#1}}
	\renewcommand{\url}[1]{\href{#1}{\footnotesize\Doi{#1}}}
	
	\bibliographystyle{alphaurl}
	
	{\small
		\bibliography{Bibliography_no_dois}
	}

	\appendix
	
	\section{The Finite Case} %
	\label{app:FiniteCaseHandled}

	Recall that in Equations~\eqref{eq:SplittingERAlphaWithLTE} and~\eqref{eq:SplittingERAlphaWithLTEsecondline} on page~\pageref{eq:SplittingERAlphaWithLTE} we have seen that
	\begin{align*}
		\ECond{\InitAssi{\alpha}}
		= \FiniteCase + \InfCase 
		= \sum_{c\in \N_0} &\PACond{\NoSel{c}}{} \cdot \EACond{\NoSel{c}}\\
		+ \:
		&\PACond{\NeverSel}{} \cdot \EACond{\NeverSel}.
	\end{align*}
	In \refSec{sec:TheInfiniteCase}, we have analyzed the term $\InfCase = \PACond{\NeverSel}{} \cdot \EACond{\NeverSel}$, \ie the case in which no clause of~$\ZV{\Clauses}$ ever gets selected. We have seen that we can write
	\begin{align*}
		\InfCase
		\cong 
		\sum_{i \in \N_0}  \prod_{k=1}^{i}  \sum_{\falsassi{\gamma}{F}} \ZV{Q},
	\end{align*}
	where~$\ZV{Q}$ is asymptotically Johnson SB distributed (cf.~Proposition~\ref{prop:QisAsyJSBdist}).
	Our aim in this appendix is to obtain a similar representation for 
	\[
	\FiniteCase \coloneqq \sum_{c\in \N_0} \PACond{\NoSel{c}}{} \cdot \EACond{\NoSel{c}}.
	\]
	We then proceed to analyze the distribution of the ensuing random variables in \refSec{sec:CombiningBothCases} of the main body.

	To abbreviate the notation, we use the next definition in the following.
	
	\begin{definition}
		\label{def:VZV}
		We let $\ZV{V} \coloneqq \Vars{\UNSATunder[\Clauses]{\beta}}$, where $\UNSATunder[\Clauses]{\beta}$ is the set of clauses in $\Clauses$ that are falsified by $\beta$.
	\end{definition}
	
	Let us now state the first intermediate result of this section.

	\begin{proposition}
		\label{prop:StatetheProp}
		It holds that
		\begin{align}
			\label{eq:GaloreFirst}
			\FiniteCase
			&= \sum_{c\in \N_0} \sum_{\falsassi{\beta}{\Clauses}} 
			\bigg\{       
			\Probcond{ \Wa{c}{\beta} }{\NoSel{c}, \InitAssi{\alpha}} \cdot
			\PACond{\NoSel{c}}{} \\
			\label{eq:GaloreSecond}
			&\phantom{MMMMMM}
			\cdot \sum_{x \in \ZV{V}}            
			\PACond{ \Flipa{c}{\beta}{x} }{ \Wa{c}{\beta}, \NoSel{c} } \\
			\label{eq:GaloreThird}
			&\phantom{MMMMMMMM}
			\cdot \Big(c + 1 + \ECond{ \InitAssi{\beta^\prime} } \Big) \bigg\}.%
		\end{align}
	\end{proposition}

	\begin{proof}
		We will apply the LTE twice to analyze the finite case (\ie terms in $\ECond{\InitAssi{\alpha}}$ with an eventual selection of a clause from $\Clauses$). 
		In the first application of the LTE 
		we condition over the event~$\Wa{c}{\beta}$ and obtain
		\begin{align*}
			\FiniteCase
			\stackrel{\phantom{\text{LTE}}}{=}&
			\sum_{c\in \N_0} \PACond{\NoSel{c}}{} \cdot \EACond{\NoSel{c}} \\
			\stackrel{\text{LTE}}{=}& \sum_{c \in \N_0} 
			\sum_{ \falsassi{\beta}{\Clauses} } %
			\Big\{
			\Probcond{ \Wa{c}{\beta} }{\NoSel{c}, \InitAssi{\alpha}} \cdot
			\PACond{\NoSel{c}}{} 
			\\
			& 
			\phantom{MMMMM}
			\cdot
			\EACond{\NoSel{c}, \Wa{c}{\beta} } \Big\}.  
		\end{align*}
		Due to Example~\ref{ex:WalkNotation}\,(\ref{ex:SelectFalsifyBeob}), we only need to consider such $\beta \in \assi$ in the sum emerging from the application of the LTE that falsify a clause in~$\Clauses$.
		
		For the second application of the~LTE 
		we condition over the event~$\Flipa{c}{\beta}{x}$.
		Then, we can write
		\begin{align}
			\FiniteCase
			\stackrel{\text{LTE}}{=}& \sum_{c \in \N_0} \sum_{ \falsassi{\beta}{\Clauses} } %
			\Big\{
			\Probcond{ \Wa{c}{\beta} }{\NoSel{c}, \InitAssi{\alpha}} \cdot
			\PACond{\NoSel{c}}{}  \\
			\label{eq:LTEGalore-LTE2-line4}
			& 
			\phantom{MMMMM}
			\cdot \sum_{x \in \ZV{V}}
			\PACond{\Flipa{c}{\beta}{x}}{\Wa{c}{\beta}, \NoSel{c}} \\
			\label{eq:LTEGalore-LTE2-line5}
			&
			\phantom{MMMMM}\phantom{MM}
			\cdot \, \ECond{\NoSel{c}, \Wa{c}{\beta}, \Flip{c}{\beta}{x}, \InitAssi{\alpha}} \Big\}.
		\end{align}

		One can notice the following two facts:
		First, the expression 
		\[
		\Probcond{ \Wa{c}{\beta} }{\NoSel{c}, \InitAssi{\alpha}}
		\]
		is no random variable depending on \ZV{\Clauses} since the event $\NoSel{c}$ ensures that in none of the first $c$ for-loop iterations, a clause from $\Clauses$ gets selected, but only in the $(c+1)$-st iteration.
		Thus, the probability that the algorithm ends up in~$\beta$ after the $c$-th flip does not depend on~$\Clauses$ since the
		algorithm behaves as on the original instance $F$.
		
		Secondly, observe that the term
		appearing in line~\eqref{eq:LTEGalore-LTE2-line5}
		can be expressed recursively as
		\begin{align}
			\label{eq:FallBackToERBetaPrime}
			\EACond{\NoSel{c}, \Wa{c}{\beta}, \Flip{c}{\beta}{x} }
			= c + 1 + \ECond{\InitAssi{\beta^\prime}},
		\end{align}
		where $\beta^\prime \coloneqq \beta[x]$, since $c$ flips were performed to get from the initial assignment~$\alpha$ to assignment~$\beta$, and one $x$-flip was performed to get from~$\beta$ to~$\beta^\prime$. The remaining expected runtime is
		independent of the previous history of the performed random walk and thus only depends on the event~$\InitAssi{\beta^\prime}$.
		Now, replacing~\eqref{eq:LTEGalore-LTE2-line5} with~\eqref{eq:FallBackToERBetaPrime}
		yields the proposition.
	\end{proof}

	\label{page:IntroVonZVP}
	While it turns out in \refSec{sec:TheFirstRV} that we are able to analyze the distribution of the random variable
	\[
	\ZV{P} \coloneqq \PACond{\Flipa{c}{\beta}{x} }{ \Wa{c}{\beta}, \NoSel{c} },
	\]
	appearing in line~\eqref{eq:GaloreSecond}, and show that it is asymptotically Johnson SB distributed, the distribution analysis of the more harmless-looking random variable $\PACond{\NoSel{c}}{}$ appearing in line~\eqref{eq:GaloreFirst} requires further work.
	
	The main idea in the following steps is to split the information that the event $\NoSel{c}$ contains into several events, apply the chain rule and analyze the resulting cases separately. While in the beginning, this might seem like a lot of unnecessary calculations, after being done, we will reap the fruit of our labor: we will be able to conduct a distribution analysis of all random variables easily.
	For this purpose, recall Definition~\ref{def:ScIndVar}:
	Let $\Sel{c}$ be the indicator variable being 1 if and only if a clause in $\ZV{\Clauses}$ gets selected in the $c$-th iteration of \SRWA{}.

	Following the above-specified plan, we next apply the chain rule to the random variable~$\PACond{\NoSel{c}}{}$.
	With the notation of Definition~\ref{def:ScIndVar} in place, and using the fact that
	\[
	\NoSel{c} = \singleevent{ \Sel{c+1} = 1 } \cap \bigcap_{i=1}^{c} \singleevent{ \Sel{i} = 0 },
	\]
	one obtains
	\begin{align}	
		&\PACond{\NoSel{c}}{}  \notag \\
		=& \:\, \PACondResize{ \singleevent{ \Sel{c+1} = 1 } \cap \bigcap_{i=1}^{c} \singleevent{ \Sel{i} = 0 } }{} \notag \\
		\label{eq:ProduktMitFallunterscheidung}
		=& \:\, \DExprOne %
		\cdot \prod_{k=1}^{c} \DExprTwo,
	\end{align}
	where
	\begin{align*}
		\DExprOne
		&\coloneqq
		\PACondResize{\Sel{c+1} = 1}{ \bigcap_{j=1}^{c} \singleevent{ \Sel{j} = 0 } }, \text{ and } \\
		\DExprTwo
		&\coloneqq
		\PACondResize{\Sel{k} = 0}{ \bigcap_{j=1}^{k-1} \singleevent{ \Sel{j} = 0 } }.
	\end{align*}
	Here, the \introduceterm{nullary intersection} $\bigcap_{j=1}^0 \singleevent{\Sel{j} = 0} \coloneqq \samplespace$, \ie as the whole sample space, because the condition of the intersection is a vacuous truth.

	We analyze the random variables of the first factor, \ie $\DExprOne$, and the factors in the big product of line~\eqref{eq:ProduktMitFallunterscheidung}, \ie $\DExprTwo$, separately in the cases~\ref{sec:CaseSc1} and~\ref{sec:CaseSk0} below.
	In \refSec{sec:PuttingTogetherCasesSk}, we finish the case analysis and present the central proposition of this appendix.

	\subsection{The Case \texorpdfstring{$\boldsymbol{\Sel{c+1} = 1}$}{Sel(c+1) = 1} of the First Factor \texorpdfstring{$\DExprOne$}{D\_1} in Line~(\ref{eq:ProduktMitFallunterscheidung})}
	\label{sec:CaseSc1}
	
	With the LTP, conditioning on the event $\Wa{c}{\gamma}$ one obtains
	\begin{align*}
		\DExprOne 
		&=
		\PACondResize{\Sel{c+1} = 1}{ \bigcap_{j=1}^{c} \singleevent{ \Sel{j} = 0 } } \\
		&=
		\sum_{\falsassi{\gamma}{F}}
		\PACondResize{ \Sel{c+1} = 1 }{ \Wa{c}{\gamma}, \bigcap_{j=1}^{c} \singleevent{ \Sel{j}=0 } } \\
		&\phantom{=
			\sum_{\falsassi{\gamma}{F}}} \cdot \probcondResize{ \Wa{c}{\gamma} }{\bigcap_{j=1}^{c} \singleevent{ \Sel{j} = 0 }, \InitAssi{\alpha}}.
	\end{align*}
	Note that once again, $\probcondResize{ \Wa{c}{\gamma} }{\bigcap_{j=1}^{c} \singleevent{ \Sel{j} = 0 }, \InitAssi{\alpha}}$ does not depend on $\ZV{\Clauses}$ since the event \mbox{$\bigcap_{j=1}^{c} \singleevent{ \Sel{j} = 0 }$} rules out the usage of any clauses of $\ZV{\Clauses}$ on the way from $\alpha$ to $\gamma$.
	One readily checks that
	\begin{align}
		\label{eq:FirstIntroQ}
		\begin{split}
			&\PACondResize{ \Sel{c+1} = 1 }{ \Wa{c}{\gamma}, \bigcap_{j=1}^{c} \singleevent{ \Sel{j}=0 } } \\
			= \: &\PACond{ \Sel{c+1} = 1 }{ \Wa{c}{\gamma} } \\
			= \: &\ZV{Q},
		\end{split}
	\end{align}
	since \SRWA{} has the property that past history does not matter, \ie once the algorithm has arrived in assignment~$\gamma$ after $c$~steps, it does not matter how it arrived in this assignment or what clauses it touched on its way.
	
	We will take care of the distribution analysis of the random variable~$\ZV{Q}$ in \refSec{sec:analysis_Q}.

	\subsection{The Case \texorpdfstring{$\boldsymbol{\Sel{k} = 0}$}{Sel(k) = 0} of the Factors in the Big Product of Line~(\ref{eq:ProduktMitFallunterscheidung})}
	\label{sec:CaseSk0}
	
	For $k \in \set{1, \dots, c}$ we analyze the factors in the big product of line~\eqref{eq:ProduktMitFallunterscheidung} with the LTP conditioning on $\Wa{k-1}{\gamma}$:
	\begin{align*}
		\DExprTwo
		&= \PACondResize{\Sel{k} = 0}{ \bigcap_{j=1}^{k-1} \singleevent{ \Sel{j} = 0 } } \\
		&= \sum_{\falsassi{\gamma}{F}}
		\PACondResize{ \Sel{k} = 0 }{ \Wa{k-1}{\gamma}, \bigcap_{j=1}^{k-1} \singleevent{ \Sel{j}=0 } } \\
		&\phantom{= \sum_{\falsassi{\gamma}{F}}} \cdot \probcondResize{ \Wa{k-1}{\gamma} }{\bigcap_{j=1}^{k-1} \singleevent{ \Sel{j} = 0 }, \InitAssi{\alpha}}.
	\end{align*}
	
	An analogous argumentation as in \refSec{sec:CaseSc1} shows that 
	\[
	\probcondResize{ \Wa{k-1}{\gamma} }{\bigcap_{j=1}^{k-1} \singleevent{ \Sel{j} = 0 }, \InitAssi{\alpha}}
	\]
	is no random variable. Similarly, one can use the reasoning brought forward in the last section to see that
	\begin{align*}
		&\PACondResize{ \Sel{k} = 0 }{ \Wa{k-1}{\gamma}, \bigcap_{j=1}^{k-1} \singleevent{ \Sel{j}=0 } } \\
		= \: &\PACond{ \Sel{k} = 0 }{ \Wa{k-1}{\gamma} }.
	\end{align*}
	
	We will use $\ZV{R}$ in the following to denote the above specified random variable. Its distribution analysis can be found in \refSec{sec:AnalysisOfR}.

	\subsection{Putting Together Both Cases}
	\label{sec:PuttingTogetherCasesSk}
	
	Having obtained the results of the last two sections, we immediately arrive at the following.

	\begin{recapitulation}[Analysis of the finite case $c < \infty$]
		\label{recap:Finite}
		The term $\ECond{\InitAssi{\alpha}}$ can be written as
		\begin{align*}
			\sum_{c\in \N_0} \sum_{ \falsassi{\beta}{\Clauses} }
			\left\{
			C_3
			\left( \sum_{ \falsassi{\gamma}{F} }
			C_4 \cdot \ZV{R} \right)
			\cdot
			\prod_{k=1}^{c}
			\left( \sum_{ \falsassi{\gamma}{F} }
			C_4 \cdot \ZV{Q} \right)
			\right.			
			\left.
			\cdot \sum_{x \in \ZV{V}}            
			\ZV{P} %
			\cdot \left(c + 1 + \ECond{ \InitAssi{\beta^\prime} } \right) 
			\right\} 
			+ \InfCase,
		\end{align*}
		with the constants being defined by
		\begin{align*}
			C_3 \coloneqq C_3(\alpha,c,\beta)
			&\coloneqq 
			\Probcond{ \Wa{c}{\beta} }{\NoSel{c}, \InitAssi{\alpha}} \in [0,1], \\
			C_4 \coloneqq C_4(\alpha,c,\gamma)
			&\coloneqq 
			\probcondResize{ \Wa{c}{\gamma} }{\bigcap_{j=1}^{c} \singleevent{ \Sel{j} = 0 }, \InitAssi{\alpha}} \in [0,1], %
		\end{align*}
		and the random variables $\ZV{P}$, $\ZV{Q}$, and $\ZV{R}$ as in Proposition~\ref{recap:FiniteCaseDone}.
	\end{recapitulation}

	\subsection{Allowing Restarts}
	\label{app:AllowingRestarts}
	At the beginning of \refsec{sec:analysis_rtd}, we mentioned that our arguments implicitly assume that \SRWA{} does not employ any restarts. However, for the most part, our arguments also apply to the restarted version of the algorithm. The only technical challenge is adapting \refEq{eq:FallBackToERBetaPrime} in the appendix:
	\begin{align}
		\label{eq:FallBackToERBetaPrimeAgain}
		\EACond{\NoSel{c}, \Wa{c}{\beta}, \Flip{c}{\beta}{x} }
		= c + 1 + \ECond{\InitAssi{\beta^\prime}}.
	\end{align}
	
	Here, the difficulty is that after some flips, a restart is due, and thus \SRWA{} does not continue its search in $\beta^\prime \coloneqq \beta[x]$. Instead, a new random assignment is chosen.
	
	At its core, the reasoning behind \refEq{eq:FallBackToERBetaPrimeAgain} is still valid. The left-hand side of \refEq{eq:FallBackToERBetaPrimeAgain} uses the condition $\NoSel{c}$, which tells us that the algorithm performed $c+1$ flips without finding a satisfying assignment. The only difference is that \SRWA{} might have performed one (or several) restarts in the meantime. This can be handled by introducing an additional counter~$T$ in the expectation $\ERL$. The counter~$T$ is used to count the number of remaining flips until a restart is performed. Incorporating this counter in the argument readily yields the following equations:
	\begin{align*}
		\EACond{T=0} = \sum_{\beta \in \{0,1\}^n} \Prob{\InitAssi{\beta}} \cdot \ECond{T=t_{\textnormal{restart}}, \InitAssi{\beta}},
	\end{align*}
	where $	\Prob{\InitAssi{\beta}} = \frac{1}{2^n}$; and
	\begin{multline*}
		\EACond{\NoSel{c}, \Wa{c}{\beta}, \Flip{c}{\beta}{x}, T=d} \\
		= \:  c + 1 + \ERL \Big( T=\big( d-(c+1) \big) \bmod{t_{\textnormal{restart}}}, \InitAssi{\beta^\prime} \Big).
	\end{multline*}
	The rest of the analysis remains unaltered besides the need to add the counter~$T$ in each part of the proofs.

\section{Connections Between Different Distributions}

\label{app:some_proofs}

\subsection{Embedded Models: Johnson SB Approaches Lognormal}
\label{sec:MoreOnEmbedding}

In this section, we provide some details about the embedding property of lognormal and Johnson SB distributions.

\begin{lemma}[\cite{cheng2017non,norman1994continuous}]
	The lognormal distribution is an embedded model of the SB distribution.
\end{lemma}

\begin{proof}[Proof (adapted from~\cite{cheng2017non})]
	Reparametrizing with $b=\alpha^{-1}$ and $\gamma = \delta \bigl(\ln{(\alpha^{-1}-a)} - \mu\bigr)$ yields	
	\begin{align*}
		f(x)=
		\frac{1}{\sqrt{2 \pi}} \frac{(\frac{1}{\alpha}-a)\delta}{(x-a)(\frac{1}{\alpha}-x)} \exp\biggl\{-\frac{1}{2}\biggl[ \delta \left(\ln{\left(\frac{1}{\alpha}- a\right)}-\mu\right) + \delta \ln{\frac{x-a}{\frac{1}{\alpha}-x}} \biggr]^2 \biggr\}.
	\end{align*}	
	The log-likelihood function $\mathcal{L}$ is then given by:	
	\begin{align*}
		\mathcal{L}(a,\delta, \alpha, \mu)=
		-\frac{1}{2}\ln{  2 \pi } +  \ln{ \delta } + \ln{  \frac{\frac{1}{\alpha}-a}{\frac{1}{\alpha}-x}  } -  \ln{(  x-a  )} %
		- \frac{1}{2} \biggl[ \delta \left(\ln{\frac{\frac{1}{\alpha} - a}{\frac{1}{\alpha}-x}}-\mu\right) + \delta \ln{(x-a)} \biggr]^2.
	\end{align*}
	Thus,	
	\begin{align*}
		\ell(a,\delta, \alpha, \mu) \coloneqq \diff{\alpha} \mathcal{L}(a,\delta, \alpha, \mu)=
		-\dfrac{\left(x-a\right)\left(\delta^2\ln\frac{\frac{1}{\alpha}-a}{\frac{1}{\alpha}-x}+\delta^2\ln\left(x-a\right)-\delta^2\mu-1\right)}{\left(a\alpha-1\right)\left(x\alpha-1\right)}.
	\end{align*}	
	Furthermore, as 
	\begin{align*}
		\lim_{\alpha \rightarrow 0} \ln{  \frac{\frac{1}{\alpha}-a}{\frac{1}{\alpha}-x}  } = 0
	\end{align*}
	holds, we have 
	\begin{align*}
		& \lim_{\alpha \rightarrow 0}\, \mathcal{L}(a,\delta, \alpha, \mu)=
		-\frac{1}{2}\ln{  2 \pi } +  \ln{ \delta } -  \ln{(  x-a  )} - \frac{1}{2} 
		\bigl[ \delta \ln{(x-a)} - \delta \mu \bigr]^2
		~ \textnormal{and}
		\\
		& \lim_{\alpha \rightarrow 0}\, \ell(a,\delta, \alpha, \mu) =
		-\left(x-a\right)\left(\delta^2\ln\left(x-a\right)-\delta^2\mu-1\right).
	\end{align*}	
	We represent $\mathcal{L}$ as a Taylor series of $\alpha$ in $\alpha=0$. Observing this series expansion for $\alpha \rightarrow 0$ yields:
	\begin{align*}
		\lim_{\alpha \rightarrow 0}\, \mathcal{L}(a,\delta, \alpha, \mu) 
		&= \lim_{\alpha \rightarrow 0}\, \left( \mathcal{L}(a,\delta, 0, \mu) 
		+ \ell(a,\delta, 0, \mu)\cdot \alpha 
		+ \bigoh{\alpha^2} \right) \\
		&= -\frac{1}{2} \ln{2\pi} 
		+ \ln{ \delta } 
		- \ln{(x-a)} 
		- \frac{1}{2} \delta^2 \bigl[ \ln{(x-a)} - \mu \bigr]^2 \\
		&\phantom{==}+ \left(x-a\right) \left[ 1+\delta^2 \left( \mu - \ln \left(x-a\right) \right) \right] \lim_{\alpha \rightarrow 0}\, \alpha
		+ \lim_{\alpha \rightarrow 0}\, \bigoh{\alpha^2}\\
		&= -\frac{1}{2}\ln{  2 \pi } +  \ln{ \delta } -  \ln{(  x-a  )} - \frac{1}{2}     \delta^2 \bigl[ \ln{(x-a)} - \mu \bigr]^2.
	\end{align*}
	This is precisely the log-likelihood function of a lognormal distribution. Thus, the Johnson SB distribution approaches a lognormal distribution for $\alpha \rightarrow 0$ (or $b\rightarrow \infty$ as well as $\gamma\rightarrow \infty$ in the original parameterization\footnote{
		There is a slight typo in the quoted source~\cite{cheng2017non}, stating that~$b$ should approach zero. However, following the presented argument in~\cite{cheng2017non}, it is clear that $b$ has to approach infinity since $\alpha$ approaches zero and $b= \frac{1}{\alpha}$. 
	}).
\end{proof}

\subsection{A Proof of Lemma~\ref{lem:1/1+Y_for_Y_LogN_ist_SB_verteilt}}
\label{sec:FirstAppProof}

In this section, we will prove one of our main lemmas. This lemma will be used in the distribution analysis of the random variables $\ZV{P}$, $\ZV{Q}$, and $\ZV{R}$.

\begin{lemma}[\refLem{lem:1/1+Y_for_Y_LogN_ist_SB_verteilt} restated]
	Let $\ZV{Y} \distr \LogN{\mu}{\sigma^2}$, then we have
	\[
	\frac{1}{c+\ZV{Y}} \distr  \SBwithequations{\gamma = \frac{\mu-\ln c}{\sigma}}{\delta = \frac{1}{\sigma}}{\lambda = \frac{1}{c}}{\xi = 0}
	\]
	for all $c \in \Rplus$.
\end{lemma}

\begin{proof}[Proof of \refLem{lem:1/1+Y_for_Y_LogN_ist_SB_verteilt} (adapted from~\cite{MOSBanswer})]
	Let $\ZV{W} \distr \Normal{\mu}{\sigma^2}$. Then $-\ZV{W} \distr \Normal{-\mu}{\sigma^2}$.
	By definition, $\ZV{Y} \coloneqq \e^{\ZV{W}} \distr \LogN{\mu}{\sigma^2}$. 
	Define the random variable
	\[
	\ZV{X} \coloneqq \frac{1}{c+\ZV{Y}} = \frac{1}{c+\e^{\ZV{W}}} = \frac{1}{c+c\cdot \e^{\ZV{W}-\ln c}}= \frac{1}{c+c\cdot \e^{\ZV{W^\prime}}} = \frac{\e^{-\ZV{W^\prime}}}{c \e^{-\ZV{W^\prime}} + c },
	\]
	where
	\[
	\ZV{W^\prime} \distr \Normal{\mu - \ln c}{\sigma^2}.
	\]
	A few simple rearrangements yield
	\[
	-\ZV{W^\prime} = \log \left( \frac{c\ZV{X}}{1-c\ZV{X}} \right)  = \log \left( \frac{\ZV{X}}{\frac{1}{c} - \ZV{X}} \right) .
	\]
	Letting $\lambda \coloneqq 1/c$ and $\xi \coloneqq 0$, we have
	\[
	-\ZV{W} = \log \left( \frac{\ZV{X}-\xi}{\xi+\lambda-\ZV{X}} \right).
	\]
	Define $\delta \coloneqq \frac{1}{\sigma}$ and $\gamma \coloneqq \frac{\mu - \ln c}{\sigma}$ and
	let $\ZV{Z}$ be a random variable such that
	\[
	\ZV{Z} = -\ZV{W^\prime} \cdot \delta + \gamma.
	\]
	Since 
	\[
	-\ZV{W^\prime}  \distr \Normal{\ln c-\mu}{\sigma^2},
	\] 
	we have
	\[
	- \ZV{W^\prime} \cdot \delta \distr \Normal{\frac{\ln c- \mu}{\sigma}}{1}
	\]
	and 
	\[
	- \ZV{W^\prime} \cdot \delta + \gamma \distr \Normal{0}{1}.
	\]
	Per definition, $\ZV{X} \distr \SBwithequations{\gamma = \frac{\mu-\ln c}{\sigma}}{\delta = \frac{1}{\sigma}}{\lambda = \frac{1}{c}}{\xi = 0}$.
\end{proof}

\subsection{A Proof of Lemma~\ref{lem:log_binomial_normal}}
\label{sec:SecondAppendixProof}

Our aim in this section will be to provide a proof of Lemma~\ref{lem:log_binomial_normal}. We have restated the lemma below for convenience.

\begin{lemma}[Lemma~\ref{lem:log_binomial_normal} restated, \cite{Lachin2011Biostatistical}]
	If $\ZV{Y} \distr \Bin{n}{p}$, then $\log \ZV{Y}$ is asymptotically normally distributed.
\end{lemma}

For the proof of the lemma, we will need to introduce some technical machinery.

\begin{definition}[Order in probability]
	For a sequence of random variables $(\ZV{X_n})_{n \in \N}$ 
	and a corresponding sequence of constants $(a_n)_{n \in \N}$,
	we write $\ZV{X_n} \in \bigohp{a_n}$
	if for all $\varepsilon > 0$, there exists an $M > 0$ and an $N>0$ such that 
	\[
	\PROB{ \left| \frac{\ZV{X_n}}{a_n} \right| > M } < \varepsilon \quad \text{for all} \quad n > N.
	\]
\end{definition}

\begin{definition}
	An estimator $\ZV{\hat{\theta}_n}$ is a \introduceterm{consistent estimator} for the parameter $\theta$ if $\ZV{\hat{\theta}} \inprobability{n \to \infty} \theta$.
	An estimator $\ZV{\hat{\theta}_n}$ is a \introduceterm{$\sqrt{n}$-consistent estimator} for the parameter $\theta$ if $\ZV{\hat{\theta}_n} - \theta = \bigohp{1/\sqrt{n}}$.
\end{definition}

\begin{example}
	\label{ex:SqrtNConsistentEstimator}
	From the central limit theorem, it follows 
	that
	for all $\varepsilon > 0$ it holds
	\[
	\Prob{|\ZV{\overline{X}_n} - \mu | > \varepsilon}
	= \PROB{ \frac{\sqrt{n} | \ZV{X_n} - \mu |}{\sigma} > \frac{\sqrt{n} \varepsilon}{\sigma} }
	= 2 \left( 1 - \Phi \left( \frac{\sqrt{n} \varepsilon}{\sigma} \right) \right)
	\stackrel{n \to \infty}{\longrightarrow} 0,
	\]
	\ie $\ZV{\overline{X}_n} \inprobability{n \to \infty} \mu$. Thus, $\ZV{\overline{X}_n}$ is a consistent estimator for $\mu$. Because $\sqrt{n} \big( \ZV{\overline{X}_n} - \mu \big) \indistribution{n \to \infty} \Normal{0}{\sigma^2}$, we can also see that $\ZV{\overline{X}_n}$ is a $\sqrt{n}$-consistent estimator of $\mu$.
\end{example}

\begin{theorem}[Slutsky's Theorem]
	Let $(\ZV{X_n})_{n \in \N}, (\ZV{A_n})_{n \in \N}, (\ZV{B_n})_{n \in \N}$ be sequences of random variables. If $\ZV{X_n}$ converges to a random variable $\ZV{X}$ in distribution, $\ZV{X_n} \indistribution{n \to \infty} \ZV{X}$, and $\ZV{A_n} \inprobability{n \to \infty} a$, as well as $\ZV{B_n} \inprobability{n \to \infty} b$, then 
	\[
	\ZV{A_n} + \ZV{B_n} \ZV{X_n} \indistribution{n \to \infty} a + b \ZV{X}.
	\]
\end{theorem}

With these definitions in place, we can now prove \refLem{lem:log_binomial_normal}.

\begin{proof}[Proof of \refLem{lem:log_binomial_normal}~(adaped from~\cite{Lachin2011Biostatistical})]
	Let $n \in \N^{+}$ and $p \in (0,1)$.
	For $i \in \set{1, \dots, n}$ let $\ZV{X_i} \distr \Bern{p}$.
	Let $(x_1, \dots, x_n)$ denote a corresponding sample of these random variables.
	Then, $\ZV{S_n} \coloneqq \ZV{X_1} + \dots + \ZV{X_n} \distr \Bin{n}{p}$.
	Since
	\[
	\expectation{\ZV{\overline{X}_n}}
	= \EXPECTATION{ \frac{1}{n} \sum_{i=1}^n \ZV{X_i} }
	= \frac{1}{n} \EXPECTATION{ \sum_{i=1}^n \ZV{X_i} }
	= \frac{1}{n} \sum_{i=1}^n \expectation{\ZV{X_i}}
	= \frac{1}{n} n p
	= p,
	\]
	a natural moment estimator of $p$ is the \introduceterm{sample mean} $\ZV{\overline{X}_n} \coloneqq \ZV{S_n}/n$.
	In the following, we let $\overline{x}_n \coloneqq \frac{1}{n} (x_1 + \dots + x_n)$ denote the concrete sample mean.
	Applying Taylor's expansion to the mapping $x \mapsto \log x$, we obtain
	\[
	\log \overline{x}_n = \log p + \frac{\dd \log p}{\dd p} (\overline{x}_n - p) + R_2(\xi),
	\]
	where
	\[
	R_2(\xi) = \frac{\dd^2 \log \xi}{\dd \xi^2} (\overline{x}_n - p)^2
	\]
	is the Lagrange remainder term and $\xi$ lies between $\overline{x}_n$ and $p$. Since $\log^{\prime}(x) = \frac{1}{x}$ and $\log^{\prime\prime}(x) = - \frac{1}{x^2}$ for all $x>0$, this yields to
	\[
	\sqrt{n} \big( \log \overline{x}_n - \log p \big)
	= \sqrt{n} \frac{\overline{x}_n - p}{p} - \sqrt{n} \frac{(\overline{x}_n - p)^2}{2 \xi^2}.
	\]

	It is well known
	that $\ZV{\overline{X}_n}$ is asymptotically normal distributed (and the same holds for the concrete realizations $\overline{x}_n$), more precisely
	\[
	\overline{x}_n \approxd \Normal{p}{ \frac{p(1-p)}{n} }.
	\]
	Thus, the first term on the right-hand side, $\sqrt{n} \frac{\overline{x}_n - p}{p}$, is likewise asymptotically normal distributed,
	more precisely 
	\[
	\sqrt{n} \frac{\overline{x}_n - p}{p} \approxd \Normal{0}{ \frac{1-p}{p} }.
	\]

	In Example~\ref{ex:SqrtNConsistentEstimator}, we saw that $\overline{x}_n$ is a $\sqrt{n}$-consistent estimator of $p$. Hence, $(\overline{x}_n - p)^2 \to 0$ as~$n \to \infty$ faster than $n^{-1/2}$. Thus,
	\[
	\sqrt{n} R_2(\xi) \inprobability{n \to \infty} 0.
	\]

	Therefore, asymptotically $\sqrt{n} \big( \log \overline{x}_n - \log p \big)$ is the sum of two random variables, the first one converging in distribution to the normal, the second one converging in probability to zero (\ie a constant). From Slutsky's Theorem, it follows that
	\[
	\sqrt{n} \big( \log \overline{x}_n - \log p \big) \indistribution{n \to \infty} \Normal{0}{ \frac{1-p}{p} }.
	\]
	Hence,
	\begin{align*}
		\log \overline{x}_n \approxd \Normal{ \log p }{ \frac{1-p}{np} }. %
		\label{eq:bin_approx_logn}
	\end{align*}

	In other words, $\overline{x}_n$ is asymptotically lognormal distributed.
	Since $\overline{x}_n$ is binomially distributed, the claim follows.
\end{proof}

\end{document}